\documentclass[runningheads,a4paper]{llncs}

\usepackage{setspace}

\usepackage{amsmath,amssymb,graphicx,subfigure}
\usepackage{xcolor}
\usepackage[pdfborder={0 0 0}]{hyperref}
\usepackage{etoolbox}
 
\usepackage{mathabx}
\usepackage[ruled,vlined]{algorithm2e}

\newtoggle{fullver}
\toggletrue{fullver}

\title{Separation Logic  with Linearly Compositional Inductive Predicates and Set Data Constraints\thanks{This work is partially supported by  NSFC grant (No.\ 61472474, 61572478, 61872340), UK EPSRC grant (EP/P00430X/1), and the INRIA-CAS joint research project VIP.}}
\author{Chong Gao\inst{1,2}, Taolue Chen\inst{3},  Zhilin Wu\inst{1}}
\institute{
	State Key Laboratory of Computer Science, \\
	Institute of Software, Chinese Academy of Sciences, Beijing, China
	\and
	University of Chinese Academy of Sciences, Beijing, China 
	\and
	Department of Computer Science and Information Systems, \\
	Birkbeck, University of London, UK 
	}

\newcommand{\Aa}{\mathcal{A}}
\newcommand{\Bb}{\mathcal{B}}

\newcommand{\Dd}{\mathcal{D}}

\newcommand{\Ff}{\mathcal{F}}

\newcommand{\Gg}{\mathcal{G}}
\newcommand{\Ll}{\mathcal{L}}

\newcommand{\op}{\bowtie}
\newcommand{\opset}{\asymp}

\newcommand{\natnum}{\mathbb{N}}

\newcommand{\intnum}{\mathbb{Z}}

\newcommand{\intset}{\mathbb{S}_\mathbb{Z}}
\newcommand{\natset}{\mathbb{S}_\mathbb{N}}

\newcommand{\mapzn}{\mathcal{M}}

\newcommand{\sep}{\ast}

\newcommand{\flds}{{\rm Flds}}

\newcommand{\lcslidset}{{\sf SLID^{S}_{LC} }}

\newcommand{\loc}{{\mathbb L}}
\newcommand{\data}{{\mathbb D}}

\newcommand{\ufld}{{\sf Ufld}}

\newcommand{\vars}{\mathsf{Vars}}
\newcommand{\lvars}{\mathsf{LVars}}
\newcommand{\dvars}{\mathsf{DVars}}
\newcommand{\svars}{\mathsf{SVars}}

\newcommand{\bvars}{\mathsf{BVars}}

\newcommand{\ldom}{{\sf ldom}}

\newcommand{\boolabs}{{\sf Abs}}

\newcommand{\dt}{{\sf dt}}

\newcommand{\Red}{{\sf Red}}

\newcommand{\slemp}{\mathtt{emp}}
\newcommand{\ltrue}{\mathtt{true}}
\newcommand{\lfalse}{\mathtt{false}}

\newcommand{\plseg}{\mathit{plseg}}

\newcommand{\ldllseg}{\mathit{ldllseg}}

\newcommand{\sdllseg}{\mathit{sdllseg}}
\newcommand{\stlseg}{\mathit{stlseg}}

\newcommand{\fnext}{\mathtt{next}}
\newcommand{\fprev}{\mathtt{prev}}

\newcommand{\fdata}{\mathtt{data}}

\newcommand{\freev}{\mathsf{free}}

\newcommand{\ectx}{\mathsf{ctx}}

\newcommand{\saturate}{\mathsf{Strt}}

\newcommand{\norm}{\mathsf{Norm}}

\newcommand{\quantel}{{\sf quantElmt}}

\newcommand{\hide}[1]{ }

\newcommand{\dbs}{\mathcal{DBS}}

\newcommand{\ps}{\mathcal{PS}}

\newcommand{\qgdbs}{\mathcal{RQSPA}}

\newcommand{\edbs}{\mathcal{EDBS}}

\newcommand{\TC}{\textsf{TC}}

\newcommand{\suc}{\textsf{succ}}

\newcommand{\idx}{{\sf idx}}

\begin{document}

\maketitle



\begin{abstract} 
 We identify difference-bound set constraints (DBS), an analogy of difference-bound arithmetic constraints for sets. DBS can express not only set constraints but also arithmetic constraints over set elements. We integrate DBS into separation logic with linearly compositional inductive predicates, obtaining a logic thereof where set data constraints of linear data structures can be specified. We show that the satisfiability of this logic is decidable. A crucial step of the decision procedure is to compute the transitive closure of DBS-definable set relations, to capture which we propose an  extension of quantified set constraints with Presburger Arithmetic (RQSPA). The satisfiability of RQSPA is then shown to be decidable by harnessing advanced automata-theoretic techniques. 
\end{abstract}

\section{Introduction} \label{sec:intro}


%

Separation Logic (SL) is a well-established approach for deductive verification of programs that manipulate dynamic data structures \cite{ORY01,Rey02}. Typically, SL is used in combination with inductive definitions (SLID), which provides a natural and convenient means to specify dynamic data structures.
To reason about the property (e.g. sortedness) of  data values stored in data structures, it is also necessary to incorporate data constraints into the inductive definitions.


%
%

One of the most fundamental questions for a logical theory is whether its satisfiability is decidable. SLID with data constraints is no exception. This problem becomes more challenging than one would probably expect, partially due to the inherent intricacy brought up by inductive definitions and data constraints. It is somewhat surprising that only disproportional research has addressed this question (cf. \emph{Related work}).
In practice, most available tools based on SLID only support heuristics without giving completeness guarantees, especially when data constraints are involved. \emph{Complete} decision procedures for satisfiability, however, have been found important in  software engineering tasks such as symbolic execution, specification debugging, counterexample generation, etc., let along the theoretical insights they usually shed on the logic system.



The dearth of complete decision procedures for SLID with data constraints has prompted us to launch a research program as of 2015, aiming to identify decidable \emph{and} sufficiently expressive instances. We have made encouraging progress insofar. In \cite{GuCW16}, we set up a general framework, but could only tackle linear data structures with data constraints in difference-bound arithmetic. In \cite{XuCW17}, we were able to tackle tree data structures by exploiting machineries such as order graphs and counter machines, though the data constraints therein remained to be in difference-bound arithmetic. 


An important class of data constraints that is currently elusive in our investigations is set constraints. They are mandatory for reasoning about, e.g., invariants of data collections stored in data structures. For instance, when specifying the correctness of a sorting algorithm on input lists, whilst the sortedness of the list can be described by difference-bound arithmetic constraints, the property that the sorting algorithm does not change the set of data values on the list requires inductive definitions with \emph{set} data constraints. Indeed, reviewers of the papers \cite{GuCW16,XuCW17} 
constantly raised the challenge of set constraints, which compelled us to write the current paper.  

\smallskip
\noindent \emph{Main contributions.} 
%
%
Our \emph{first} contribution is to carefully design the difference-bound set constraints ($\dbs$), and to integrate them into the linearly compositional inductive predicates introduced in \cite{GuCW16}, yielding $\lcslidset$: \emph{SL with linearly compositional inductive predicates and set data constraints}. The rationale of $\dbs$ is two-fold: (1) it must be sufficiently expressive to represent common set data constraints as well as arithmetic constraints over set elements one usually needs when specifying linear data structures, (2) because of the inductive predicates, it must be sufficiently ``simple" to be able to capture the \emph{transitive closure} of $\dbs$-definable set relations\footnote{This shall be usually referred to as ``transitive closure of $\dbs$" to avoid clumsiness.}  in an effective means, in order to render the satisfiability of $\lcslidset$ decidable. As the \emph{second} contribution, we show that the transitive closure of $\dbs$ can indeed be captured in the \emph{restricted extension of quantified set constraints with Presburger arithmetic} ($\qgdbs$) introduced in this paper. Finally, our \emph{third} contribution is to show that the satisfiability of $\qgdbs$  is decidable by establishing a connection of $\qgdbs$ with Presburger automata \cite{SSM08}. This extends the well-known connection of Monadic Second-Order logic on words (MSOW) and finite-state automata \emph{a la} B\"{u}chi and Elgot \cite{Bu60,Elg61}. These contributions, together with a procedure which constructs an abstraction (as an $\qgdbs$ formula) from a given $\lcslidset$ formula and which we adapt from our previous work \cite{GuCW16}, show the satisfiability of $\lcslidset$ is decidable.


We remark that sets are conceptually related to  second---rather than first--- order logics. While the transitive closure of logic formulae with \emph{first-order variables}  is somehow well-studied (especially for simple arithmetic; cf.\ \emph{Related Work}), the transitive closure of logic formulae with \emph{second-order variables} is rarely addressed in literature. (They easily lead to undecidability.) To  our best knowledge, the computation of transitive closures of $\dbs$ here 
represents one of the first practically relevant examples of the computation of this type 
for a class of logic formulae with second-order variables, which may be of independent interests. 
%

\smallskip
\noindent \emph{Related work.} 
We first review the work on SLID with data constraints. (Due to space limit, the work on SLID \emph{without} data constraints will be skipped.)
%
In \cite{CDN+12,CJT15,MQS12}, SLID with set/multiset/size data constraints were considered, but only (incomplete) heuristics were provided. To reason about invariants of data values stored in lists,  SL with list segment predicates and data constraints in universally quantified Presburger arithmetic  was considered \cite{BDES12a}. 
The work \cite{PWZ13,PWZ14}  provided decision procedures for SLID with data constraints by translating into many-sorted first-order logic with reachability predicates. In particular, in \cite[Section 6]{PWZ14}, extensions of basic logic GRIT are given to cover set data constraints as well as order constraints over set elements. However, it seems that this approach does not address arithmetic constraints over set elements (cf. the ``Limitations'' paragraph in the end of Section 6 in \cite{PWZ14}). For instance, 
a list where the data values in adjacent positions are consecutive can be captured in $\lcslidset$ (see the predicate $\plseg$ in Section~\ref{sec-prelm}), but appears to go beyond the work \cite{PWZ13,PWZ14}. Moreover, there is no precise characterisation of the limit of extensions under which the decidability retains. 
%
%
 %
The work \cite{ESW15} introduced the concept of compositional inductive predicates, which may alleviate the difficulties of the entailment problem for SLID. 
Nevertheless, \cite{ESW15} only provided sound heuristics rather than decision procedures. More recently, the work \cite{LSC16,TLC16} investigated SLID with Presburger arithmetic data constraints.

Furthermore, several logics other than separation logic have been considered to reason about both shape properties and data constraints of data structures. The work \cite{SDK10} proposed a generic decision procedure for recursive algebraic data types with abstraction functions encompassing lengths (sizes) of data structures, sets or multisets of data values as special cases. Nevertheless, the work \cite{SDK10} 
focused on functional programs while 
this work aims to verify imperative programs, which requires to reason about \emph{partial} data structures such as list segments (rather than complete data structures such as lists). It is unclear how the decision procedure in \cite{SDK10} can be generalised to partial data structures.
The work \cite{MPQ11} introduced STRAND, 
a fragment of monadic
second-order logic, to reason about
tree structures. Being undecidable in general, several decidable
fragments were identified. STRAND does not provide an explicit means to describe sets of data values, although it allows using set variables to represent sets of locations.
%



Our work is also related to classical logics with set constraints, for which we can only give a brief (but by no means comprehensive) summary. 
Presburger arithmetic extended with sets was studied dating back to 80's, with highly undecidability results \cite{CCS90,Haplern91}. 
However, decidable fragments do exist: \cite{WPK09} studied the non-disjoint combination of theories that share set variables and set operations. 
\cite{KPS10} considered 
QFBAPA$^<_\infty$, a quantifier-free logic of sets of real numbers supporting integer sets and variables, linear arithmetic, the cardinality operator, infimum and supremum. \cite{Voigt17,HorbachVW17} investigated two extensions of the Bernays-Sch{\"{o}}nfinkel-Ramsey fragment of first-order predicate logic (BSR) 
 with simple linear arithmetic over integers and difference-bound constraints over reals (but crucially, 
 the ranges of the universally quantified variables must be bounded). Since the unary predicate symbols in BSR are uninterpreted and represent sets over integers or reals, the two extensions of BSR can also be used to specify the set constraints on integers or reals. \cite{CR17} presented a decision procedure for quantifier-free constraints on restricted intensional sets (i.e., sets given by a property rather than by enumerating their elements). None of these logics are able to capture the transitive closure of $\dbs$ as $\qgdbs$ does. MSOW extended with linear cardinality constraints was investigated in \cite{KR03}. Roughly speaking, $\qgdbs$ can be considered as an extension of MSOW with linear arithmetic expressions on the maximum or minimum value of free set variables.
Therefore, the two extensions in \cite{KR03} and 
this paper are largely incomparable.




In contrast to set constraints, the computation of transitive closures of relations definable in  first-order logic (in particular, difference-bound and octagonal arithmetic constraints) has been considered in 
for instance, \cite{ComonJ98,BozgaIL09,BozgaGI09,BozgaIK10,Kon16}. 

\section{Logics for sets} \label{sec:logic}


We write $\intnum$, $\natnum$ for the set of integers and natural numbers;  $\intset$ and $\natset$ for \emph{finite} subsets of $\intnum$ and $\natnum$. For $n\in \natnum$, $[n]$ stands for $\{1, \cdots, n\}$. We shall work exclusively on finite subsets of $\intnum$ or $\natnum$ unless otherwise stated. For any finite $A\neq\emptyset$, we write $\min(A)$ and $\max(A)$ for the minimum and maximum element of $A$. These functions, however, are \emph{not} defined over empty sets.  

In the sequel, we introduce a handful of logics for sets which will be used later in this paper. We mainly consider two data types, i.e., integer type $\intnum$ and (finite) set 
type $\intset$. Typically, $c,c',\dots \in \intnum$ and $A, A',\dots \in \intset$. 
Accordingly, two types of variables occur: integer variables (ranged over by $x, y, \cdots$) and set variables (ranged over by $S, S', \cdots$). 
Furthermore, we reserve $\op\ \in \{=, \le, \ge\}$ for comparison operators between integers,\footnote{The operators $<$ and $>$ can be seen as abbreviations, for instance, $x < y$ is equivalent to $x \le y -1$, which will be used later on as well.} and $\opset\ \in \{=, \subseteq, \supseteq, \subset, \supset\}$ for comparison operators between sets.
We start with \emph{difference-bound} set constraints ($\dbs$). 

\begin{definition}[Difference-bound set constraints] Formulae of $\dbs$ are defined by the rules:
\vspace{-1mm}
\[
\begin{array}{l c l c r}
	\varphi &::= & S =  S' \cup T_s \mid  T_i \ \op\ T_i + c \mid \varphi \wedge \varphi & & \\
	T_s & ::= &  \emptyset \mid \{\min(S)\} \mid \{\max(S)\}\mid T_s \cup T_s & \qquad  &  \mathrm{(\textbf{s}et\ terms)}\\
\vspace{-3mm}
	T_i & ::= &  \min(S) \mid \max(S) & \qquad & \mathrm{(\textbf{i}nteger\ terms)}
\end{array}
\]
\end{definition}

\medskip
 
\noindent \emph{Remark.} $\dbs$ is a rather limited logic, but it has been carefully devised to serve the data formulae in inductive predicates of $\lcslidset[P]$ (cf.\ Section~\ref{sec-prelm}). In particular, we remark that only conjunction, but not disjunction, of atomic constraints is allowed. The main reason is, once the disjunction is introduced, the computation of transitive closures becomes infeasible simply because one would be able to encode the computation of Minsky's two-counter machines.\qed 
\smallskip

To capture the transitive closure of $\dbs$, we introduce \emph{Restricted extension of Quantified Set constraints with Presburger Arithmetic}\footnote{An unrestricted extension of quantified set constraints with Presburger Arithmetic is undecidable, as shown in \cite{CCS90}.} ($\qgdbs$). 
Intuitively, an $\qgdbs$ formula is a \emph{quantified} set constraint extended with Presburger Arithmetic satisfying the following restriction: 
each atomic formula containing quantified variables must be a difference-bound arithmetic constraint.

\vspace{-1mm}
\begin{definition}[Restricted extension of Quantified Set constraints with Presburger Arithmetic] \label{def:qgdbs} Formulae of $\qgdbs$ are defined by the rules:
\vspace{-2mm}
	\begin{align*}
	\Phi & ::= T_s\ \opset\ T_s \mid   T_i\ \op\ T_i + c \mid T_m\ \op\ 0 \mid  \Phi \wedge \Phi \mid  \neg \Phi \mid \forall x.\ \Phi \mid \forall S.\ \Phi,\\
		T_s &::=   \emptyset   \mid S \mid \{T_i\}  \mid T_s \cup T_s \mid T_s \cap T_s \mid T_s \setminus T_s,  \\
	T_i &::= c \mid x \mid \min(T_s) \mid \max(T_s), \\
	T_m &::=   c \mid x \mid \max(T_s) \mid \min(T_s) \mid T_m+T_m\mid T_m-T_m.
	\end{align*}
\vspace*{-8mm}
	%
	
\end{definition}
Here, $T_s$ (resp. $T_i$) represents set (resp. integer) terms which are more general than those in $\dbs$, and $T_m$ terms are \emph{Presburger arithmetic expressions}. Let $\vars(\Phi)$ (resp. $\freev(\Phi)$) denote the set of variables (resp. free variables) occurring in $\Phi$.  
%
We require that {\bf 
all set variables in atomic formulae $T_m\ \op\ 0$ are free}.
%
To make the free variables explicit, we usually write $\Phi(\vec{x}, \vec{S})$ for a $\qgdbs$ formula $\Phi$. 
Free variable names are assumed \emph{not} to clash with the quantified ones. 

\vspace{-1mm}
\begin{example}
	$\max(S_1\cup S_2) - \min(S_1) - \max(S_2) < 0$ and $\forall S_1\forall S_2. (S_2 \neq \emptyset \rightarrow \max(S_2) \le \max(S_1\cup S_2))$ are $\qgdbs$ formulae, while $\forall S_2.\ \max(S_1\cup S_2) - \min(S_1) - \max(S_2) < 0$ is \emph{not}. \qed
\end{example}
\vspace{-1mm}

The work \cite{CCS90}, among others, studied \emph{Presburger arithmetic extended with Sets} ($\ps$), which is 
\emph{quantifier-free} $\qgdbs$ formulae.  
In this paper, $\ps$ will serve the data formula part of  $\lcslidset[P]$, and we reserve $\Delta, \Delta',\ldots$ to denote formulae from $\ps$ (see Section~\ref{sec-prelm}).  
%
\hide{
\begin{definition}[$\ps{[}\qgdbs${]}] The formulae of $\qgdbs$ are defined by the following rules,
	\[
	\begin{array}{r c l c} 
	\Delta & ::= & \ltrue \mid \Phi \mid T_i \ \op\  T_i  \mid T_s\ \opset\ T_s \mid T_i\in T_s \mid \Delta \wedge \Delta \mid \neg \Delta  \\
	T_s & ::=& \emptyset \mid S \mid \{T_i\} \mid T_s \cup T_s \mid T_s \cap T_s \mid T_s \setminus T_s &   \\ 
	T_i & ::= & c \mid x \mid \max(T_s) \mid \min(T_s) \mid T_i+T_i\mid T_i-T_i &  \\
	\end{array}
	\]
	where $\Phi$ is a $\qgdbs$ formula defined in Definition~\ref{def:qgdbs}.
\end{definition}
}
%
%
%

%
%
%
%

\smallskip
\noindent\emph{Semantics.} All of these logics ($\dbs$, $\qgdbs$, $\ps$) can be considered as instances of weak monadic second-order logic, and thus their semantics are largely self-explanatory. In particular, set variables are interpreted as \emph{finite} subsets of $\intnum$ and integer variables are interpreted as integers. We emphasize that, if a set term $T_s$ is interpreted as $\emptyset$, $\min(T_s)$ and $\max(T_s)$ are undefined. As a result, we stipulate that \textbf{any  atomic formula containing an undefined term is interpreted as $\lfalse$.}

For an $\qgdbs$ formula $\Phi(\vec{x}, \vec{S})$ with $\vec{x} = (x_1,\cdots, x_k)$ and $\vec{S} = (S_1, \cdots, S_l)$, $\Ll(\Phi(\vec{x}, \vec{S}))$ denotes

\medskip
$
\hspace{4mm} \{(n_1,\cdots, n_k, A_1,\cdots, A_l) \in \intnum^k \times \intset^l \mid \Phi(n_1,\cdots, n_k, A_1,\cdots, A_l) \}
$.
\smallskip

\noindent As expected, typically we use $\dbs$ formulae to define relations between (tuples) of sets  from $\intset^k$. We say a relation $R \subseteq  \intset^{k} \times \intset^k$ a \emph{difference-bound set relation} if there is a $\dbs$ formula $\varphi(\vec{S}, \vec{S'})$ over set variables $\vec{S}$ and $\vec{S}'$ such that $R= \{(\vec{A}, \vec{A'}) \in \intset^k \times \intset^k \mid \varphi(\vec{A}, \vec{A'}) \}$.  
The \emph{transitive closure} of $R$ is defined in a standard way, viz., $\bigcup \limits_{i \ge 0} R^i$, where $R^0=\{(\vec{A}, \vec{A}) \mid \vec{A} \in \intset^k\}$ and $R^{i+1} = R^i \cdot R$.

\hide{
	\paragraph{The logic $\edbs$.} 
	Formulae of $\edbs$ are of the form $\varphi\wedge \psi$. $\varphi$ is defined by the following rules:
	\begin{align*}
	\varphi & ::= T_s\ \opset\ T_s \mid  T_i\ \op\ c \mid T_i\ \op\ T_i + c \mid  \varphi \wedge \varphi \mid  \neg \varphi \mid \forall x.\ \varphi,
	\end{align*}
	where 
	\[
	\begin{array}{l c l r}
	T_i &::=& x \mid \inf(T_s) \mid \sup(T_s) &\qquad \mbox{(integer terms)}\\
	T_s &::= &  \emptyset   \mid S \mid \{T_i\}  \mid T_s \cup T_s \mid T_s \cap T_s \mid T_s \setminus T_s &\qquad \mbox{(set terms)}\\
	\end{array}
	\]
	and $\psi$ is a conjunction of the atomic formulae of the form $T_a\ \op\ T_a$,
	where $T_a$ is defined by,
	\[
	\begin{array}{l c l}
	T_a &::= & c \mid \min(S) \mid \max(S) \mid cT_a \mid T_a + T_a,
	\end{array}
	\]
	%
	%
	where $c$ is an integer constant and $\op \in \{=, \ge, \le\}$. 
}


\section{Linearly compositional SLID with set data constraints}\label{sec-prelm}


In this section, we introduce separation logic with \emph{linearly compositional} inductive predicates and \emph{set data} constraints, denoted by $\lcslidset[P]$, where $P$ is an \emph{inductive predicate}. 
In addition to the integer and set data types introduced in Section~\ref{sec:logic}, we also consider 
the \emph{location} data type $\loc$. 
As a convention, $l,l',\dots \in \loc$ denote locations 
%
and $E,F,X,Y,\cdots$ range over location variables.
%
%
%
We consider location fields associated with $\loc$ 
and data fields associated with $\intnum$. 
 
$\lcslidset[P]$ formulae may contain inductive predicates, each of which is of the form $P(\vec{\alpha}; \vec{\beta}; \vec{\xi}) $ and has an associated inductive definition. The parameters 
are classified into three groups: \emph{source parameters} $\vec{\alpha}$, \emph{destination parameters} $\vec{\beta}$, and  \emph{static parameters} $\vec{\xi}$. We require that the source parameters $\vec{\alpha}$ and the destination parameters $\vec{\beta}$  are \emph{matched} in type, namely, the two tuples have the same length $\ell>0$ and for each $i\in[\ell]$, $\alpha_i$ and $\beta_i$ have the same data type. Static parameters are typically used to store some static (global) information of dynamic data structures, e.g., the target location of tail pointers (cf.\ $\stlseg$ in Example~\ref{ex:datastructures}). 
Moreover, we assume that for each $i \in [\ell]$, $\alpha_i$ is of either the location type, or the \emph{set type}. (There are no parameters of the integer type.) Without loss of generality, it is assumed that the first components of $\vec{\alpha}$ and $\vec{\beta}$ are  location variables; we usually explicitly write $E, \vec{\alpha}$ and $F, \vec{\beta}$.

$\lcslidset[P]$ formulae comprise three types of formulae:  \emph{pure formulae} $\Pi$,   \emph{data formulae} $\Delta$, 
and   \emph{spatial formulae} $\Sigma$. The data formulae are simply $\ps$ introduced in
Section~\ref{sec:logic}, while $\Pi$ and $\Sigma$ are defined by the following rules, 
\vspace{-1mm}
$$
\begin{array}{r c l cr}
\Pi &::=& E = F \mid E \neq F \mid \Pi \wedge \Pi & \ \ & \mbox{(pure formulae)}\\
\Sigma &::=& \slemp \mid E \mapsto (\rho) \mid P(E, \vec{\alpha}; F, \vec{\beta}; \vec{\xi}) \mid \Sigma \sep \Sigma  &\ \ & \mbox{(spatial formulae)}\\
\vspace{-2mm}
\rho & ::= & (f, X) \mid  (d, T_i) \mid \rho, \rho & \ \ & \mbox{(fields)}\\
\end{array}
$$
where 
$T_i$ is an integer term as  in Definition~\ref{def:qgdbs}, and $f$ (resp.\ $d$) is  a location (resp.\ data) field.  
%
For spatial formulae $\Sigma$, formulae of the form $\slemp$, $E \mapsto (\rho)$, or $P(E, \vec{\alpha}; F, \vec{\beta};\vec{\xi})$ are called \emph{spatial atoms}. In particular, formulae of the form $E \mapsto (\rho)$ and $P(E, \vec{\alpha}; F, \vec{\beta};\vec{\xi})$ are called   \emph{points-to} and  \emph{predicate} atoms respectively. Moreover, $E$ is \emph{the root} of these points-to or predicate atoms.

\smallskip

\noindent\emph{Linearly compositional inductive predicates.} 
An inductive predicate $P$ is \emph{linearly compositional} if the inductive definition of $P$ is given by the following two rules,
\vspace{-5mm}
\begin{itemize}
\item base rule $R_0: P(E, \vec{\alpha}; F, \vec{\beta}; \vec{\xi})::= E = F \wedge \vec{\alpha} = \vec{\beta} \wedge \slemp$,
\item inductive rule $R_1: P(E, \vec{\alpha}; F, \vec{\beta}; \vec{\xi})::= \exists \vec{X} \exists \vec{S}.\ \varphi \wedge E \mapsto (\rho) \sep P(Y, \vec{\gamma}; F, \vec{\beta}; \vec{\xi})$. 
\end{itemize}
 \vspace{-1mm}
The left-hand (resp. right-hand) side of a rule is called the \emph{head} (resp. \emph{body}) of the rule. We note that the body of $R_1$ does \emph{not} contain  pure formulae.

In the sequel, we specify some constraints on the inductive rule $R_1$ which are vital to obtain \emph{complete} decision procedures for the satisfiability problem. 
%
\vspace{-2mm}
\begin{description}
\item[C1] None of the variables from $F, \vec{\beta}$ occur elsewhere in the right-hand side of $R_1$, that is, in $\varphi$, $E \mapsto (\rho)$.

\item[C2] The data constraint $\varphi$ in the body of $R_1$ is a $\dbs$ formula.


\item[C3] For each atomic formula in $\varphi$, there is $i$ such that all the variables in the atomic formula are from $\{\alpha_i,\gamma_i\}$. 


%


%
\item[C4] Each variable occurs in each of $P(Y,\vec{\gamma}; F, \vec{\beta}; \vec{\xi})$ and $\rho$ at most once.
 

\item[C5] $\vec{\xi}$ contains \emph{only} location variables and all location variables from $\vec{\alpha} \cup \vec{\xi} \cup \vec{X}$ occur in $\rho$.
 

\item[C6] $Y \in \vec{X}$ and $\vec{\gamma} \subseteq \{E\} \cup \vec{X} \cup \vec{S}$.
\end{description}
\vspace{-3mm}

\noindent Note that, by {\bf C6}, none of the variables from $\vec{\alpha} \cup \vec{\xi}$ occur in $\vec{\gamma}$. Moreover, from {\bf C5} and {\bf C6}, $Y$ occurs in $\rho$, which
guarantees that in each model of $P(E, \vec{\alpha}; F, \vec{\beta}; \vec{\xi})$, the sub-heap represented by $P(E, \vec{\alpha}; F, \vec{\beta}; \vec{\xi})$, seen as a directed graph, is connected.
We remark that these constraints are undeniably technical. However, in practice the inductive predicates satisfying these constraints are usually sufficient to define linear data structures with set data constraints, cf.\ Example~\ref{ex:datastructures}. 



For  an inductive predicate $P$, let $\flds(P)$ 
denote the set of all fields 
occurring in the inductive rules of $P$.  
For a spatial atom $a$, let $\flds(a)$ denote the set of fields that $a$ refers to: if $a=E \mapsto (\rho)$, then $\flds(a)$ is the set of fields occurring in $\rho$;  if $a=P(-)$, then $\flds(a) =\flds(P)$. 

We write $\lcslidset[P]$ for the collection of separation logic formulae $\phi = \Pi \wedge \Delta \wedge \Sigma$ satisfying the following constraints: (1) $P$ is a linearly compositional inductive predicate, and (2) each predicate atom of $\Sigma$ is of the form $P(-)$, and for each points-to atom occurring in $\Sigma$, the set of fields of this atom is $\flds(P)$.

For an $\lcslidset[P]$ formula $\phi$, let $\vars(\phi)$ (resp.\ $\lvars(\phi)$, 
resp. $\dvars(\phi)$, resp. $\svars(\phi)$) denote the set of (resp.\ location, 
resp.\ integer, resp. set) variables occurring in $\phi$. Moreover, we use $\phi[\vec{\mu}/\vec{\alpha}]$ to denote the simultaneous replacement of the variables $\alpha_j$ by $\mu_j$ in $\phi$.
We adopt the standard \emph{classic, precise semantics} of $\lcslidset[P]$ in terms of \emph{states}. In particular, 
a \emph{state} is a pair $(s,h)$, where $s$ is an assignment and $h$ is a heap.
The details can be found in \iftoggle{fullver}{Appendix~\ref{app:slsemantics}}{\cite{GCW18}}.  

\begin{example}\label{ex:datastructures}
We collect a few examples of linear data structures with set data constraints definable in $\lcslidset[P]$:   

\vspace{-4mm}
\begin{center}
\begin{tabular}{|l|}

	\hline
	
	$\sdllseg$ for sorted doubly linked list segments, \\
	\small{ 
	$
	\begin{array} {l c l}
	\sdllseg(E,P, S; F, L, S') & ::= & E = F \wedge P = L  \wedge S = S'  \wedge \slemp,\\
	\sdllseg(E,P, S; F, L, S') & ::= & \exists X, S''.\ S=S'' \cup \{\min(S)\}\ \wedge \\
	&&  \hspace{-2cm} E \mapsto ((\fnext,X),(\fprev,P), (\fdata, \min(S))) \sep \sdllseg(X,E, S''; F, L, S').\\
	\end{array}
	$}\\
	
%
	
	\hline
	
	$\plseg$ for list segments where the data values are consecutive, \\
	\small{	
	$
	\begin{array}{l c l}
	\plseg(E, S; F, S') & ::=  & E = F \wedge S = S' \wedge \slemp,\\
	\plseg(E, S; F, S') &::= & \exists X, S''.\ S= S'' \cup \{\min(S)\} \wedge \min(S'') = \min(S)+1 \ \wedge  \\
	& & \hspace{1cm}  E \mapsto ((\fnext, X), (\fdata,\min(S))) \sep \plseg(X, S''; F,S').\\
	\end{array} 
	$} \\
	
	\hline
	
	$\ldllseg$ for doubly list segments, to mimic lengths with sets, 
	\\
	\small{	
	$
	\begin{array}{l c l} 
	\ldllseg(E, P, S; F, L, S') & ::=  & E = F \wedge P = L \wedge S = S'  \wedge \slemp,\\
	\ldllseg(E, P, S; F, L, S') &::= & \exists X, S''.\ S= S'' \cup \{\max(S)\} \wedge \max(S'')= \max(S)-1\ \wedge    \\
	& & \hspace{2mm} E \mapsto ((\fnext, X),(\fprev,P))  \sep \ldllseg(X, E, S''; F,L, S').\\
	\end{array}
	$}\\
	\hline
\end{tabular}
\end{center}
\vspace{-4mm}
\end{example}

\section{Satisfiability of $\lcslidset[P]$}\label{sec-sat}


The satisfiability problem is to decide whether there is a state (an assignment-heap pair) satisfying $\phi$ 
for a given   $\lcslidset[P]$ formula $\phi$. We shall follow the approach adopted in \cite{ELSV14,GuCW16}, i.e., to construct $\boolabs(\phi)$, an abstraction of $\phi$ that is equisatisfiable to $\phi$. The key ingredient of the construction is to compute the transitive closure of the data constraints extracted from the inductive rule of $P$. 

Let $\phi = \Pi \wedge \Delta \wedge \Sigma$ be an 
$\lcslidset[P]$ formula. Suppose $\Sigma= a_1 \sep \cdots \sep a_n$, where each $a_i$ is either a points-to atom or a predicate atom.  For predicate atom $a_i=P(Z_1,\vec{\mu}; Z_2, \vec{\nu}; \vec{\chi})$ 
we assume that the inductive rule for $P$ is 

\hspace{5mm} $
R_1: P(E, \vec{\alpha}; F, \vec{\beta}; \vec{\xi})::= \exists \vec{X} \exists \vec{S}.\ \varphi \wedge E \mapsto (\rho) \sep P(Y, \vec{\gamma}; F, \vec{\beta}; \vec{\xi}).\hspace*{\fill} (*)
$

\medskip
We extract the data constraint $\varphi_P(\dt(\vec{\alpha}), \dt(\vec{\beta}))$ out of  $R_1$. Formally, we define $\varphi_P(\dt(\vec{\alpha}), \dt(\vec{\beta}))$ as $\varphi[\dt(\vec{\beta})/\dt(\vec{\gamma})]$,  where $\dt(\vec{\alpha})$ (resp. $\dt(\vec{\gamma})$, $\dt(\vec{\beta})$) is the projection of $\vec{\alpha}$ (resp. $\vec{\gamma}$, $\vec{\beta}$) to  data variables. For instance, $\varphi_{\ldllseg}(S, S'):= \left(S= S'' \cup \{\max(S)\} \wedge \max(S'')= \max(S)-1\right)[S'/S'']=S= S' \cup \{\max(S)\} \wedge \max(S') = \max(S) - 1$.  

We can construct $\boolabs(\phi)$ 
with necessary adaptations from \cite{GuCW16}.  
For each spatial atom $a_i$, $\boolabs(\phi)$ introduces a Boolean variable to denote whether $a_i$ corresponds to a nonempty heap or not. With these Boolean variables, the semantics of separating conjunction are encoded in $\boolabs(\phi)$. Moreover, for each predicate atom $a_i$, $\boolabs(\phi)$ contains an abstraction of $a_i$, where the formulae $\ufld_1(a_i)$ and $\ufld_{\ge 2}(a_i)$ are used. Intuitively, $\ufld_1(a_i)$ and $\ufld_{\ge 2}(a_i)$ correspond to the separation logic formulae obtained by unfolding the rule $R_1$ \emph{once} and \emph{at least twice} respectively. We include the construction here so one can see the role of the transitive closure in $\boolabs(\phi)$.
The details of $\boolabs(\phi)$ can be found in \iftoggle{fullver}{Appendix~\ref{app-sec-sat}}{\cite{GCW18}}.  


%



Let $a_i = P(Z_1,\vec{\mu}; Z_2, \vec{\nu}; \vec{\chi})$ and $R_1$ be the inductive rule in Eqn. ($*$). 
If $E$ occurs in $\vec{\gamma}$ in the body of $R_1$,  we  use $\idx_{(P,\vec{\gamma},E)}$ to denote the unique index $j$ such that $\gamma_j = E$. (The uniqueness follows from {\bf C4}.)

\vspace{-1mm}
\begin{definition}[$\ufld_1(a_i)$ and $\ufld_{\ge 2}(a_i)$]
$\ufld_1(a_i)$ and $\ufld_{\ge 2}(a_i)$ are defined by distinguishing the following two cases:
 \vspace{-1mm}
 \begin{itemize}
 \item 
If $E$ occurs in $\vec{\gamma}$ in the body of $R_1$, then  
{\small
\smallskip
$\ufld_1(a_i)  :=   (E  = \beta_{\idx_{(P,\vec{\gamma},E)}} \wedge \varphi_P(\dt(\vec{\alpha}), \dt(\vec{\beta})))[Z_1/E, \vec{\mu}/\vec{\alpha}, Z_2/F, \vec{\nu}/\vec{\beta}, \vec{\chi}/\vec{\xi}]$
\smallskip
}
and $\ufld_{\ge 2}(a_i):=$

{\small
\smallskip
{
\hspace{-7mm}$ 
\begin{array}{l c l } 
\left(
\begin{array}{l}
E \neq \beta_{\idx_{(P,\vec{\gamma},E)}} \wedge E \neq \gamma_{2, \idx_{(P,\vec{\gamma},E)}}\ \wedge \\
\varphi_P[\dt(\vec{\gamma_1})/\dt(\vec{\beta})] \wedge  \varphi_P[\dt(\vec{\gamma_1})/\dt(\vec{\alpha}), \dt(\vec{\gamma_2})/\dt(\vec{\beta})]\ \wedge\\
 (\TC[\varphi_P])[\dt(\vec{\gamma_2})/\dt(\vec{\alpha})]
\end{array}
\!\!\!\right)
[Z_1/E, \vec{\mu}/\vec{\alpha}, Z_2/F, \vec{\nu}/\vec{\beta}, \vec{\chi}/\vec{\xi}],
\end{array}
$
}
\smallskip
}
 
where $\vec{\gamma_1}$ and $\vec{\gamma_2}$ are fresh variables.
%

\item Otherwise, let 
$
\ufld_1(a_i) :=\varphi_P [Z_1/E, \vec{\mu}/\vec{\alpha}, Z_2/F, \vec{\nu}/\vec{\beta}, \vec{\chi}/\vec{\xi}]
$
and 

{\small
\smallskip
$\ufld_{\ge 2}(a_i):=
\begin{array}{l c l }
\left(
\begin{array}{l}
\varphi_P[\dt(\vec{\gamma_1})/\dt(\vec{\beta})]\ \wedge \\
\varphi_P[\dt(\vec{\gamma_1})/\dt(\vec{\alpha}), \dt(\vec{\gamma_2})/\dt(\vec{\beta})]\ \wedge\\ 
(\TC[\varphi_P])[\dt(\vec{\gamma_2})/\dt(\vec{\alpha})]
\end{array}
\!\!\!\right)
[Z_1/E, \vec{\mu}/\vec{\alpha}, Z_2/F, \vec{\nu}/\vec{\beta}, \vec{\chi}/\vec{\xi}],
\end{array}
$
\smallskip
}

where $\vec{\gamma_1}$ and $\vec{\gamma_2}$ are fresh variables.
\end{itemize}
\vspace{-2mm}
\end{definition}

 

\vspace{-1mm}
\noindent Here, $\TC[\varphi_P](\dt(\vec{\alpha}), \dt(\vec{\beta}))$ denotes the transitive closure of $\varphi_P$. In Section~\ref{sec:tc}, it will be shown that $\TC[\varphi_P](\dt(\vec{\alpha}), \dt(\vec{\beta}))$ can be written as an $\qgdbs$ formula. As a result, since we are only concerned with satisfiability and can treat the location data type $\loc$ simply as integers $\intnum$,  $\boolabs(\phi)$ can also be read as an $\qgdbs$ formula. In Section~\ref{sec:sat-qgdbs}, we shall show that the satisfiability of $\qgdbs$ is  decidable. Following this chain of reasoning, we conclude that \emph{the satisfiability of $\lcslidset[P]$ formulae is decidable}. 
%


\section{Transitive closure of difference-bound set relations} \label{sec:tc}


In this section, we show how to compute the transitive closure of the difference-bound set relation $R$ given by a  $\dbs$ formula $\varphi_R(\vec{S}, \vec{S}')$. Our approach is, in a nutshell, to encode $TC[\varphi_R](\vec{S}, \vec{S'})$ into $\qgdbs$. 
%
%
We shall only sketch part of a simple case,  i.e., in $\varphi_R(S, S')$ only one source and destination set parameter are present. The details are however given in \iftoggle{fullver}{Appendix~\ref{app:tcsimple}}{\cite{GCW18}}.

Recall that, owing to the simplicity of $\dbs$, the integer terms $T_i$ in $\varphi_R(S,S')$ can only be $\min(S)$, $\max(S)$, $\min(S')$ or $\max(S')$, whereas the set terms $T_s$ are $\emptyset$, $\{\min(S)\}$, $\{\min(S')\}$, $\{\max(S)\}$, $\{\max(S')\}$, or their union. 
%
For reference, we write $\varphi_R(S, S') = \varphi_{R,1} \wedge \varphi_{R,2}$, where $\varphi_{R,1}$ is an equality of set terms (i.e., they are of the form $S = S' \cup T_s$ or $S' = S \cup T_s$), and $\varphi_{R,2}$ is a conjunction of constraints over integer terms (i.e., a conjunction of formulae $T_i \ \le\ T_i + c$). $\varphi_{R,1}$ and $\varphi_{R,2}$ will be referred to as the \emph{set} and \emph{integer subformula} of $\varphi_R(S,S')$ respectively. 
We shall focus on the case $\varphi_{R,1} := S = S' \cup T_s$. The symmetrical case  $\varphi_{R,1} := S' = S \cup T_s$ can be adapted easily. 

The integer subformula $\varphi_{R,2}$ can be represented by an \emph{edge-weighted directed graph} $\Gg(\varphi_{R,2})$, where the vertices are all integer terms appearing in $\varphi_{R,2}$, and there is an edge from $T_1$ to $T_2$ with weight $c$ iff $T_1=T_2+c$ (equivalent to $T_2 = T_1 - c$), or  $T_1\leq T_2+c$, or $T_2+c\geq T_1$ appears in $\varphi_{R,2}$. 
%
%
The weight of a path in $\Gg(\varphi_{R,2})$ is the sum of the weights of the edges along the path. A \emph{negative cycle} in $\Gg(\varphi_{R,2})$ is a cycle with negative weight. It is known that $\varphi_{R,2}$ is satisfiable iff $\Gg(\varphi_{R,2})$ contains no negative cycles \cite{Mine01}. 
Suppose $\varphi_{R,2}$ is satisfiable. We define the \emph{normal form} of $\varphi_{R,2}$, denoted by  $\norm(\varphi_{R,2})$, as the conjunction of the formulae $T_1 \le T_2 + c$ such that $T_1 \neq T_2$, $T_2$ is reachable from $T_1$ in $\Gg(\varphi_{R,2})$, and $c$ is 
path from $T_1$ to $T_2$ with the minimal weight in $\Gg(\varphi_{R,2})$. 

$S$ (resp. $S'$) is said to be \underline{\emph{surely nonempty} in $\varphi_R$} if $\min(S)$ or $\max(S)$ (resp. $\min(S')$ or $\max(S')$) occurs in $\varphi_R$; otherwise, $S$ (resp. $S'$) 
is  \underline{\emph{possibly empty} in $\varphi_R$}. Recall that, according to the semantics, an occurrence of $\min(S)$ or $\max(S)$ (resp. $\min(S')$ or $\max(S')$) in $\varphi_R$ implies that $S$ (resp. $S'$) is interpreted as a nonempty set in every satisfiable assignment. Provided that $S'$ is nonempty, we know that $\min(S')$ and $\max(S')$ belong to $S'$. Therefore, for simplicity, here we assume that in $S= S' \cup T_s$, $T_s$ contains neither $\min(S')$ nor $\max(S')$.
The situation that $T_s$ contains $\min(S')$ and $\max(S')$ can be dealt with in a similar way.
%

\medskip
\noindent\emph{Saturation.} For technical convenience, we introduce a concept of saturation. The main purpose of saturation is to regularise $T_s$ and $\varphi_{R,2}$, which would make the transitive closure construction more ``syntactic".  

\vspace{-1mm}
\begin{definition}
	\label{def-saturate}
Let $\varphi_R(S, S') := S = S' \cup T_s \wedge \varphi_{R,2}$ be a $\dbs$ formula. Then $\varphi_R(S, S')$ is saturated if $\varphi_R(S,S')$ satisfies the following conditions 
%
\vspace{-1mm}
\begin{itemize}
\item $\varphi_{R,2}$ is satisfiable and  in normal forms,

\item   
$T_s \subseteq\{\max(S),\min(S)\}$, 
\item if $S$ (resp. $S'$) is surely nonempty in $\varphi_R$, then $\varphi_{R,2}$ contains a conjunct $\min(S) \le \max(S) - c $ for some $c \ge 0$ (resp. $\min(S') \le \max(S') - c'$ for some $c' \ge 0$),
\item if both $S$ and $S'$ are surely nonempty in $\varphi_R$, then 
\vspace{-1mm}
\begin{itemize}
\item $\varphi_{R,2}$ contains two conjuncts $\min(S) \le \min(S') - c$ and $\max(S') \le \max(S) - c'$ for some $c,c' \ge 0$,
\item $\min(S) \not \in T_s$ iff $\varphi_{R,2}$ contains the conjuncts $\min(S) \le \min(S')$ and $\min(S') \le \min(S)$,
\item $\max(S) \not \in T_s$ iff $\varphi_{R,2}$ contains the conjuncts $\max(S') \le \max(S)$ and $\max(S) \le \max(S')$,
\end{itemize}
%

\item if $\varphi_{R,2}$ contains the conjuncts $\min(S) \le \max(S)$ and  $\max(S) \le \min(S)$, then $\max(S) \not \in T_s$ (possibly $\min(S) \in T_s$).
\end{itemize}
\vspace{-1mm}
\end{definition}
\vspace{-3mm}
For a formula $\varphi_R(S,S') :=  S = S' \cup T_s \wedge \varphi_{R,2}$, one can easily saturate $\varphi_{R}$, yielding a saturated formula  $\saturate(\varphi_R(S,S'))$. (It is possible, however, to arrive at an unsatisfiable formula, then we are done.) 


\vspace{-2mm}
\begin{proposition}\label{prop-saturate}
Let $\varphi_R(S,S') := \varphi_{R,1} \wedge \varphi_{R,2}$ be a $\dbs$ formula such that $\varphi_{R,1} := S = S' \cup T_s$ and $\varphi_{R,2}$ is satisfiable. Then $\varphi_R$ can be transformed, in polynomial time, to an equisatisfiable formula  $\saturate(\varphi_R(S,S'))$, and if the integer subformula  of $\saturate(\varphi_R(S,S'))$ is satisfiable, then  $\saturate(\varphi_R(S,S'))$ is saturated.
%
\end{proposition}
\vspace{-1mm}

In the sequel, we assume that $\varphi_R(S,S') := \varphi_{R,1} \wedge \varphi_{R,2}$ is satisfiable and \emph{saturated}. For notational convenience, for $A \subseteq \{\min(S), \max(S), \min(S'), \max(S')\}$ with 
$|A| = 2$, let $\lfloor\varphi_{R,2}\rfloor_A$ denote the conjunction of atomic formulae in $\varphi_{R,2}$ where \emph{all} the elements of $A$ occur.  

%
%
%
%
Evidently, $\lfloor\varphi_{R,2}\rfloor_A$ gives a partition of atomic formulae of $\varphi_{R,2}$. Namely, 

\smallskip
\hspace{1cm} $\varphi_{R,2}= \bigwedge_{A\subseteq \{\min(S), \max(S), \min(S'), \max(S')\}, |A| = 2}\  \lfloor\varphi_{R,2}\rfloor_A$.
\smallskip

\noindent We proceed by a case-by-case analysis of $\varphi_{R,1}$. There are four cases: (I) $\varphi_{R,1} := S = S'$, (II) $\varphi_{R,1} := S = S' \cup \{\min(S)\}$, (III) $\varphi_{R,1} = S = S' \cup \{\max(S)\}$ and 
(IV)  $\varphi_{R,1} = S = S' \cup \{\min(S), \max(S)\}$. Case (I) is trivial, and Case (III) is symmetrical to (II). However, both (II) and (IV) are technically involved. 
%
%
%
%
%
%
%
We shall only give a ``sample" treatment of these cases, i.e., part of arguments for Case (II); the full account of Case (II) and (IV) are given in \iftoggle{fullver}{Appendix~\ref{app:tcsimple}}{\cite{GCW18}}. 

To start with, Case (II) can be illustrated schematically as $\underbrace{|-\overbrace{|------|}^{S'}}_S$. 
We observe that $S$ is surely nonempty in $\varphi_R$.  
We then distinguish two subcases depending on whether $S'$ is possibly empty or surely nonempty in $\varphi_R$. Here we give the details of the latter subcase because it is more interesting. In this case,  both $S$ and $S'$ are surely nonempty in $\varphi_R$. By Definition~\ref{def-saturate}(4--5), 
$\varphi_{R,2}$ contains a conjunct $\min(S) \le \min(S') - c$  for some $c \ge 0$, as well as $\max(S') \le \max(S)$ and $\max(S) \le \max(S')$ (i.e., $\max(S')=\max(S)$).
Therefore, we can assume  

\smallskip
$
\begin{array}{l c l}
\varphi_{R,2} & = & \max(S') \le \max(S) \wedge \max(S) \le \max(S') \wedge \lfloor \varphi_{R,2} \rfloor_{\min(S),\min(S')}\ \wedge \\
& &  
\lfloor \varphi_{R,2} \rfloor_{\min(S),\max(S)} \wedge \lfloor \varphi_{R,2} \rfloor_{\min(S'),\max(S')}.
\end{array}
$
\smallskip

\noindent Note that in $\varphi_{R,2} $ above, the redundant subformulae 
$\lfloor \varphi_{R,2} \rfloor_{\min(S),\max(S')}$  and $\lfloor \varphi_{R,2} \rfloor_{\min(S'),\max(S)}$ have been omitted.

The formula $\lfloor \varphi_{R,2} \rfloor_{\min(S),\min(S')}$ is said to be \emph{strict} if it contains a conjunct $\min(S) \le \min(S') - c$ for some $c > 0$. Otherwise, it is said to be \emph{non-strict}. Intuitively, if $\lfloor \varphi_{R,2} \rfloor_{\min(S),\min(S')}$ is strict, then for $n, n' \in \intnum$, the validity of $(\lfloor \varphi_{R,2} \rfloor_{\min(S),\min(S')})[n/\min(S), n'/\min(S')]$ implies that $n < n'$. 
For the sketch we only present \emph{the case that $\lfloor \varphi_{R,2} \rfloor_{\min(S),\min(S')}$ is strict}; the other cases are similar and can be found in \iftoggle{fullver}{Appendix~\ref{app:tcsimple}}{\cite{GCW18}}.

Evidently, $\TC[\varphi_R](S, S')$ can be written as $(S= S') \vee \bigvee \limits_{n \ge 1} \varphi^{(n)}_R$, 
where $\varphi^{(n)}_R$ is obtained by unfolding $\varphi_R$ for $n$ times, that is, 

{\small
\smallskip
$
\varphi^{(n)}_R = \exists S_1,\cdots, S_{n+1}. \left(
\begin{array}{l}
S_1 = S \wedge S_{n+1} = S'\ \wedge \\
\bigwedge \limits_{i \in [n]} (S_i = S_{i+1} \cup \{\min(S_i)\} \wedge \varphi_{R,2}[S_i/S, S_{i+1}/S'])
\end{array}
\right),
$
}
\smallskip

\noindent where $\varphi_{R,2}[S_i/S, S_{i+1}/S']$ is obtained from $\varphi_{R,2}$ by replacing $S$ (resp. $S'$) with $S_i$ (resp. $S_{i+1}$).
%

\noindent Clearly, $\varphi^{(1)}_R = \varphi_R$, and

{\small
\smallskip
$
\varphi^{(2)}_R = \exists S_{2}.\  (S = S_{2} \cup \{\min(S)\} \wedge S_2 = S' \cup \{\min(S_2)\} \wedge \varphi_{R,2}[S_2/S'] \wedge \varphi_{R,2}[S_2/S]).
$
}

\noindent For $\varphi^{(n)}_R$ where $n \ge 3$, we first simplify $\varphi^{(n)}_R$  to construct a finite formula for $\TC[\varphi_R](S, S')$. 
The subformula $\bigwedge \limits_{i \in [n]} (S_i = S_{i+1} \cup \{\min(S_i)\} \wedge \varphi_{R,2}[S_i/S, S_{i+1}/S'])$ can be rewritten as 

{\small
\smallskip
$
\bigwedge \limits_{i \in [n]} 
\left(
\begin{array}{l}
S_i = S_{i+1} \cup \{\min(S_i)\} \wedge \max(S_i) = \max(S_{i+1})\ \wedge \\
(\lfloor \varphi_{R,2} \rfloor_{\min(S),\min(S')}[S_i/S,S_{i+1}/S']) \wedge
(\lfloor \varphi_{R,2} \rfloor_{\min(S),\max(S)}[S_i/S])\ \wedge \\ 
(\lfloor \varphi_{R,2} \rfloor_{\min(S'),\max(S')}[S_{i+1}/S'])
\end{array}
\right).
$
\smallskip
}

\noindent Because $S_{i} = S_{i+1} \cup \{\min(S_{i})\}$ for each $i \in [n]$, we have $\max(S_1) = \cdots = \max(S_n)$ and $\min(S_1) \le \cdots \le \min(S_n)$. Since $\lfloor \varphi_{R,2} \rfloor_{\min(S), \max(S)}$ is a conjunction of difference-bound constraints involving $\min(S)$ and $\max(S)$ only, we have $\bigwedge \limits_{i \in [n]} \lfloor \varphi_{R,2} \rfloor_{\min(S), \max(S)}[S_i/S]$ is equivalent to 
$\lfloor \varphi_{R,2} \rfloor_{\min(S), \max(S)}[S_1/S] \wedge \lfloor \varphi_{R,2} \rfloor_{\min(S), \max(S)} [S_n/S]$. To see this, assume, for instance,  

\smallskip
\hspace{2cm}$\lfloor \varphi_{R,2} \rfloor_{\min(S), \max(S)} \equiv c \le \max(S) - \min(S) \le c'$
\smallskip
 
\noindent for some constants $c, c' \ge 0$ with $c \le c'$. Then $\max(S_1) - \min(S_1) \le c'$ implies $\max(S_i) - \min(S_i) \le c'$ for each $i \in [n]$, and $c \le \max(S_n) - \min(S_n)$ implies $c \le \max(S_i) - \min(S_i)$ for each $i \in [n]$. Therefore, 
$\lfloor \varphi_{R,2} \rfloor_{\min(S), \max(S)}[S_1/S] \wedge \lfloor \varphi_{R,2} \rfloor_{\min(S), \max(S)} [S_n/S] \equiv c \le \max(S_1) - \min(S_1) \le c' \wedge c \le \max(S_n) - \min(S_n) \le c'$ implies that
$\bigwedge \limits_{i \in [n]} \lfloor \varphi_{R,2} \rfloor_{\min(S), \max(S)}[S_i/S]$, thus they are equivalent. (The other direction is trivial.) Likewise, one has $\lfloor \varphi_{R,2} \rfloor_{\min(S'), \max(S')}[S_2/S'] \wedge \lfloor \varphi_{R,2} \rfloor_{\min(S'), \max(S')} [S_{n+1}/S']$ implies $\bigwedge \limits_{i \in [n]} \lfloor \varphi_{R,2} \rfloor_{\min(S'), \max(S')}[S_{i+1}/S']$, thus they are equivalent.
Therefore, $\varphi^{(n)}_R$ can be transformed into 

\smallskip
{\small
\noindent$
\exists S_2, S_n.
\left(\begin{array}{l}
\lfloor \varphi_{R,2} \rfloor_{\min(S),\max(S)} \wedge (\lfloor \varphi_{R,2} \rfloor_{\min(S),\max(S)}[S_n/S])\ \wedge \\
(\lfloor \varphi_{R,2} \rfloor_{\min(S'),\max(S')}[S_{2}/S']) \wedge \lfloor \varphi_{R,2} \rfloor_{\min(S'),\max(S')} \wedge S = S_2 \cup \{\min(S)\}\  \wedge \\
S_n = S' \cup \{\min(S_n)\} \wedge \max(S) = \max(S_2)\wedge \max(S_n) = \max(S')\ \wedge \\
(\lfloor \varphi_{R,2} \rfloor_{\min(S),\min(S')}[S_2/S']) \wedge (\lfloor \varphi_{R,2} \rfloor_{\min(S),\min(S')}[S_n/S])  \wedge \\
\exists S_3,\cdots, S_{n-1}.
\bigwedge \limits_{2 \le i \le n-1} 
\left(
\begin{array}{l}
S_i = S_{i+1} \cup \{\min(S_i)\} \wedge \max(S_i) = \max(S_{i+1})\ \wedge \\
(\lfloor \varphi_{R,2} \rfloor_{\min(S),\min(S')}[S_i/S,S_{i+1}/S'])
\end{array}
\right)
\end{array}
\right).
$
}

\smallskip

\noindent {\bf Claim}.
{\it
Suppose $n \ge 3$ and $\lfloor \varphi_{R,2} \rfloor_{\min(S),\min(S')}$ is strict. Then 

{\small
$\exists S_3,\cdots, S_{n-1}.
\bigwedge \limits_{2 \le i \le n-1} 
\left(
\begin{array}{l}
S_i = S_{i+1} \cup \{\min(S_i)\} \wedge \max(S_i) = \max(S_{i+1})\ \wedge \\
(\lfloor \varphi_{R,2} \rfloor_{\min(S),\min(S')}[S_i/S,S_{i+1}/S'])
\end{array}
\right)
$		
}

\noindent is equivalent to 

\smallskip
\noindent $
\small
\begin{array}{l}
S_n \neq \emptyset \wedge S_2 \setminus S_n \neq \emptyset  \wedge S_n \subseteq S_2 \wedge |S_2 \setminus S_n| = n-2 \wedge \max(S_2\setminus S_n) < \min(S_n)\ \wedge \\
\forall y, z.\ \suc((S_2 \setminus S_n)  \cup \{\min(S_n)\}, y, z) \rightarrow (\lfloor \varphi_{R,2} \rfloor_{\min(S), \min(S')}[y/\min(S), z/\min(S')]),
\end{array}
$
\smallskip
}

\noindent where $\suc(S, x, y)$ specifies intuitively that $y$ is the successor of $x$ in $S$, that is, 

\smallskip
\hspace{1cm} $\suc(S, x, y)= x \in S \wedge y \in S \wedge x < y \wedge \forall z  \in S.\ (z \le x \vee y \le z).$
\smallskip

\noindent   Note that  $|\cdot|$ denotes the set cardinality which can be easily encoded into $\qgdbs$. (\iftoggle{fullver}{Appendix~\ref{app:tcsimple}}{\cite{GCW18}} gives the proof of the claim.)
It follows that $\TC[\varphi_R](S, S')=$ 

\smallskip
$
\small
\begin{array}{l}
(S = S') \vee  \varphi_R(S,S') \vee \varphi^{(2)}_R(S,S')\ \vee \\
\exists S_1, S_2.
\left(\begin{array}{l}
S = S_1 \cup \{\min(S)\} \wedge S_2 = S' \cup \{\min(S_2)\}\ \wedge  \\
\max(S)  = \max(S_1) \wedge \max(S_2) = \max(S')\ \wedge\\
 S_2 \neq \emptyset \wedge S_1 \setminus S_2 \neq \emptyset \wedge 
S_2 \subseteq S_1 \wedge \max(S_1 \setminus S_2) < \min(S_2)\ \wedge \\
\lfloor \varphi_{R,2} \rfloor_{\min(S),\max(S)} \wedge (\lfloor \varphi_{R,2} \rfloor_{\min(S),\max(S)}[S_2/S])\ \wedge \\
(\lfloor \varphi_{R,2} \rfloor_{\min(S'),\max(S')}[S_1/S']) \wedge \lfloor \varphi_{R,2} \rfloor_{\min(S'),\max(S')} \  \wedge \\
(\lfloor \varphi_{R,2} \rfloor_{\min(S),\min(S')}[S_1/S']) \wedge (\lfloor \varphi_{R,2} \rfloor_{\min(S),\min(S')}[S_2/S]) \ \wedge \\
\forall y, z.
\left(
\begin{array}{l c l}
\suc((S_1 \setminus S_2)  \cup \{\min(S_2)\}, y, z) & \rightarrow & \\
& & \hspace{-3cm} (\lfloor \varphi_{R,2} \rfloor_{\min(S), \min(S')}[y/\min(S), z/\min(S')])
\end{array}
\right)
\end{array}
\right).
\end{array}
$

\hide{

\vspace{-5mm}
\subsection{The general case} \label{sec-tc-general}

\vspace{-1mm}

In this section, we consider the general case $\varphi_R(\vec{S}, \vec{S'})$ where a \emph{vector} of set variables is present. We write $len(\vec{S})$ for the \emph{length} of $\vec{S}$. 


Recall two conditions {\bf C3} and {\bf C4} in Section~\ref{sec-prelm}. It follows that, for a predicate $P$, the data formula $\varphi_P(\vec{S}; \vec{S'})$ extracted from the inductive rule of $P$ satisfies the \emph{independence} property. Namely, for each atomic formula $\varphi$ in $\varphi_P(\vec{S}, \vec{S'})$, there is some $i \in [len(\vec{S})]$  such that all the variables in $\varphi$ are from $\{S_i, S'_i\}$. This property is crucial to obtain a complete decision procedure.  

Let $len(\vec{S})=k$. Then $\varphi_R(\vec{S}, \vec{S'})$ can be rewritten into $\bigwedge \limits_{i \in [k]} \varphi^{[i]}_{R}(S_i, S'_i)$, where $\varphi^{[i]}_{R}(S_i, S'_i)$ is the conjunction of the atomic formulae of $\varphi_R$ that involve only the variables from $\{S_i, S'_i\}$. Moreover, for $i \in [k]$, let $\varphi^{[i]}_{R}(S_i, S'_i) = \varphi^{[i]}_{R,1}(S_i, S'_i) \wedge \varphi^{[i]}_{R,2}(S_i, S'_i)$, where $\varphi^{[i]}_{R,1}(S_i, S'_i)$ and $\varphi^{[i]}_{R,2}(S_i, S'_i)$ are the set and integer subformula of $\varphi^{[i]}_{R}(S_i, S'_i)$ respectively.

To compute $\TC[\varphi_R](\vec{S}, \vec{S'})$ for $\varphi_R(\vec{S}, \vec{S'}) = \bigwedge \limits_{i \in [k]} \varphi^{[i]}_{R}(S_i, S'_i)$,  we compute $\TC[\varphi^{[i]}_{R}](S_i, S'_i)$ separately for each $i\in [k]$, but  we also need to take the \emph{synchronisation} of different $\varphi^{(i)}_{R}(S_i, S'_i)$ into account. For instance, let $\varphi_R(S_1, S_2, S'_1, S'_2) \equiv \varphi^{[1]}_{R}(S_1, S'_1) \wedge \varphi^{[2]}_{R}(S_2, S'_2)$, where  $\varphi^{[i]}_{R}(S_i, S'_i) $ is obtained from $\varphi_R(S,S')$ in Example~\ref{exmp-subcase-II-ii} by replacing $S,S'$ with $S_i, S'_i$ respectively for $i =1,2$. 
It is not hard to see that $\TC[\varphi_R](S_1, S_2, S'_1, S'_2)$ should also include $\min(S'_1)- \min(S_1) = \min(S'_2)- \min(S_2)$. 

For $i \in [k]$, let $\Phi_i$ denote the formula obtained from $\TC[\varphi^{[i]}_{R}](S_i, S'_i)$ by removing the disjunct $S_i = S'_i$. Then we compute $\TC[\varphi_R](\vec{S}, \vec{S'})$ as,

\smallskip
$
\small
\begin{array}{l}
TC[\varphi_R](\vec{S}, \vec{S'}) = \big (\bigwedge_{i \in [k]} S_i = S'_i \big)\ \vee  \\
\left(
\bigwedge_{i \in [k]} \Phi_i \wedge 
\quantel\left(
\exists x.\ x > 0 \wedge \bigwedge_{i \in [k]} 
\left( 
\begin{array}{l}
(\lfloor \varphi^{[i]}_{R,2} \rfloor_{\min(S_i), \min(S'_i)})'\ \wedge \\
(\lfloor \varphi^{[i]}_{R,2} \rfloor_{\max(S_i), \max(S'_i)})'
\end{array} 
\right)
\right)
\right),
\end{array}
$
\smallskip

\noindent where $(\lfloor \varphi^{[i]}_{R,2} \rfloor_{\min(S_i), \min(S'_i)})'$ is obtained from $\lfloor \varphi^{[i]}_{R} \rfloor_{\min(S_i), \min(S'_i)}$ by replacing $\min(S_i) \le \min(S'_i) - c$ with $\min(S_i) \le \min(S'_i) - cx$, and $\min(S'_i) \le \min(S_i) + c'$ with $\min(S'_i) \le \min(S_i) + c'x$; similarly for $(\lfloor \varphi^{[i]}_{R,2} \rfloor_{\max(S_i), \max(S'_i)})'$. Here, $\quantel$ denotes the quantifier elimination procedure to remove the variable $x$. This is possible as $(\lfloor \varphi^{[i]}_{R,2} \rfloor_{\min(S_i), \min(S'_i)})'$ and $(\lfloor \varphi^{[i]}_{R,2} \rfloor_{\max(S_i), \max(S'_i)})'$ are both Presburger arithmetic formulae.  An example of the construction is given in Appendix~\ref{app:exgen}. 

}

\hide{
\begin{example}\label{exmp-general-case}
Let $\varphi_R(S_1, S_2, S'_1, S'_2) \equiv \varphi^{(1)}_{R}(S_1, S'_1) \wedge \varphi^{(2)}_{R}(S_2, S'_2)$, where for $i =1,2$, $\varphi^{(i)}_{R}(S_i, S'_i) $ is obtained from $\varphi_R(S,S')$ in Example~\ref{exmp-sat} by replacing $S,S'$ with $S_i, S'_i$ respectively. Then in $TC[\varphi_R](\vec{S}, \vec{S'})$, 

\smallskip
\hspace{4mm} $
\begin{array}{l l}
& \quantel\left(
\exists x.\ x > 0 \wedge \bigwedge_{i \in [k]} 
\left( 
\begin{array}{l}
(\lfloor \varphi^{(i)}_{R,2} \rfloor_{\min(S_i), \min(S'_i)})'\ \wedge \\
(\lfloor \varphi^{(i)}_{R,2} \rfloor_{\max(S_i), \max(S'_i)})'
\end{array} 
\right)
\right) \\
= & \exists x.\ x > 0 \wedge \bigwedge_{i =1,2} (\min(S'_i) = \min(S_i) + x \wedge \max(S_i) =  \max(S'_i))\\
= & \min(S_1 ) < \min(S'_1) \wedge \min(S'_1) - \min(S_1) = \min(S'_2) - \min(S_2)\ \wedge\\
&  \max(S_1) = \max(S'_1) \wedge   \max(S_2) = \max(S'_2).
\end{array}
$
\end{example}
}

\hide{
%
}

\section{Satisfiability of $\qgdbs$} \label{sec:sat-qgdbs}


In this section, we focus on the second ingredient of the procedure for deciding satisfiability of $\lcslidset[P]$,
i.e., the satisfiability of  $\qgdbs$. 
We first note that $\qgdbs$ is defined over $\intnum$. To show the decidability, it turns to be much easier to work on $\natnum$. We shall write  $\qgdbs_\intnum$ and $\qgdbs_\natnum$ to differentiate them when necessary. 
Moreover, for technical reasons, we also introduce $\qgdbs^-$, the fragment of $\qgdbs$ excluding formulae of the form $T_m \op 0$. 

The decision procedure for the satisfiability of $\qgdbs$ proceeds with the following three steps:
\vspace{-1mm}
\begin{description}
	\item[Step I.] Translate $\qgdbs_\intnum$ to $\qgdbs_\natnum$,
	\item[Step II.] Normalize an $\qgdbs_\natnum$ formula $\Phi(\vec{x}, \vec{S})$ into $\bigvee \limits_{i} (\Phi^{(i)}_{\mathrm{core}} \wedge \Phi^{(i)}_{\mathrm{count}})$, where $\Phi^{(i)}_{\mathrm{core}}$ is an $\qgdbs^-_\natnum$ formula, and $\Phi^{(i)}_{\mathrm{count}}$ is a conjunction of formulae of the form $T_m \op 0$ which contain only variables from $\vec{x} \cup \vec{S}$,
	\item[Step III.] For each disjunct $\Phi^{(i)}_{\mathrm{core}}\wedge \Phi^{(i)}_{\mathrm{count}}$, 
	construct a Presburger automaton (PA)  $\mathcal{A}^{(i)}_{\Phi}$ which captures the models of $\Phi^{(i)}_{\mathrm{core}}\wedge \Phi^{(i)}_{\mathrm{count}}$. Satisfiability is thus reducible to the nonemptiness of PA, which is decidable \cite{SSM08}. 
\end{description}
These steps are technically involved. In particular, the third step requires exploiting Presburger automata \cite{SSM08}. The details can be found in \iftoggle{fullver}{Appendix~\ref{app-sat-qgdbs}}{\cite{GCW18}}.

\hide{
\vspace{-1mm}

Although \textbf{Step I} and \textbf{Step II} are technically involved as well, the most interesting part is \textbf{Step III} which we will elaborate below. (The details of \textbf{Step I-II} can be found in \iftoggle{fullver}{Appendix~\ref{app-sat-qgdbs}}{\cite{GCW18}}.)
 
We start with some additional notations. 
First observe that there is a one-to-one correspondence between models of  $\Phi(\vec{x}, \vec{S})$ and finite words over  $2^{AP}$ with $AP =\{x_1, \cdots, x_k, S_1,\cdots, S_l\}$ satisfying that $x_j$ occurs in \emph{exactly one} position for each $j \in [k]$. A finite word $w = w_0 \cdots w_{n-1}$ over $2^{AP}$ is a finite sequence such that $w_i \in 2^{AP}$ for each $i\in \{0\} \cup [n-1]$. On the one hand, any model $(n_1, \cdots, n_k, A_1,\cdots, A_l) \in \natnum^k \times \natset^l$ of $\Phi(\vec{x}, \vec{S})$ can be interpreted as a finite word $w$ as follows: If $k = 0$ and $A_i =  \emptyset$ for all $i \in [l]$, then $w = \varepsilon$; otherwise let $|w|=1+\max(\{n_1,\cdots,n_k\} \cup \bigcup \limits_{i \in [l]} A_i)$, and, for each position $i \in \{0\} \cup [|w|-1]$, $w_i = P \subseteq AP$ iff $P = \{x_j \mid j \in [k], i = n_j\} \cup \{S_j \mid j \in [l], i \in A_j\}$. On the other hand, for a word $w \in (2^{AP})^*$ where $x_j$ occurs in exactly one position for each $j \in [k]$, a tuple $(n_1,\cdots, n_k, A_1,\cdots, A_l) \in \natnum^k \times \natset^l$ can be constructed such that for each $j \in [k]$, $n_j  = i$ iff $x_j \in w_i$, and for each $j \in [l]$, $A_j = \{i \in \{0\} \cup [|w|-1] \mid S_j \in w_i \}$. 
By slightly abusing the notation, we also use $\Ll(\Phi(\vec{x}, \vec{S}))$ to denote the set of words $w \in (2^{AP})^*$ such that $w \models \Phi$.

\vspace{-1mm}
\begin{definition}[Presburger automata]
A Presburger automaton (PA) $\Aa$ is a tuple $(Q, \Sigma, \delta, q_0, F, \Psi)$, where $(Q, \Sigma, \delta, q_0, F)$ is an NFA with $Q=\{q_0, q_1, \ldots, q_m\}$, and $\Psi(x_{q_0}, \cdots, x_{q_m})$ is a \emph{quantifier-free Presburger arithmetic} formula over the set of variables $\{x_{q_i} \mid i \in \{0\} \cup [m]\}$. 
\end{definition} 
\vspace{-1mm}

A word $w = w_0 \cdots w_{n-1} \in (2^{AP})^*$ is accepted by $\Aa$ if there is a run $R = q_0 \xrightarrow{w_0} q_1 \cdots q_{n-1} \xrightarrow{w_{n-1}} q_n$ such that $q_n \in F$ and 
$\Psi(|R|_{q_0}/x_{q_0}, \cdots, |R|_{q_m}/x_{q_m})$ holds, 
where 
the vector $(|R|_q)_{q \in Q}$ is the Parikh image of the sequence $q_0, \cdots, q_n$, that is,  $|R|_q$ is the number of occurrences of $q$ in $R$. 
We use $\Ll(\Aa)$ to denote the set of words accepted by $\Aa$.

\vspace{-2mm}
\begin{theorem}[\cite{SSM08}]
Nonemptiness of Presburger automata is decidable. 
\end{theorem}
\vspace{-1mm}

Given $\Phi_{\mathrm{core}}\wedge \Phi_{\mathrm{count}}$, where $\Phi_{\mathrm{core}}$ is an $\qgdbs^-_\natnum$ formula and $\Phi_{\mathrm{count}}$ is a conjunction of the formulae of the form $T_m\ \op\ 0$ which contains only variables from $\vec{x} \cup \vec{S}$, our aim is to construct a PA to accept models---as words---of $\Phi_{\mathrm{core}}\wedge \Phi_{\mathrm{count}}$. To this end, we first show how to construct an NFA from $\Phi_{\mathrm{core}}$, an $\qgdbs^-_\natnum$ formula. 

It is a simple observation that an $\qgdbs^-_\natnum$ formula can be rewritten in exponential time into a formula in 
MSOW defined by the following rules,

\smallskip
\hspace{1cm}$
\Phi ::= x + 1 = y \mid x < y \mid S(x) \mid \Phi \wedge \Phi \mid \neg \Phi \mid \forall x.\ \Phi \mid \forall S.\ \Phi,
$
\smallskip

\noindent where $x,y$ are variables ranging over $\natnum$, and $S$ is a (second-order) set variable ranging over the set of \emph{finite} subsets of $\natnum$. 
Note that the exponential blow-up is because, only the successor operator is available in MSOW 
while constants $c$ are encoded in binary. 
For instance, $x_1 \le x_2 + 2$ has to be rewritten into $\exists z, z'.\ z = x_2 + 1 \wedge z' = z +1 \wedge x_1 \le z'$.


We  can  then invoke the celebrated B\"{u}chi-Elgot theorem:

\vspace{-1mm}
\begin{theorem}[\cite{Bu60,Elg61}]\label{thm-msow-nfa}
Let $\Phi(S_1,\cdots, S_k)$ be an MSOW formula. Then an NFA $\Aa_\Phi$ over $2^{\{S_1,\cdots,S_k\}}$ can be constructed so that $\Ll(\Aa_\Phi) = \Ll(\Phi(S_1,\cdots, S_k))$. 
\end{theorem}
\vspace{-1mm}

It follows from Theorem~\ref{thm-msow-nfa} that an NFA $\Aa_\Phi=(Q, AP, \delta, q_0, F)$ can be constructed from a $\qgdbs^-_\natnum$ formula $\Phi(\vec{x}, \vec{S})$ such that $\Ll(\Phi)=\Ll(\Aa_\Phi)$.
As the next step we construct a \emph{quantifier-free Presburger arithmetic} formula $\Psi$ for the sought PA out of $\Phi_{\mathrm{count}}$. We first construct, for each $x_i$, an NFA $\Aa_i$ illustrated in Fig.\ref{fig-auto-b}(a), and for each $S_j$, an NFA $\Bb_j$ illustrated in Fig.\ref{fig-auto-b}(b). We then consider an NFA  $\Aa^\times_\Phi$ which is the product of $\Aa_{\Phi}$ and all $\Aa_i$ for $i \in [k]$ and $\Bb_j$ for $j \in [l]$. Note that each state of  $\Aa^\times_\Phi$ is a vector of states $\vec{q}=(q, q_1, \cdots, q_k, q_{k+1},\cdots, q_{k+l})$ such that $q \in Q$, $q_i  \in \{p_{0,i}, p_{1,i}\}$ for each $i \in [k]$, and $q_{k+j} \in \{q_{0, j}, q_{1, j}, q_{2, j}\}$ for each $j \in [l]$. We write $\vec{q}_r$ for the $r$-th entry of $\vec{q}$, i.e., $\vec{q}_0=q$ and $\vec{q}_i=q_i$ for each $i \in [k+l]$.


We observe that, for each $i \in [k]$, $x_i$ is expressed by $\sum_{\{\vec{q} \mid \vec{q}_{i}=p_{0,i}\}} x_{\vec{q}}-1$, and for each $j \in [l]$, $\min(S_j)$ is expressed by $\sum_{\{\vec{q} \mid \vec{q}_{k+j}=q_{0,j}\}} x_{\vec{q}}-1$ and $\max(S_j)$ is expressed by $\sum_{\{\vec{q} \mid \vec{q}_{k+j}=q_{0,j}\}} x_{\vec{q}}  + \sum_{\{\vec{q}\mid\vec{q}_{k+j}=q_{1,j}\}} x_{\vec{q}}-1$. We then substitute them into $\Phi_{\mathrm{count}}$ and obtain $\Psi$ which is over the variables $x_{\vec{q}}$. 
	%

\begin{figure}[htbp]
\vspace{-4mm}
	\begin{center}
		\includegraphics[scale=0.67]{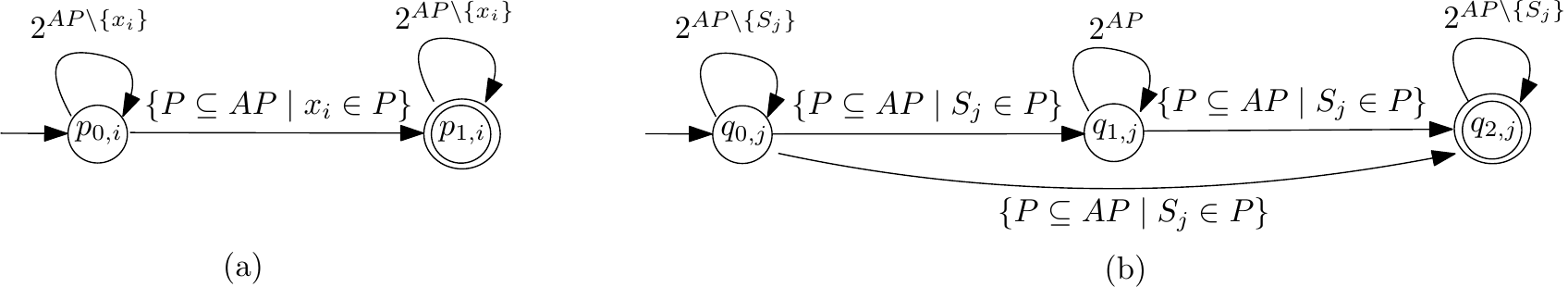}
		\vspace{-2mm}
		\caption{NFA $\Aa_i$ for $x_i$ and $\Bb_j$ for $S_j$}
		\label{fig-auto-b}
	\end{center}
\vspace{-1cm}
\end{figure}
%

\hide{
\begin{example} 
Let us illustrate the idea by considering an example that $\Phi_{\sf core}$ contains two set variables $S_1,S_2$ and $\Phi_m = (\max(S_1) - \min(S_1)) - (\max(S_2) - \min(S_2)) = 0$. 



%
%
%
%

%
Then $\min(S_i)$ is expressed by $x_{q_{0,i}} + 1$ and $\max(S_i)$ is expressed by $x_{q_{0,i}} + x_{q_{1,i}}+1$.

Therefore, the Presburger automaton $\Aa_{\Phi}$ is constructed as $(\Aa', \Psi')$, where 
\begin{itemize}
\item $\Aa' = (Q', AP, \delta', q'_0, F')$ is the product of $\Aa_{\Phi_{\sf core}}$, $\Bb_1$ and $\Bb_2$, 
\item Then $\Phi'$ is from $\Phi_m$ by replacing 
\begin{itemize}
\item $\min(S_1)$ with $\sum \limits_{(q, q_{0,1}, q_{j,2}) \in Q'} x_{(q,q_{0,1}, q_{j,2})}+ 1$, 
\item $\max(S_1)$ with  $\sum \limits_{(q, q_{0,1}, q_{j,2}) \in Q'} x_{(q,q_{0,1}, q_{j,2})}+ \sum \limits_{(q, q_{1,1}, q_{j,2}) \in Q'} x_{(q,q_{1,1}, q_{j,2})}+ 1$, 
\item $\min(S_2)$ with $\sum \limits_{(q, q_{j,1}, q_{0,2}) \in Q'} x_{(q, q_{j,1}, q_{0,2})}+ 1$, 
\item $\max(S_2)$ with $\sum \limits_{(q, q_{j,1}, q_{0,2}) \in Q'} x_{(q, q_{j,1}, q_{0,2})} + \sum \limits_{(q, q_{j,1}, q_{1,2}) \in Q'} x_{(q, q_{j,1}, q_{1,2})}+1$.
\end{itemize}
Therefore, we have
\[
\Psi' = \sum \limits_{(q, q_{1,1}, q_{j,2}) \in Q'} x_{(q,q_{1,1}, q_{j,2})} - \sum \limits_{(q, q_{j,1}, q_{1,2}) \in Q'} x_{(q, q_{j,1}, q_{1,2})} = 0.
\]
\end{itemize}

It is easy to figure out the construction for the general case from the illustration above.
\end{example}
}

\vspace{-1mm}
\begin{proposition}
For an $\qgdbs_\natnum$ formula $\Phi = \Phi_{\mathrm{core}} \wedge \Phi_{\mathrm{count}}$, $\Phi_{\mathrm{core}}$ is an $\qgdbs^-_\natnum$ formula, and $\Phi_{\mathrm{count}}$ is a conjunction of formulae of the form $T_m\ \op\ 0$ which contain only variables from $\vec{x} \cup \vec{S}$, a PA $\Aa_\Phi=(\Aa^\times_{\Phi}, \Psi)$ can be constructed effectively such that $\Ll(\Aa_\Phi) = \Ll(\Phi)$.
\end{proposition}
\vspace{-1mm}
}


%
%
%
%
%
%




\section{Conclusion}

In this paper, we have defined $\lcslidset$, SL with linearly compositional inductive predicates and set data constraints. The main feature is to identify $\dbs$ as a special class of set data constraints in the inductive definitions. We 
encoded the transitive closure of $\dbs$ into  $\qgdbs$, which was shown to be decidable. These together yield a complete decision procedure for the satisfiability of $\lcslidset$. 

The precise complexity of the decision procedure---\textsc{Nonelementary} is the best upper-bound we have now---is left open for further studies. 
Furthermore, the entailment problem of $\lcslidset$ is immediate future work.

\newpage

 


\iftoggle{fullver}{
\newpage

\appendix


\section{
	Semantics of $\lcslidset[P]$} \label{app:slsemantics}
 
Each formula in $\lcslidset[P]$ is interpreted on the states. Formally, a \emph{state} is a pair $(s,h)$, where
\vspace{-2mm} 
\begin{itemize}
	\item $s$ is an assignment function which is a partial function from $\lvars \cup \dvars \cup \svars$ to $\loc \cup \intnum \cup \intset$ such that $dom(s)$ is finite and $s$ respects the data type,
	\item $h$ is a \emph{heap} which is a partial function from $\loc \times (\Ff \cup \Dd)$ to $\loc \cup \data$ such that
	\begin{itemize} 
		\item $h$ respects the data type of fields, that is, for each $l \in \loc$ and $f \in \Ff$ (resp. $l \in \loc$ and $d \in \Dd$), if $h(l,f)$ (resp. $h(l,d)$) is defined, then $h(l,f) \in \loc$ (resp. $h(l,d) \in \intnum$); and 
		\item $h$ is field-consistent, i.e. every location in $h$ possesses the same set of fields. 
	\end{itemize}
\end{itemize}

For a heap $h$, we use $\ldom(h)$ to denote the set of locations $l \in \loc$ such that $h(l,f)$ or $h(l,d)$ is defined for some $f \in \Ff$ and $d \in \Dd$. Moreover, we use $\flds(h)$ to denote the set of fields $f \in \Ff$ or $d \in \Dd$ such that $h(l,f)$ or $h(l,d)$ is defined for some $l \in \loc$. Two heaps $h_1$ and $h_2$ are said to be \emph{field-compatible} if $\flds(h_1)=\flds(h_2)$.  We write $h_1 \# h_2$ if $\ldom(h_1) \cap \ldom(h_2)=\emptyset$ and $\flds(h_1)=\flds(h_2)$. Moreover, we write $h_1 \uplus h_2$ for the disjoint union of two field-compatible heaps $h_1$ and $h_2$, which implies $h_1\# h_2$.
%

Let $(s,h)$ be a state and $\phi$ be an $\lcslidset[P]$ formula. Then the semantics of $\lcslidset[P]$ formulae is defined as follows,

\begin{itemize}
	\item $(s,h) \vDash E = F$ (resp. $(s,h) \vDash E \neq F$) if $s(E)=s(F)$ (resp. $s(E) \neq s(F)$),
	%
	%
	\item $(s,h) \vDash \Pi_1 \wedge \Pi_2$ if $(s,h) \vDash \Pi_1$ and $(s,h) \vDash \Pi_2$,
	%
	%
	%
	%
	\item $(s,h) \vDash \Delta$ if $s \vDash \Delta$ (see semantics of $\qgdbs$ in Section~\ref{sec:logic}),
	%
	\item $(s,h) \vDash \slemp$ if $\ldom(h)=\emptyset$,
	\item $(s,h) \vDash E \mapsto (\rho)$ if $\ldom(h)=s(E)$, and for each $(f,X) \in \rho$ (resp. $(d,x) \in \rho$), $h(s(E),f)=s(X)$ (resp. $h(s(E),d)=s(x)$),
	\item $(s,h) \vDash P(E,\vec{\alpha}; F, \vec{\beta}; \vec{\xi})$ if $(s,h) \in \ldbrack P(E, \vec{\alpha}; F, \vec{\beta}; \vec{\xi}) \rdbrack$, 
	\item $(s,h) \vDash \Sigma_1 \sep \Sigma_2$ if there are $h_1,h_2$ such that $h= h_1 \uplus h_2$, $(s,h_1) \vDash \Sigma_1$ and $(s,h_2) \vDash \Sigma_2$.
\end{itemize}
where the semantics of predicates $\ldbrack P(E, \vec{\alpha}; F, \vec{\beta}; \vec{\xi}) \rdbrack$ is given by the least fixed point of a monotone operator constructed from the body of rules for $P$ in a standard way as in \cite{BFP+14}.  


\section{
Construction of $\boolabs(\phi)$ for $\phi = \Pi \wedge \Delta \wedge \Sigma$
}\label{app-sec-sat}

For each spatial atom $a_i$ rooted at $Z$, $\boolabs(\phi)$ introduces a Boolean variable $[Z, i]$ to denote whether $a_i$ corresponds to a nonempty heap or not. Moreover, for each predicate atom $a_i=P(Z_1,\vec{\mu}; Z_2, \vec{\nu}; \vec{\chi})$ in $\Sigma$ such that in the inductive rule of $P$, $E$ occurs in $\vec{\gamma}$, we introduce a Boolean variable $[\nu_{\idx_{(P, \vec{\gamma}, E)}}, i]$. Let $\bvars(\phi)$ denote the set of introduced Boolean variables. 
\emph{The abstraction of $\phi$} is defined as $\boolabs(\phi)::= \Pi \wedge \Delta \wedge \phi_\Sigma \wedge \phi_\ast$ over $\bvars(\phi)\cup \vars(\phi)$, where $\phi_\Sigma$ and $\phi_\ast$ are defined as follows.
\vspace{-1mm}
\begin{itemize}
%
\item $\phi_\Sigma = \bigwedge \limits_{1 \le i \le n} \boolabs(a_i)$ is an abstraction of $\Sigma$ where
\begin{itemize}
\item if $a_i = E \mapsto \rho$, then $\boolabs(a_i)=[E,i]$, 

\item if  $a_i = P(Z_1,\vec{\mu}; Z_2, \vec{\nu}; \vec{\chi})$ and in the body of the inductive rule of $P$, $E$ occurs in $\vec{\gamma}$, then
\[
\begin{array}{l c l}
\boolabs(a_i)  & = & (Z_1 = Z_2 \wedge \vec{\mu} = \vec{\nu}) \vee  \\
& & ([Z_1,i] \wedge [\nu_{\idx_{(P,\vec{\gamma}, E)}}, i] \wedge \ufld_1(P(Z_1,\vec{\mu}; Z_2, \vec{\nu}; \vec{\chi}))) \\
 & & \vee\ ([Z_1,i] \wedge [\nu_{\idx_{(P,\vec{\gamma}, E)}}, i]  \wedge \ufld_{\ge 2}(P(Z_1,\vec{\mu}; Z_2, \vec{\nu}; \vec{\chi}))),
\end{array}
\]

\item if  $a_i = P(Z_1,\vec{\mu}; Z_2, \vec{\nu}; \vec{\chi})$ and  in the body of the inductive rule of $P$, $E$ does not occur in $\vec{\gamma}$, then
\[
\begin{array}{l c l}
\boolabs(a_i)  & = & (Z_1 = Z_2 \wedge \vec{\mu} = \vec{\nu}) \vee  \\
& & ([Z_1,i] \wedge \ufld_1(P(Z_1,\vec{\mu}; Z_2, \vec{\nu}; \vec{\chi}))) \\
 & & \vee\ ([Z_1,i] \wedge \ufld_{\ge 2}(P(Z_1,\vec{\mu}; Z_2, \vec{\nu}; \vec{\chi}))),
\end{array}
\]
\end{itemize}
%
%
\item $\phi_\ast$ encodes the semantics of separating conjunction, 

\[
\phi_\ast   = \bigwedge \limits_{[Z_1,i],[Z'_1,j] \in \bvars(\phi), i \neq j}(Z_1=Z'_1 \wedge [Z_1,i]) \rightarrow \neg [Z'_1,j].
\]
\end{itemize}

%
%
%

\begin{proposition} \label{prop:eqisat}
For each $\lcslidset[P]$ formula $\phi$, $\phi$ is satisfiable iff $\boolabs(\phi)$ is satisfiable. 
\end{proposition}

\begin{remark}
From the construction of $\boolabs(\phi)$, one can observe that  the static parameters of the predicate atoms are irrelevant to the satisfiability of $\phi$.
\end{remark}

\section{Details of Section~\ref{sec:tc}}\label{app:tc}

\subsection{Proof of Proposition~\ref{prop-saturate}}

\noindent {\bf Proposition~\ref{prop-saturate}}.
\emph{Let $\varphi_R(S,S') := \varphi_{R,1} \wedge \varphi_{R,2}$ be a $\dbs$ formula such that $\varphi_{R,1} := S = S' \cup T_s$ and $\varphi_{R,2}$ is satisfiable. Then $\varphi_R$ can be transformed, in polynomial time, to a saturated formula  $\saturate(\varphi_R(S,S'))$ or to a formula where the integer subformula is unsatisfiable.}

\begin{proof}
    Firstly, $\varphi_R := \varphi_{R,1} \wedge \varphi_{R,2}$ can be transformed, in polynomial time, to a saturated formula $\saturate(\varphi_R(S,S')) := \saturate(\varphi_{R,1}) \wedge \saturate(\varphi_{R,2})$, satisfying the conditions of Definition 4.
    \begin{enumerate}
    \item $\saturate(\varphi_R(S,S'))$ satisfy the condition 1.
    
     $\varphi_{R,2}$ is satisfiable. Obviously, $\varphi_{R,2}(S,S')$ can be transformed to $\norm(\varphi_{R,2}(S,S'))$. The Definition 4 (1) is ok. So, we can add $\norm(\varphi_{R,2}(S,S'))$ into $\saturate(\varphi_{R,2})$.
    
    \item $\saturate(\varphi_R(S,S'))$ satisfy the condition 2.
    
    We focus on the case $\varphi_{R,1} := S = S' \cup T_s$. The symmetrical case $\varphi_{R,1} := S' = S \cup T_s$ can be adapted easily.
    According to the $\dbs$ syntax, $S' \in T_S$ is possible. Owing to $\varphi_{R,1} := S = S' \cup T_s$ and $min(S'), max(S') \in S'$. For example, $\varphi_{R,1} := S = S' \cup \{min(S), min(S')\}$ is equivalent to $\varphi_{R,1} := S = S' \cup \{min(S)\}$.
    Hence, $\varphi_{R,1}$ can be transformed to $\saturate(\varphi_{R,1}) := S \cup T_s, \quad T_s  \in \{\emptyset, \{min(S)\}, \{max(S)\}, \{min(S), max(S)\}\}$. The Definition 4 (2) is ok.
    
    \item $\saturate(\varphi_R(S,S'))$ satisfy the condition 3.
    
    If $S$ is surely nonempty in $\varphi_{R}$, we can add $min(S) \le max(S)$ into $\saturate(\varphi_{R,2})$. Obviously, $\saturate(\varphi_{R,2})$ contains a conjuct $min(S) \le max(S) - c \text{ for some } c \ge 0$.
    
    If $S'$ is surely nonempty in $\varphi_{R}$, we can add $min(S') \le max(S')$ into $\saturate(\varphi_{R,2})$. Obviously, $\saturate(\varphi_{R,2})$ contains a conjuct $min(S') \le max(S') - c' \text{ for some } c' \ge 0$.
    Hence, the Definition 4 (3) is ok.
    \item $\saturate(\varphi_R(S,S'))$ satisfy the condition 4.
    
    If $S$ and $S'$ are surely nonempty in $\varphi_{R}$, we can add $min(S) \le min(S')$ and $max(S') \le max(S)$ into $\saturate(\varphi_{R,2})$ owing to $S' \subseteq S$.
    If $min(S) \notin T_s $, we can get $min(S)=min(S')$. Hence, we can add $min(S) \le min(S')$ and $min(S') \le min(S)$ into $\saturate(\varphi_{R,2})$. If  $\saturate(\varphi_{R,2})$ contains $min(S) \le min(S')$ and $min(S') \le min(S)$, we can get $min(S)=min(S')$. Hence, we get $min(S) \notin T_s $.
    If $max(S) \notin T_s $, we can get $max(S)=max(S')$. Hence, we can add $max(S) \le max(S')$ and $max(S') \le max(S)$ into $\saturate(\varphi_{R,2})$. If  $\saturate(\varphi_{R,2})$ contains $max(S) \le max(S')$ and $max(S') \le max(S)$, we can get $max(S)=max(S')$. Hence, we get $max(S) \notin T_s $.
    If $\saturate(\varphi_{R,2})$ contains the conjucts $min(S) \le max(S)$ and $max(S) \le min(S)$, we can get $min(S) = max(S)$. Hence, we can get $\saturate(\varphi_{R,1}) := S \cup T_s$, in which $max(S) \notin T_s$. 
    Hence, the Definition 4 (4) is ok.
    
    Finally, $\varphi_{R,2}$ is satisfiable, but $\saturate(\varphi_{R,2})$ may be unsatisfiable, so the integer subformula in $\saturate(\varphi_R(S,S'))$ may be unsatisfiable.
    \end{enumerate}
\end{proof}

\subsection{An example of Saturation}\label{app:sat}

\begin{example}\label{exmp-sat}
	Let $\varphi_R(S,S') = \varphi_{R,1} \wedge \varphi_{R,2}$, where $\varphi_{R,1} := S = S' \cup \{\min(S)\}$ and $\varphi_{R,2}:= \min(S') = \min(S)+1$. Then $\saturate(\varphi_R(S,S'))$ is constructed 
	as follows: 
	\vspace{-1mm}
	\begin{enumerate}
		\item Since both $S$ and $S'$ are surely nonempty in $\varphi_R$, according to  Definition~\ref{def-saturate}(3),  add the conjuncts $\min(S) \le \max(S)$, $\min(S') \le \max(S')$ into $\varphi_{R,2}$.
		\item Because both $S$ and $S'$ are surely nonempty in $\varphi_R$ and $\max(S) \not \in T_s$, according to Definition~\ref{def-saturate}(4), add the conjuncts $\min(S) \le \min(S')$, $\max(S') \le \max(S)$, and $\max(S) \le \max(S')$ into $\varphi_{R,2}$. Then $\varphi_{R,2}$ becomes 
		$$
		\begin{array}{l c l}
		\varphi'_{R,2} & = & \min(S') = \min(S)+1 \wedge \min(S) \le \max(S) \wedge \min(S') \le \max(S')\ \wedge \\
		& & \min(S) \le \min(S') \wedge \max(S') \le \max(S) \wedge \max(S) \le \max(S').
		\end{array}
		$$
		See Figure~\ref{fig-varphi-r-2}(a) for $\Gg(\varphi'_{R,2})$, where an edge from $\min(S')$ to $\min(S)$  with weight $+1$ and an edge from  $\min(S)$ to $\min(S')$ with weight $-1$ are from $\min(S') = \min(S)+1$. 
		\item Turn $\varphi'_{R,2}$ into the normal form (see Figure~\ref{fig-varphi-r-2}(b)).
	\end{enumerate}
	\vspace{-1mm}
	Therefore, $\saturate(\varphi_R(S,S'))$ is a conjunction of $S = S' \cup \{\min(S)\}$ and the integer subformula illustrated in Figure~\ref{fig-varphi-r-2}(b).
\end{example}
\vspace{-1mm}

\begin{figure}[htbp]
	\begin{center}
		\includegraphics[scale=0.7]{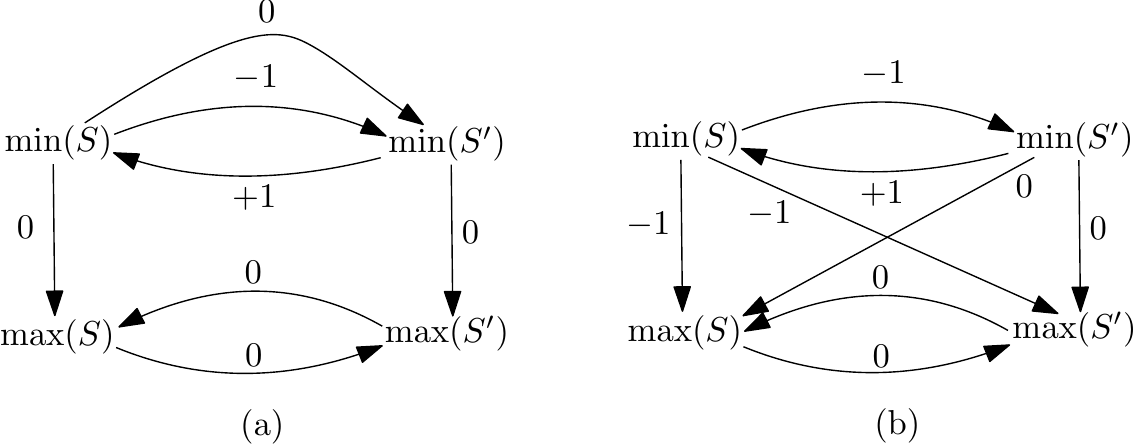}
		\vspace{-2mm}
		\caption{$\varphi'_{R,2}$ and its normal form.}
		\label{fig-varphi-r-2}
	\end{center}
\end{figure}

\subsection{The situation that there are only one source and destination set parameter}\label{app:tcsimple}


\noindent {\bf \underline{Case I: $\varphi_{R,1} := S = S'$}}. Then $\TC[\varphi_R](S, S') := S = S'$.

\medskip


\noindent {\bf \underline{Case II: $\varphi_{R,1} := S = S' \cup \{\min(S)\}$}}. This case is illustrated schematically as $\underbrace{|-\overbrace{|------|}^{S'}}_S$. 
We observe that $S$ is surely nonempty in $\varphi_R$.  We now distinguish the subcases according to whether $S'$ is possibly empty or surely nonempty in $\varphi_R$. 

\medskip

\noindent{\bf Subcase II(i): $S'$ is possibly empty in $\varphi_R$.} 
In this case, neither $\min(S')$ nor $\max(S')$ occurs in $\varphi_{R,2}$. Therefore, $\varphi_{R,2}$ is a formula involving  $\min(S)$ or $\max(S)$ only. Moreover, if $S'$ is nonempty, then $\max(S) = \max(S')$. 

Evidently, $\TC[\varphi_R](S, S')$ can be specified by (an infinite disjunction) $(S= S') \vee \bigvee \limits_{n \ge 1} \varphi^{(n)}_R$, where $\varphi^{(n)}_R$ is obtained by unfolding $\varphi_R$ for $n$ times, that is, 
\[
\exists S_1,\cdots, S_{n+1}.\ S_1 = S \wedge S_{n+1} = S' \wedge \bigwedge \limits_{i \in [n]} (S_{i} = S_{i+1} \cup \{\min(S_{i})\} \wedge \varphi_{R,2}[S_i/S]), 
\]
where $\varphi_{R,2}[S_i/S]$ is obtained from $\varphi_{R,2}$ by replacing $S$ with $S_i$.
We use the following figure to help the reader understand $\varphi^{(n)}_R$.
$$\underbrace{-\overbrace{---\overbrace{-\overbrace{------}^{S_{n+1} = S'}}^{\cdots}}^{S_2}}_{S=S_1}$$

Evidently, $\varphi^{(1)}_R = \varphi_R$, and
\[
\varphi^{(2)}_R = \exists S_{2}.\  (S = S_{2} \cup \{\min(S)\} \wedge S_2 = S' \cup \{\min(S_2)\} \wedge \varphi_{R,2} \wedge \varphi_{R,2}[S_2/S]).
\]
Let us consider $\varphi^{(n)}_R$ for $n \ge 3$ in the following.

Because $S_{i} = S_{i+1} \cup \{\min(S_{i}\}$ for each $i \in [n]$, we have $\max(S_1) = \cdots = \max(S_n)$ and $\min(S_1) \le \cdots \le \min(S_n)$. Then from the fact that $\varphi_{R,2}$ is a conjunction of difference-bound constraints involving $\min(S)$ and $\max(S)$ only, we deduce that $\bigwedge \limits_{i \in [n]} \varphi_{R,2}[S_i/S]$ is equivalent to $\varphi_{R,2}[S_1/S] \wedge \varphi_{R,2}[S_n/S]$. For instance, if $\varphi_{R,2} \equiv c \le \max(S) - \min(S) \le c'$ for some constants $c, c' \ge 0$ with $c \le c'$, then $\max(S_1) - \min(S_1) \le c'$ implies $\max(S_i) - \min(S_i) \le c'$ for each $i \in [n]$, and $c \le \max(S_n) - \min(S_n)$ implies $c \le \max(S_i) - \min(S_i)$ for each $i \in [n]$. Therefore, in this situation,  $\varphi_{R,2}[S_1/S] \wedge \varphi_{R,2}[S_n/S] \equiv c \le \max(S_1) - \min(S_1) \le c' \wedge c \le \max(S_n) - \min(S_n) \le c'$ implies  $\bigwedge \limits_{i \in [n]} \varphi_{R,2}[S_i/S]$, thus they are equivalent. We then have that $\bigwedge \limits_{i \in [n]} (S_{i} = S_{i+1} \cup \{\min(S_{i})\} \wedge \varphi_{R,2}[S_i/S])$ is equivalent to 
$$\varphi_{R,2}[S_1/S] \wedge \varphi_{R,2}[S_n/S] \wedge \bigwedge \limits_{i \in [n]} S_{i} = S_{i+1} \cup \{\min(S_{i})\}.$$

Thus $\varphi^{(n)}_R$ (where $n \ge 3$) can be rewritten into 
$$
\exists S_n.\left(
\begin{array}{l}
S_n = S' \cup \{\min(S_n)\} \wedge \varphi_{R,2} \wedge \varphi_{R,2}[S_n/S]\ \wedge\\
\exists S_2, \cdots, S_{n-1}.\left(S = S_{2} \cup \{\min(S)\}  \wedge \bigwedge \limits_{2 \le i \le n-1} S_{i} = S_{i+1} \cup \{\min(S_{i})\}\right)
\end{array}
\right).
$$

\medskip

\noindent {\bf Claim}. The formula $\theta(S, S_n):=$
$$\exists S_2, \cdots, S_{n-1}.\left(S = S_{2} \cup \{\min(S)\}  \wedge \bigwedge \limits_{2 \le i \le n-1} S_{i} = S_{i+1} \cup \{\min(S_{i})\}\right)$$ 
is equivalent to $\theta'(S, S_n):=$
$$
S \neq \emptyset \wedge S_n \subseteq S \wedge |S \setminus S_n|  \le n-1 \wedge  
\left(
\begin{array}{l}
(S \setminus S_n \neq \emptyset \wedge S_n \neq \emptyset) \rightarrow \\
\hspace{2cm} \max(S \setminus S_n) < \min(S_n)
\end{array}
\right).
$$ 

\smallskip

\begin{proof}[of the claim]
%

\begin{itemize}
	\item $\theta(S, S_n)$ implies $\theta'(S, S_n)$. 
Suppose that $A, A_n$ are two finite subsets of $\intnum$ such that $\theta(A, A_n)$ holds.
Then there are finite subsets  $A_2, \cdots, A_{n-1}  \subseteq \intnum$ such that $A = A_{2} \cup \{\min(A)\}  \wedge \bigwedge \limits_{2 \le i \le n-1} A_{i} = A_{i+1} \cup \{\min(A_{i})\}$ holds.
Therefore, 
$$A = A_n \cup \{\min(A), \min(A_2), \cdots, \min(A_{n-1})\}.$$ 
From this, we deduce that $A \neq \emptyset$ and $A_n \subseteq A$. Moreover, from 
$$A \setminus A_n \subseteq    \{\min(A), \min(A_2), \cdots, \min(A_{n-1})\},$$ 
we have $|A \setminus A_n| \le n-1$. 

Now suppose $A_n \neq \emptyset$ and $A \setminus A_n \neq \emptyset$. Then 
$$\max(A \setminus A_n) \in \{ \min(A), \min(A_2), \cdots, \min(A_{n-1})\}.$$ 
From $A_n \subseteq A$ and $A_n \subseteq A_i$ for each $2 \le i \le n-1$, we have  $\min(A_n) \ge \min(A)$ and $\min(A_n) \ge \min(A_i)$ for each $2 \le i \le n-1$. In other words, $\min(A_n)$ is an upper bound of $\{\min(A), \min(A_2), \cdots, \min(A_{n-1})\}$. Consequently, $\max(A \setminus A_n) \le \min(A_n)$. Since $\max(A \setminus A_n) \neq \min(A_n)$, we have  $\max(A \setminus A_n) < \min(A_n)$. We conclude that $\theta'(A, A_n)$ holds.

\item $\theta'(S, S_n)$ implies $\theta(S, S_n)$. Suppose that $A, A_n$ are two finite subsets of $\intnum$ such that $\theta'(A, A_n)$ holds. Then $A \neq \emptyset$, $A_n \subseteq A$, $|A \setminus A_n| \le n-1$. Moreover, if $A \setminus A_n \neq \emptyset$ and $A_n \neq \emptyset$, then $\max(A \setminus A_n) < \min(A_n)$.

If $A \setminus A_n = \emptyset$, then define $A_2, \cdots, A_{n-1}$ as $A_n$. Since $A = A_2 = \cdots A_{n-1} = A_n$, we deduce that $A= A_2 \cup \{\min(A)\}$, and $A_i = A_{i+1} \cup \{\min(A_i)\}$ for each $i: 2 \le i \le n-1$. Therefore $\theta(A, A_n)$ holds.

We now assume $A \setminus A_n \neq \emptyset$.
From $|A \setminus A_n| \le n-1$, we know that there are $r \in [n-1]$ and $i_1,\cdots, i_r \in \intnum$ such that $i_1 < \cdots < i_r$ and $A \setminus A_n =  \{i_1, \cdots, i_r\}$. Moreover, if $A_n \neq \emptyset$, then $i_r = \max(A \setminus A_n) < \min(A_n)$. We then define $A_2, \cdots, A_{n-1}$ as follows: 
\begin{itemize}
\item For each $j \in [r-1]$, define $A_{j+1}$ as $A_n \cup \{i_{j+1},\cdots, i_r\}$. 
\item For each $j: r+1 \le j \le n-1$, define $A_j$ as $A_{r}$.
\end{itemize}
From $A \setminus A_n = \{i_1,\cdots, i_r\}$ and $i_r = \max(A \setminus A_n) < \min(A_n)$ if $A_n \neq \emptyset$, we deduce that $\min(A) = i_1$. Moreover, for each $j \in [r-1]$, $\min(A_{j+1})= i_{j+1}$, and for each $j: r+1 \le j \le n-1$, $\min(A_j) = \min(A_r) = i_r$.
Therefore, 
\begin{itemize}
\item $A = A_n \cup \{i_1, \cdots, i_r\} = (A_n \cup \{i_2,\cdots, i_r\}) \cup \{i_1\} = A_2 \cup \{i_1\} = A_2 \cup \{\min(A)\}$, 
\item for each $j \in [r-1]$, $A_{j+1} = A_n \cup \{i_{j+1},\cdots, i_r\} = (A_n \cup \{i_{j+2}, \cdots, i_r\}) \cup \{i_{j+1}\} = A_{j+2} \cup \{i_{j+1}\} = A_{j+2} \cup \{\min(A_{j+1})\}$,
\item for each $j: r+1 \le j \le n-1$, $A_j = A_r = A_{j+1} \cup \{\min(A_j)\}$.
\end{itemize}
We conclude that $\theta(A, A_n)$ holds.
\end{itemize}
\qed
\end{proof}

According to the claim, $\varphi^{(n)}_R$ for $n \ge 3$ can be simplified into 
$$
\exists S''.\left(
\begin{array}{l}
S'' = S' \cup \{\min(S'')\} \wedge \varphi_{R,2} \wedge (\varphi_{R,2}[S''/S]) \wedge S \neq \emptyset \wedge S'' \subseteq S\ \wedge \\
|S \setminus S''|  \le n-1 \wedge  ((S \setminus S'' \neq \emptyset \wedge S'' \neq \emptyset) \rightarrow \max(S \setminus S'') < \min(S''))
\end{array}
\right). \hfill (*)
$$
Moreover, it is not hard to observe that the formula $(*)$ above is equivalent to $\varphi^{(n)}_R$ even for $n=1,2$. 
Since $\TC[\varphi_R](S, S')$ is equal to $(S= S') \vee \bigvee \limits_{n \ge 1} \varphi^{(n)}_R$, we conclude that
$$
\begin{array}{l}
\TC[\varphi_R](S, S') = (S = S')\ \vee  \\
\ \exists S''. \left(
\begin{array}{l}
S'' = S' \cup \{\min(S'')\} \wedge \varphi_{R,2} \wedge (\varphi_{R,2}[S''/S])\ \wedge\\
S \neq \emptyset \wedge S'' \subseteq S \wedge ((S\setminus S'' \neq \emptyset \wedge S'' \neq \emptyset) \rightarrow \max(S \setminus S'') < \min(S''))
\end{array}
\right).
\end{array}
$$ 


\medskip

\noindent{\bf Subcase II(ii): $S'$ is surely nonempty in $\varphi_R$.} We distinguish between whether $\lfloor \varphi_{R,2} \rfloor_{\min(S),\min(S')}$ is strict or not.

\paragraph{Case that $\lfloor \varphi_{R,2} \rfloor_{\min(S),\min(S')}$ is strict.}
The arguments for Subcase II(ii) in this situation have already been  presented in the main text, but with the proof of the claim missing. In the following, we will  present this proof.

\medskip

\noindent {\bf Claim}.
{\it
Suppose $n \ge 3$ and $\lfloor \varphi_{R,2} \rfloor_{\min(S),\min(S')}$ is strict. Then  

{\small
$\exists S_3,\cdots, S_{n-1}.
\bigwedge \limits_{2 \le i \le n-1} 
\left(
\begin{array}{l}
S_i = S_{i+1} \cup \{\min(S_i)\} \wedge \max(S_i) = \max(S_{i+1})\ \wedge \\
(\lfloor \varphi_{R,2} \rfloor_{\min(S),\min(S')}[S_i/S,S_{i+1}/S'])
\end{array}
\right)
$		
}

\noindent is equivalent to 

\smallskip
\noindent $
\small
\begin{array}{l}
S_n \neq \emptyset \wedge S_2 \setminus S_n \neq \emptyset  \wedge S_n \subseteq S_2 \wedge |S_2 \setminus S_n| = n-2 \wedge \max(S_2\setminus S_n) < \min(S_n)\ \wedge \\
\forall y, z.\ \suc((S_2 \setminus S_n)  \cup \{\min(S_n)\}, y, z) \rightarrow (\lfloor \varphi_{R,2} \rfloor_{\min(S), \min(S')}[y/\min(S), z/\min(S')]).
\end{array}
$
}

\begin{proof}[of the claim]
Let $\theta(S_2, S_n)$ and $\theta'(S_2, S_n)$ denote the two formulae in the claim. Our goal is to show the equivalence of $\theta(S_2, S_n)$ and $\theta'(S_2, S_n)$.

From the fact that $\lfloor \varphi_{R,2} \rfloor_{\min(S),\min(S')}$ is strict, we know that 
$$\lfloor \varphi_{R,2} \rfloor_{\min(S),\min(S')}  \equiv c  \le \min(S') - \min(S) \le c'$$ 
for some $c, c': 0 < c \le c'$.

\begin{itemize}
	\item   $\theta(S_2, S_n)$ implies $\theta'(S_2, S_n)$. Suppose  $A_2, A_n$ are finite subsets of $\intnum$ such that $\theta(A_2, A_n)$ holds. Then there are \emph{nonempty} finite subsets $A_3, \cdots, A_{n-1} \subseteq \intnum$ such that $A_i = A_{i+1} \cup \{\min(A_i)\}$, $\max(A_i) = \max(A_{i+1})$, and for each $i: 2 \le i \le n-1$, $(\lfloor \varphi_{R,2} \rfloor_{\min(S),\min(S')}[S_i/S,S_{i+1}/S'])(A_i, A_{i+1})$. 

From the fact that $(\lfloor \varphi_{R,2} \rfloor_{\min(S),\min(S')}[S_i/S,S_{i+1}/S'])(A_{n-1}, A_{n})$ holds, we know that $0 < c \le \min(A_{n}) - \min(A_{n-1}) \le c'$. Therefore, $A_n \neq \emptyset$. Moreover, it is easy to observe that $A_n \subseteq A_2$. 

From $A_i = A_{i+1} \cup \{\min(A_i)\}$ for each $i: 2 \le i \le n-1$, we deduce that $\min(A_2) \le \cdots \le \min(A_n)$. In addition, from the fact that $\lfloor \varphi_{R,2} \rfloor_{\min(S),\min(S')}$ is strict and $(\lfloor \varphi_{R,2} \rfloor_{\min(S),\min(S')}[S_i/S,S_{i+1}/S'])(A_i, A_{i+1})$ holds for each $i: 2 \le i \le n-1$, we have $\min(A_2) < \cdots < \min(A_n)$. We then deduce that $A_2 \setminus A_n   \neq \emptyset$ and $|A_2 \setminus A_n|= |\{\min(A_2), \cdots, \min(A_{n-1})\} | = n-2$.

Because  $\max(A_2 \setminus A_n) \in \{\min(A_2), \cdots, \min(A_{n-1})\}$ and $\min(A_n)$ is an upper bound of $\{\min(A_2), \cdots, \min(A_{n-1})\}$, we have $\max(A_2 \setminus A_n) < \min(A_n)$. 

Finally, from the fact that $(\lfloor \varphi_{R,2} \rfloor_{\min(S),\min(S')}[S_i/S,S_{i+1}/S'])(A_i, A_{i+1})$ holds for each $i: 2 \le i \le n-1$, and $A_2 \setminus A_n = \{\min(A_2), \cdots, \min(A_{n-1})\}$, we deduce that for each pair of distinct numbers $i_1, i_2 \in A_2 \setminus A_n$ such that $i_1 < i_2$ and $A_2 \setminus A_n$ contains no other numbers (strictly) between $i_1 $ and $i_2$, 
$$(\lfloor \varphi_{R,2} \rfloor_{\min(S), \min(S')}[y/\min(S), z/\min(S')])(i_1, i_2)$$ holds.
Therefore, the pair $(A_2, A_n)$ satisfies the formula
{
\small
$$\forall y, z.\ \suc((S_2 \setminus S_n)  \cup \{\min(S_n)\}, y, z) \rightarrow (\lfloor \varphi_{R,2} \rfloor_{\min(S), \min(S')}[y/\min(S), z/\min(S')]).$$
}
We conclude that $\theta'(A_2, A_n)$ holds.

\item $\theta'(S_2, S_n)$ implies $\theta(S_2, S_n)$. Suppose that $A_2, A_n$ are finite subsets of $\intnum$ such that $\theta'(A_2, A_n)$ holds.  Then $A_n \neq \emptyset$, $A_2 \setminus A_n \neq \emptyset$, $A_n \subseteq A_2$, $|A_2 \setminus A_n| = n-2$, and $\max(A_2\setminus A_n) < \min(A_n)$. 

Suppose $A_2 \setminus A_n = \{m_1, \cdots, m_{n-2}\}$  with $m_1 < \cdots < m_{n-2}$. Moreover, for convenience, we use $m_{n-1}$ to denote $\min(A_n)$.
Then for each $i \in [n-2]$, we have 
$$(\lfloor \varphi_{R,2} \rfloor_{\min(S), \min(S')}[y/\min(S), z/\min(S')])(m_i, m_{i+1})$$
holds, that is, $c  \le m_{i+1} - m_i \le c'$.

Define $A_3, \cdots, A_{n-1}$ as follows: For each $i: 3 \le i \le n-1$, define $A_{i}$ as $A_n \cup \{m_{i-1}, \cdots, m_{n-2}\}$.
Then for each $i: 2 \le i \le n-1$, $\min(A_{i}) = m_{i-1}$. Therefore, for each $i: 2 \le i \le n-1$, $A_i = A_{i+1} \cup \{\min(A_i)\}$, $\max(A_i) = \max(A_{i+1})$,  and 
$$
\begin{array} {l c l}
(\lfloor \varphi_{R,2} \rfloor_{\min(S),\min(S')}[S_i/S,S_{i+1}/S'])(A_i, A_{i+1}) & \equiv  & c  \le \min(A_{i+1}) - \min(A_i) \le c'\\
& \equiv & c  \le m_{i+1} - m_i \le c'
\end{array}
$$ 
holds. 
We conclude that $\theta(A_2, A_n)$ holds.
\end{itemize}
\qed
\end{proof}


\paragraph{Case that $\lfloor \varphi_{R,2} \rfloor_{\min(S),\min(S')}$ is non-strict.} 

The arguments are similar to the situation that $\lfloor \varphi_{R,2} \rfloor_{\min(S),\min(S')}$ is strict, but with the following adaptation: The claim is adapted into the following one and the construction of $\TC[\varphi_R](S, S')$ is adapted accordingly. 

\medskip

\noindent {\bf Claim'}.
{\it
Suppose $n \ge 3$ and $\lfloor \varphi_{R,2} \rfloor_{\min(S),\min(S')}$ is non-strict. Then  

{\small
$\exists S_3,\cdots, S_{n-1}.
\bigwedge \limits_{2 \le i \le n-1} 
\left(
\begin{array}{l}
S_i = S_{i+1} \cup \{\min(S_i)\} \wedge \max(S_i) = \max(S_{i+1})\ \wedge \\
(\lfloor \varphi_{R,2} \rfloor_{\min(S),\min(S')}[S_i/S,S_{i+1}/S'])
\end{array}
\right)
$		
}

\noindent is equivalent to 

\smallskip
\noindent$
\small
\begin{array}{l}
S_n \neq \emptyset \wedge S_n \subseteq S_2 \wedge |S_2 \setminus S_n| \le n-2 \wedge (S_2 \setminus S_n \neq \emptyset \rightarrow \max(S_2\setminus S_n) < \min(S_n))\ \wedge \\
\forall y, z.\ \suc((S_2 \setminus S_n)  \cup \{\min(S_n)\}, y, z) \rightarrow (\lfloor \varphi_{R,2} \rfloor_{\min(S), \min(S')}[y/\min(S), z/\min(S')]).
\end{array}
$
\medskip
}

\begin{example} \label{exmp-subcase-II-ii} 
	Let $\varphi_{R}(S, S') \equiv S = S' \cup \{\min(S)\} \wedge \max(S) = \max(S') \wedge \min(S') = \min(S)+1 \wedge \min(S) \le \max(S) - 1 \wedge \min(S') \le \max(S')$. This falls into Subcase II(ii). One can   obtain that $\TC[\varphi_R](S,S')=$ 
	
	\smallskip
	\noindent {\small
		$
		\begin{array}{l}
		(S = S') \vee \varphi_R(S,S') \vee \varphi^{(2)}_R(S,S') \ \vee \\ 
		\exists S_1,S_2. \left(
		\begin{array}{l}
		S = S_1 \cup \{\min(S)\} \wedge S_2 = S' \cup \{\min(S_2)\}  \ \wedge \\
		\max(S)  = \max(S_1) \wedge \max(S_2) = \max(S')\ \wedge\\
		S_2 \neq \emptyset \wedge S_1 \setminus S_2 \neq \emptyset \wedge S_2 \subseteq S_1 \wedge \max(S_1 \setminus S_2) < \min(S_2)\ \wedge\\
		\min(S) \le \max(S) - 1 \wedge \min(S_2) \le \max(S_2) - 1  \wedge \min(S_1) \le \max(S_1)\ \wedge\\
		\min(S') \le \max(S') \wedge \min(S) +1 = \min(S_1) \wedge \min(S_2) +1 = \min(S')\ \wedge \\
		\forall y, z.\ \suc((S_1\setminus S_2) \cup \{\min(S_2)\}, y, z) \rightarrow y+ 1 = z
		\end{array}
		\right),
		\end{array}
		$
	}
	\smallskip
	
	\noindent which can be simplified into 
	
	\smallskip
	{\small
		\hspace{5mm}$
		\begin{array}{l}
		(S = S')\vee \varphi_R(S,S') \vee \varphi^{(2)}_R(S,S') \ \vee \\ 
		\exists S_1,S_2. \left(
		\begin{array}{l}
		S = S_1 \cup \{\min(S)\} \wedge S_2 = S' \cup \{\min(S_2)\}\ \wedge\\
		S_2 \neq \emptyset \wedge S_1 \setminus S_2 \neq \emptyset \wedge S_2 \subseteq S_1 \wedge \max(S_1 \setminus S_2) < \min(S_2)\ \wedge\\
		\min(S) +1 = \min(S_1) \wedge \min(S_2) +1 = \min(S')\ \wedge \\
		\forall y, z.\ \suc((S_1\setminus S_2) \cup \{\min(S_2)\}, y, z) \rightarrow y+ 1 = z
		\end{array}
		\right).
		\end{array}
		$
	}
\end{example}

\hide{
As both $S$ and $S'$ are surely nonempty in $\varphi_R$, by Definition~\ref{def-saturate}(4),(5), 
$\varphi_{R,2}$ contains a conjunct $\min(S) \le \min(S') - c$  for some $c \ge 0$, as well as $\max(S') \le \max(S)$ and $\max(S) \le \max(S')$ (i.e., $\max(S')=\max(S)$).
Therefore, we can assume that 
\[
\begin{array}{l c l}
\varphi_{R,2} & = & \max(S') \le \max(S) \wedge \max(S) \le \max(S') \wedge \lfloor \varphi_{R,2} \rfloor_{\min(S),\min(S')}\ \wedge \\
& &  
 \lfloor \varphi_{R,2} \rfloor_{\min(S),\max(S)} \wedge \lfloor \varphi_{R,2} \rfloor_{\min(S'),\max(S')}.
\end{array}
\]
Note that in $\varphi_{R,2} $ above, the redundant subformulae 
$\lfloor \varphi_{R,2} \rfloor_{\min(S),\max(S')}$  and $\lfloor \varphi_{R,2} \rfloor_{\min(S'),\max(S)}$ are omitted.

Similarly to Subcase II(i), $\TC[\varphi_R](S, S')$ is specified by $(S= S') \vee \bigvee \limits_{n \ge 1} \varphi^{(n)}_R$, and we will simplify $\varphi^{(n)}_R$  to construct a finite formula for $\TC[\varphi_R](S, S')$. 
Similarly to Subcase II(i), let us focus on $\varphi^{(n)}_R$ for $n \ge 3$.
It is easy to see that 
\[
\varphi^{(n)}_R = \exists S_1,\cdots, S_{n+1}. \left(
\begin{array}{l}
S_1 = S \wedge S_{n+1} = S'\ \wedge \\
\bigwedge \limits_{i \in [n]} (S_i = S_{i+1} \cup \{\min(S_i)\} \wedge \varphi_{R,2}[S_i/S, S_{i+1}/S'])
\end{array}
\right).
\]

The subformula $\bigwedge \limits_{i \in [n]} (S_i = S_{i+1} \cup \{\min(S_i)\} \wedge \varphi_{R,2}[S_i/S, S_{i+1}/S'])$ can be rewritten as 
\[
\bigwedge \limits_{i \in [n]} 
\left(
\begin{array}{l}
S_i = S_{i+1} \cup \{\min(S_i)\} \wedge \max(S_i) = \max(S_{i+1})\ \wedge \\
(\lfloor \varphi_{R,2} \rfloor_{\min(S),\min(S')}[S_i/S,S_{i+1}/S'])\ \wedge
(\lfloor \varphi_{R,2} \rfloor_{\min(S),\max(S)}[S_i/S]) \wedge \\ (\lfloor \varphi_{R,2} \rfloor_{\min(S'),\max(S')}[S_{i+1}/S'])
\end{array}
\right).
\]
Since $\max(S_1) = \cdots = \max(S_{n+1})$ and $\min(S_1) \le \cdots \le \min(S_{n+1})$, we can simplify the above formula into 
\[
\begin{array}{l}
\lfloor \varphi_{R,2} \rfloor_{\min(S),\max(S)} \wedge (\lfloor \varphi_{R,2} \rfloor_{\min(S),\max(S)}[S_n/S])\ \wedge \\
(\lfloor \varphi_{R,2} \rfloor_{\min(S'),\max(S')}[S_{2}/S']) \wedge \lfloor \varphi_{R,2} \rfloor_{\min(S'),\max(S')}\ \wedge \\
\bigwedge \limits_{i \in [n]} 
\left(
\begin{array}{l}
S_i = S_{i+1} \cup \{\min(S_i)\} \wedge \max(S_i) = \max(S_{i+1})\ \wedge \\
(\lfloor \varphi_{R,2} \rfloor_{\min(S),\min(S')}[S_i/S,S_{i+1}/S'])
\end{array}
\right).
\end{array}
\]
Therefore, $\varphi^{(n)}_R$ can be transformed into
\[
\small
\exists S_2, S_n.
\left(\begin{array}{l}
\lfloor \varphi_{R,2} \rfloor_{\min(S),\max(S)} \wedge (\lfloor \varphi_{R,2} \rfloor_{\min(S),\max(S)}[S_n/S])\ \wedge \\
(\lfloor \varphi_{R,2} \rfloor_{\min(S'),\max(S')}[S_{2}/S']) \wedge \lfloor \varphi_{R,2} \rfloor_{\min(S'),\max(S')} \wedge S = S_2 \cup \{\min(S)\}\  \wedge \\
S_n = S' \cup \{\min(S_n)\} \wedge \max(S) = \max(S_2)\wedge \max(S_n) = \max(S')\ \wedge \\
(\lfloor \varphi_{R,2} \rfloor_{\min(S),\min(S')}[S_2/S']) \wedge (\lfloor \varphi_{R,2} \rfloor_{\min(S),\min(S')}[S_n/S]) \ \wedge \\
\exists S_3,\cdots, S_{n-1}.
\bigwedge \limits_{2 \le i \le n-1} 
\left(
\begin{array}{l}
S_i = S_{i+1} \cup \{\min(S_i)\} \wedge \max(S_i) = \max(S_{i+1})\ \wedge \\
(\lfloor \varphi_{R,2} \rfloor_{\min(S),\min(S')}[S_i/S,S_{i+1}/S'])
\end{array}
\right)
\end{array}
\right).
\]
Since $n\geq 3$, we observe that 
$$\exists S_3,\cdots, S_{n-1}.
\bigwedge \limits_{2 \le i \le n-1} 
\left(
\begin{array}{l}
S_i = S_{i+1} \cup \{\min(S_i)\} \wedge \max(S_i) = \max(S_{i+1})\ \wedge \\
(\lfloor \varphi_{R,2} \rfloor_{\min(S),\min(S')}[S_i/S,S_{i+1}/S'])
\end{array}
\right)
$$		
is equivalent to 
\[
\small
\begin{array}{l}
S_n \subseteq S_2 \wedge |S_2 \setminus S_n| \le n-2 \wedge (S_2 \setminus S_n \neq \emptyset \rightarrow \max(S_2\setminus S_n) < \min(S_n))\ \wedge \\
\forall y, z.\ \suc((S_2 \setminus S_n)  \cup \{\min(S_n)\}, y, z) \rightarrow (\lfloor \varphi_{R,2} \rfloor_{\min(S), \min(S')}[y/\min(S), z/\min(S')]),
\end{array}
\]
where $\suc(S, x, y)$ specifies that $y$ is the successor of $x$ in $S$, that is, 
\[\suc(S, x, y)= x \in S \wedge y \in S \wedge x < y \wedge \forall z  \in S.\ (z \le x \vee y \le z).\]
Moreover, we also observe that when $n=2$,  these two formulae become $\ltrue$, so their equivalence holds trivially. 



Therefore, we deduce that 
$\TC[\varphi_R](S, S')=$ 
\[
\small
\begin{array}{l}
(S = S') \vee  \varphi_R(S,S') \ \vee \\
\exists S_1, S_2.
\left(\begin{array}{l}
S = S_1 \cup \{\min(S)\} \wedge S_2 = S' \cup \{\min(S_2)\}\ \wedge \\
\max(S)  = \max(S_1) \wedge \max(S_2) = \max(S')\ \wedge\\
 S_2 \subseteq S_1 \wedge (S_1 \setminus S_2 \neq \emptyset \rightarrow \max(S_1 \setminus S_2) < \min(S_2))\ \wedge \\
\lfloor \varphi_{R,2} \rfloor_{\min(S),\max(S)} \wedge (\lfloor \varphi_{R,2} \rfloor_{\min(S),\max(S)}[S_2/S])\ \wedge \\
(\lfloor \varphi_{R,2} \rfloor_{\min(S'),\max(S')}[S_1/S']) \wedge \lfloor \varphi_{R,2} \rfloor_{\min(S'),\max(S')} \  \wedge \\
(\lfloor \varphi_{R,2} \rfloor_{\min(S),\min(S')}[S_1/S']) \wedge (\lfloor \varphi_{R,2} \rfloor_{\min(S),\min(S')}[S_2/S]) \ \wedge \\
\forall y, z.
\left(
\begin{array}{l c l}
\suc((S_1 \setminus S_2)  \cup \{\min(S_2)\}, y, z) & \rightarrow & \\
& & \hspace{-3cm} (\lfloor \varphi_{R,2} \rfloor_{\min(S), \min(S')}[y/\min(S), z/\min(S')])
\end{array}
\right)
\end{array}
\right).
\end{array}
\]


\begin{example}\label{exmp-subcase-II-ii} Let $\varphi_{R}(S, S')$ be the formula illustrated in Figure~\ref{fig-varphi-r-2}(b). This falls into Subcase II(ii). One can   obtain that $\TC[\varphi_R](S,S')=$ 
	\[
	\small
	\begin{array}{l}
	(S = S') \vee \varphi_R(S,S') \ \vee \\ 
	\exists S_1,S_2. \left(
	\begin{array}{l}
	S = S_1 \cup \{\min(S)\} \wedge S_2 = S' \cup \{\min(S_2)\}  \ \wedge \\
	\max(S)  = \max(S_1) \wedge \max(S_2) = \max(S')\ \wedge\\
	S_2 \subseteq S_1 \wedge (S_1 \setminus S_2 \neq \emptyset \rightarrow \max(S_1 \setminus S_2) < \min(S_2))\ \wedge\\
	 \min(S) \le \max(S) - 1 \wedge \min(S_2) \le \max(S_2) - 1  \wedge \min(S_1) \le \max(S_1)\ \wedge\\
	 \min(S') \le \max(S') \wedge \min(S) +1 = \min(S_1) \wedge \min(S_2) +1 = \min(S')\ \wedge \\
	\forall y, z.\ \suc((S_1\setminus S_2) \cup \{\min(S_2)\}, y, z) \rightarrow y+ 1 = z
	\end{array}
	\right),
	\end{array}
	\]
	which can be simplified into 
	\[
	\small
	\begin{array}{l}
	(S = S')\vee \varphi_R(S,S') \ \vee \\ 
	\exists S_1,S_2. \left(
	\begin{array}{l}
	S = S_1 \cup \{\min(S)\} \wedge S_2 = S' \cup \{\min(S_2)\}\ \wedge\\
	 S_2 \subseteq S_1 \wedge (S_1 \setminus S_2 \neq \emptyset \rightarrow \max(S_1 \setminus S_2) < \min(S_2))\ \wedge\\
	\min(S) +1 = \min(S_1) \wedge \min(S_2) +1 = \min(S')\ \wedge \\
	\forall y, z.\ \suc((S_1\setminus S_2) \cup \{\min(S_2)\}, y, z) \rightarrow y+ 1 = z
	\end{array}
	\right).
	\end{array}
	\]
	%
\end{example}
}

\medskip

\noindent {\bf \underline{Case III: $\varphi_{R,1} = S = S' \cup \{\max(S)\}$}}.

\medskip

This is similar to (actually symmetrical to) Case II. 

\medskip

\noindent {\bf \underline{Case IV: $\varphi_{R,1} = S = S' \cup \{\min(S), \max(S)\}$}}.

\medskip
 
We still distinguish between whether $S'$ is possibly empty in $\varphi_R$ or not.

\medskip
\noindent{\bf Subcase IV(i): $S'$ is possibly empty in $\varphi_R$.}

\medskip
 
In this case, neither $\min(S')$ nor $\max(S')$ occurs in $\varphi_{R,2}$. Therefore, $\varphi_{R,2}$ is a formula involving  $\min(S)$ or $\max(S)$ only. As before, we analyse the structure of $\varphi^{(n)}_R$ and construct $\TC[\varphi_R](S,S')$. 
$$
\begin{array}{l}
\TC[\varphi_R](S, S')= (S = S') \vee  \\
\exists S_1,S_2, S_3. \left(
\begin{array}{l}
S_2 = S' \cup \{\min(S_2), \max(S_2)\}\ \wedge\\
 S = S_1 \cup S_2 \cup S_3\ \wedge  \varphi_{R,2} \wedge (\varphi_{R,2}[S_2/S])\ \wedge \\
(S_1 \neq \emptyset \rightarrow \max(S_1) < \min(S_2))\ \wedge \\
(S_3 \neq \emptyset \rightarrow \max(S_2) < \min(S_3))
\end{array}
\right).
\end{array}
$$ 

\begin{example}
	Let $\varphi_R(S,S')$ be the normal form of $S = S' \cup \{\min(S), \max(S)\} \wedge \min(S) \le \max(S) - 10 \wedge  \max(S) \le \min(S) + 100$. Then 
	$\TC[\varphi_R](S, S')=$ 
	$$
	\begin{array}{l}
	(S = S')\ \vee  \\
	\exists S_1,S_2, S_3. \left(
	\begin{array}{l}
	S_2 = S' \cup \{\min(S_2), \max(S_2)\}\ \wedge S = S_1 \cup S_2 \cup S_3\  \wedge\\ 
	\min(S) + 10 \le \max(S)  \le \min(S) + 100\ \wedge \\
	\min(S_2) + 10  \le \max(S_2) \le \min(S_2) + 100\ \wedge \\
	(S_1 \neq \emptyset \rightarrow \max(S_1) < \min(S_2))\ \wedge \\
	(S_3 \neq \emptyset \rightarrow \max(S_2) < \min(S_3))
	\end{array}
	\right).
	\end{array}
	$$ 
\end{example}
 
\medskip

\noindent{\bf Subcase IV(ii): $S'$ is surely nonempty in $\varphi_R$.}
Similar to Subcase II(ii), we distinguish between whether $\lfloor \varphi_R \rfloor_{\min(S), \min(S')}$ and $\lfloor \varphi_R \rfloor_{\max(S), \max(S')}$ are strict or not. We exemplify the arguments by considering the situation that both $\lfloor \varphi_R \rfloor_{\min(S), \min(S')}$ and $\lfloor \varphi_R \rfloor_{\max(S), \max(S')}$ are strict. The arguments for the other situations are similar.

Suppose that both $\lfloor \varphi_R \rfloor_{\min(S), \min(S')}$ and $\lfloor \varphi_R \rfloor_{\max(S), \max(S')}$ are strict.
As before, we analyse the structure of $\varphi^{(n)}_R$ and construct $\TC[\varphi_R](S,S')$. 
%
$$
\small
\begin{array}{l}
\TC[\varphi_R](S,S') = (S = S') \vee \varphi_R(S,S') \vee  \varphi^{(2)}_R(S,S') \ \vee \\
\exists S_1, S_2, S_3, S_4. \left(
\begin{array}{l}
S = S_1 \cup \{\min(S), \max(S)\}  \wedge S_1 = S_3 \cup S_2 \cup S_4 \ \wedge \\
S_2 = S' \cup \{\min(S_2), \max(S_2)\} \wedge S_3 \neq \emptyset \wedge  S_4 \neq \emptyset \ \wedge \\
\max(S_3) < \min(S_2) \wedge \max(S_2) < \min(S_4)\ \wedge\\
\lfloor \varphi_{R,2}\rfloor_{\min(S), \max(S)} \wedge  \lfloor \varphi_{R,2}\rfloor_{\min(S), \max(S)}[S_2/S]\ \wedge \\ 
\lfloor \varphi_{R,2}\rfloor_{\min(S'), \max(S')}[S_1/S'] \wedge \lfloor \varphi_{R,2}\rfloor_{\min(S'), \max(S')} \ \wedge \\
\lfloor \varphi_{R,2}\rfloor_{\min(S), \max(S')}[S_1/S'] \wedge \lfloor \varphi_{R,2}\rfloor_{\min(S), \max(S')}[S_2/S]  \ \wedge \\
\lfloor \varphi_{R,2}\rfloor_{\min(S'), \max(S)}[S_1/S']  \wedge \lfloor \varphi_{R,2}\rfloor_{\min(S'), \max(S)}[S_2/S] \ \wedge \\
\lfloor \varphi_{R,2}\rfloor_{\min(S), \min(S')}[S_1/S']  \wedge \lfloor \varphi_{R,2}\rfloor_{\min(S), \min(S')}[S_2/S] \ \wedge \\
\lfloor \varphi_{R,2}\rfloor_{\max(S), \max(S')}[S_1/S']  \wedge \lfloor \varphi_{R,2}\rfloor_{\max(S), \max(S')}[S_2/S] \ \wedge \\
\forall y, z.\ \left(
\begin{array}{l}
\suc(S_3 \cup \{ \min(S_2)\}, y, z) \rightarrow \\
\hspace{1cm} \lfloor \varphi_{R,2}\rfloor_{\min(S), \min(S')} [y/\min(S),z/\min(S')]
\end{array}
\right)
\ \wedge\\
\forall y, z.\ 
\left(
\begin{array}{l}
\suc(S_4 \cup \{\max(S_2)\}, y, z) \rightarrow \\
\hspace{1cm} \lfloor \varphi_{R,2}\rfloor_{\max(S), \max(S')} [y/\max(S'),z/\max(S)]
\end{array}
\right)
\ \wedge\\
\quantel
\left(\exists x.
\left(
\begin{array}{l}
x > 0 \wedge (\lfloor \varphi_{R,2} \rfloor_{\min(S), \min(S')})' \wedge \\
(\lfloor \varphi_{R,2} \rfloor_{\max(S), \max(S')})'
\end{array}
\right)\right)
\end{array}
\right),
\end{array}
$$ 
where $\quantel$ means quantifier elimination, $(\lfloor \varphi_{R,2} \rfloor_{\min(S), \min(S')})'$ is obtained from $\lfloor \varphi_{R,2} \rfloor_{\min(S), \min(S')}$ by replacing any (possible) occurrence of $\min(S) \le \min(S') - c$ with $\min(S) \le \min(S') - cx$, and any (possible) occurrence of $\min(S') \le \min(S) + c'$  with $\min(S') \le \min(S) + c'x$. Similarly, $(\lfloor \varphi_{R,2} \rfloor_{\max(S), \max(S')})'$ is obtained from $\lfloor \varphi_{R,2} \rfloor_{\max(S), \max(S')}$ by replacing any (possible) occurrence of  $\max(S') \le \max(S) - c$  with $\max(S') \le \max(S) - cx$, and any (possible) occurrence of  $\max(S) \le \max(S') + c'$ with $\max(S) \le \max(S') + c'x$. 
Note that the aforementioned formula $\quantel(x >0 \wedge \cdots)$ contains no first-order variables, moreover, the conjunct $S_3 \neq \emptyset$ (resp. $S_4 \neq \emptyset$) is a result of the fact that $\lfloor \varphi_{R,2} \rfloor_{\min(S), \min(S')}$ (resp. $\lfloor \varphi_{R,2} \rfloor_{\max(S), \max(S')}$) is strict.

\begin{example}
	Suppose $\varphi_R(S, S')$ is the normal form of the formula
	\[
	S = S' \cup \{\min(S), \max(S)\} \wedge \min(S') = \min(S)+2 \wedge \max(S) = \max(S')+3.
	\]
	We have 
\[
\begin{array}{l c l}
	\lfloor \varphi_{R,2}\rfloor_{\min(S), \max(S)}  & = & \min(S) \le \max(S) - 5, \\
	\lfloor \varphi_{R,2}\rfloor_{\min(S'), \max(S')} & = &  \min(S') \le \max(S'), \\
	\lfloor \varphi_{R,2}\rfloor_{\min(S), \min(S')} & = &  \min(S') = \min(S)+2,\\  
	\lfloor \varphi_{R,2}\rfloor_{\max(S), \max(S')}  & = &   \max(S) = \max(S')+3,\\  
	\lfloor \varphi_{R,2}\rfloor_{\min(S), \max(S')}  & = &   \min(S) \le \max(S')-2, \\ 
	\lfloor \varphi_{R,2}\rfloor_{\min(S'),\max(S)}  & = &    \min(S') \le \max(S)-3.
\end{array}
\]

Then 
	$$
	\small
	\begin{array}{l}
	\TC[\varphi_R](S,S') = (S = S') \vee \varphi_R(S,S')\ \vee \\
	\exists S_1, S_2, S_3, S_4. \left(
	\begin{array}{l}
	S = S_1 \cup \{\min(S), \max(S)\}  \wedge S_1 = S_3 \cup S_2 \cup S_4 \ \wedge \\
	S_2 = S' \cup \{\min(S_2), \max(S_2)\} \wedge S_3 \neq \emptyset \wedge S_4 \neq \emptyset\ \wedge \\
	\max(S_3) < \min(S_2)) \wedge \max(S_2) < \min(S_4)\ \wedge\\
	\min(S) \le \max(S) - 5 \wedge \min(S_2) \le \max(S_2) - 5\ \wedge \\
 	\min(S_1) \le \max(S_1) \wedge \min(S') \le \max(S')\ \wedge \\
	\min(S_1) = \min(S)+2 \wedge  \min(S') = \min(S_2)+2 \ \wedge \\
	\max(S) = \max(S_1)+3 \wedge \max(S_2) = \max(S')+3 \ \wedge \\
	\min(S) \le \max(S_1)-2 \wedge \min(S_2) \le \max(S')-2  \ \wedge \\
	\min(S_1) \le \max(S)-3  \wedge \min(S') \le \max(S_2)-3 \ \wedge \\
	\forall y, z.\ \left(
	\suc(S_3 \cup \{ \min(S_2)\}, y, z) \rightarrow  z = y+2
	\right)
	\ \wedge\\
	\forall y, z.\ 
	\left(
	\suc(S_4 \cup \{  \max(S_2) \}, y, z) \rightarrow   z = y+3
	\right)
	\ \wedge\\
	3(\min(S') - \min(S)) = 2(\max(S)- \max(S'))
	\end{array}
	\right).
	\end{array}
	$$ 
	%
\end{example}



\subsection{The general situation that there are at least two source and destination parameters} \label{sec-tc-general}

In this section, we consider the general situation $\varphi_R(\vec{S}, \vec{S'})$ that there are at least two source and destination parameters. We write $len(\vec{S})$ for the \emph{length} of $\vec{S}$. 


Recall two conditions {\bf C3} and {\bf C4} in Section~\ref{sec-prelm}. It follows that, for a predicate $P$, the data formula $\varphi_P(\vec{S}; \vec{S'})$ extracted from the inductive rule of $P$ satisfies the \emph{independence} property. Namely, for each atomic formula $\varphi$ in $\varphi_P(\vec{S}, \vec{S'})$, there is some $i \in [len(\vec{S})]$  such that all the variables in $\varphi$ are from $\{S_i, S'_i\}$. This property is crucial to obtain a complete decision procedure.  

Let $len(\vec{S})=k$. Then $\varphi_R(\vec{S}, \vec{S'})$ can be rewritten into $\bigwedge \limits_{i \in [k]} \varphi^{(i)}_{R}(S_i, S'_i)$, where $\varphi^{(i)}_{R}$ is the conjunction of the atomic formulae of $\varphi_R$ that involve only the variables from $\{S_i, S'_i\}$. Moreover, for $i \in [k]$, let $\varphi^{(i)}_{R}(S_i, S'_i) = \varphi^{(i)}_{R,1}(S_i, S'_i) \wedge \varphi^{(i)}_{R,2}(S_i, S'_i)$, where $\varphi^{(i)}_{R,1}(S_i, S'_i)$ and $\varphi^{(i)}_{R,2}(S_i, S'_i)$ are the set and integer subformula of $\varphi^{(i)}_{R}(S_i, S'_i)$ respectively.

To compute the transitive closure for $\varphi_R(\vec{S}, \vec{S'}) = \bigwedge \limits_{i \in [k]} \varphi^{(i)}_{R}(S_i, S'_i)$,  we compute $\TC[\varphi^{(i)}_{R}](S_i, S'_i)$ separately for each $i\in [k]$, but  we also need to take the \emph{synchronisation} of $\varphi^{(i)}_{R}(S_i, S'_i)$ into account. For instance, let $\varphi_R(S_1, S_2, S'_1, S'_2) \equiv \varphi^{(1)}_{R}(S_1, S'_1) \wedge \varphi^{(2)}_{R}(S_2, S'_2)$, where  $\varphi^{(i)}_{R}(S_i, S'_i) $ is obtained from 
$$
\begin{array}{l}
S = S' \cup \{\min(S)\} \wedge \max(S) = \max(S') \wedge \min(S') = \min(S)+1\ \wedge\\
\hspace{4mm} \min(S) \le \max(S) - 1 \wedge \min(S') \le \max(S')
\end{array}
$$ 
by replacing $S,S'$ with $S_i, S'_i$ respectively for $i =1,2$. 
It is not hard to see that $\TC[\varphi_R](S_1, S_2, S'_1, S'_2)$ should also include $\min(S'_1)- \min(S_1) = \min(S'_2)- \min(S_2)$. 

For $i \in [k]$, let $\Phi_i$ denote the formula obtained from $\TC[\varphi^{(i)}_{R}](S_i, S'_i)$ by removing the disjunct $S_i = S'_i$. Then we compute $\TC[\varphi_R](\vec{S}, \vec{S'})$ as,

$
\small
\begin{array}{l}
TC[\varphi_R](\vec{S}, \vec{S'}) = \big (\bigwedge_{i \in [k]} S_i = S'_i \big)\ \vee  \\
\left(
\bigwedge_{i \in [k]} \Phi_i \wedge 
\quantel\left(
\exists x.\ x > 0 \wedge \bigwedge_{i \in [k]} 
\left( 
\begin{array}{l}
(\lfloor \varphi^{(i)}_{R,2} \rfloor_{\min(S_i), \min(S'_i)})'\ \wedge \\
(\lfloor \varphi^{(i)}_{R,2} \rfloor_{\max(S_i), \max(S'_i)})'
\end{array} 
\right)
\right)
\right),
\end{array}
$
\smallskip

\noindent where $(\lfloor \varphi^{(i)}_{R,2} \rfloor_{\min(S_i), \min(S'_i)})'$ is obtained from $\lfloor \varphi^{(i)}_{R} \rfloor_{\min(S_i), \min(S'_i)}$ by replacing $\min(S_i) \le \min(S'_i) - c$ with $\min(S_i) \le \min(S'_i) - cx$, and $\min(S'_i) \le \min(S_i) + c'$ with $\min(S'_i) \le \min(S_i) + c'x$; similarly for $(\lfloor \varphi^{(i)}_{R,2} \rfloor_{\max(S_i), \max(S'_i)})'$. Here, $\quantel$ denotes the quantifier elimination procedure to remove the variable $x$. This is possible as $(\lfloor \varphi^{(i)}_{R,2} \rfloor_{\min(S_i), \min(S'_i)})'$ and $(\lfloor \varphi^{(i)}_{R,2} \rfloor_{\max(S_i), \max(S'_i)})'$ are both Presburger arithmetic formulae.  

\begin{example}\label{exmp-general-case}
	Let $\varphi_R(S_1, S_2, S'_1, S'_2) \equiv \varphi^{(1)}_{R}(S_1, S'_1) \wedge \varphi^{(2)}_{R}(S_2, S'_2)$, where for $i =1,2$, $\varphi^{(i)}_{R}(S_i, S'_i) $ is obtained from the formula 
$$
\begin{array}{l}
S = S' \cup \{\min(S)\} \wedge \max(S) = \max(S') \wedge \min(S') = \min(S)+1\ \wedge\\
\hspace{4mm} \min(S) \le \max(S) - 1 \wedge \min(S') \le \max(S')
\end{array}
$$ 	
by replacing $S,S'$ with $S_i, S'_i$ respectively. Then in $TC[\varphi_R](\vec{S}, \vec{S'})$, 
\[
	\begin{array}{l l}
	& \quantel\left(
	\exists x.\ x > 0 \wedge \bigwedge_{i \in [k]} 
	\left( 
	\begin{array}{l}
	(\lfloor \varphi^{(i)}_{R,2} \rfloor_{\min(S_i), \min(S'_i)})'\ \wedge \\
	(\lfloor \varphi^{(i)}_{R,2} \rfloor_{\max(S_i), \max(S'_i)})'
	\end{array} 
	\right)
	\right) \\
	= & \exists x.\ x > 0 \wedge \bigwedge_{i =1,2} (\min(S'_i) = \min(S_i) + x \wedge \max(S_i) =  \max(S'_i))\\
	= & \min(S_1 ) < \min(S'_1) \wedge \min(S'_1) - \min(S_1) = \min(S'_2) - \min(S_2)\ \wedge\\
	&  \max(S_1) = \max(S'_1) \wedge   \max(S_2) = \max(S'_2).
	\end{array}
\]
\end{example}



\begin{example}
Let $\varphi_R(S_1, S_2, S'_1, S'_2) \equiv \varphi^{[1]}_{R}(S_1, S'_1) \wedge \varphi^{[2]}_{R}(S_2, S'_2)$, where for $i =1,2$, $\varphi^{[i]}_{R}(S_i, S'_i) $ is obtained from 
$$
\begin{array}{l}
S = S' \cup \{\min(S)\} \wedge \max(S) = \max(S') \wedge \min(S') = \min(S)+1\ \wedge\\
\hspace{4mm} \min(S) \le \max(S) - 1 \wedge \min(S') \le \max(S')
\end{array}
$$ 	
by replacing $S,S'$ with $S_i, S'_i$ respectively. 
From Example~\ref{exmp-subcase-II-ii}, we know that for $i = 1, 2$, 
\[
\begin{array}{l}
TC[\varphi^{[i]}_{R}](S_i, S'_i) =  (S_i = S'_i) \vee \varphi^{[i]}_{R}(S_i,S'_i) \vee (\varphi^{[i]}_{R})^{(2)}(S_i,S'_i) \ \vee \\ 
	\exists S_{i, 1},S_{i, 2}. \left(
	\begin{array}{l}
	S_i = S_{i, 1} \cup \{\min(S_i)\} \wedge S_{i, 2} = S'_i \cup \{\min(S_{i, 2})\} \wedge S_{i, 2} \neq \emptyset \ \wedge\\
	 S_{i, 2} \subseteq S_{i, 1} \wedge S_{i, 1} \setminus S_{i, 2} \neq \emptyset \wedge \max(S_{i, 1} \setminus S_{i, 2}) < \min(S_{i, 2})\ \wedge\\
	\min(S_i) +1 = \min(S_{i, 1}) \wedge \min(S_{i, 2}) +1 = \min(S'_i)\ \wedge \\
	\forall y, z.\ \suc((S_{i, 1} \setminus S_{i, 2}) \cup \{\min(S_{i, 2})\}, y, z) \rightarrow y+ 1 = z
	\end{array}
	\right).
\end{array}
\]

Then 
{\small
\[
\begin{array}{l}
TC[\varphi_R](\vec{S}, \vec{S'}) =  (S_1 = S'_1 \wedge S_2 = S'_2 )\ \vee  \\
\left(
 \Phi_1 \wedge \Phi_2 \wedge 
\quantel\left(
\exists x.\ x > 0 \wedge \bigwedge_{i =1,2} 
\left( 
\begin{array}{l}
(\lfloor \varphi^{(i)}_{R,2} \rfloor_{\min(S_i), \min(S'_i)})'\ \wedge \\
(\lfloor \varphi^{(i)}_{R,2} \rfloor_{\max(S_i), \max(S'_i)})'
\end{array} 
\right)
\right)
\right),
\end{array}
\]
}
where $\Phi_1$ and $\Phi_2$ are obtained from $TC[\varphi^{[1]}_R](S_1, S'_1)$  and $TC[\varphi^{[2]}_R](S_2, S'_2)$ by removing the disjunct $S_1 = S'_1$ and $S_2 = S'_2$ respectively. Moreover,
\[
\begin{array}{l l}
& \quantel\left(
\exists x.\ x > 0 \wedge \bigwedge_{i =1,2} 
\left( 
\begin{array}{l}
(\lfloor \varphi^{[i]}_{R,2} \rfloor_{\min(S_i), \min(S'_i)})'\ \wedge \\
(\lfloor \varphi^{[i]}_{R,2} \rfloor_{\max(S_i), \max(S'_i)})'
\end{array} 
\right)
\right) \\
= & \exists x.\ x > 0 \wedge \bigwedge_{i =1,2} (\min(S'_i) = \min(S_i) + x \wedge \max(S_i) =  \max(S'_i))\\
= & \min(S_1 ) < \min(S'_1) \wedge \min(S'_1) - \min(S_1) = \min(S'_2) - \min(S_2)\ \wedge\\
&  \max(S_1) = \max(S'_1) \wedge   \max(S_2) = \max(S'_2).
\end{array}
\]

Therefore, $TC[\varphi_R](\vec{S}, \vec{S'})$ can be simplified into
{\small
\[
\begin{array}{l}
 (S_1 = S'_1 \wedge S_2 = S'_2 ) \vee  \varphi^{[1]}_R(S_1,S'_1) \vee (\varphi^{[1]}_R)^{(2)}(S_1,S'_1)\ \vee  \\
\left(
\begin{array}{l}
\left( \begin{array}{l}
 (\min(S_1)+ 1 = \min(S'_1))\ \vee \\ 
	\exists S_{1, 1},S_{1, 2}. 
	\left(
	\begin{array}{l}
	S_1 = S_{1, 1} \cup \{\min(S_1)\} \wedge S_{1, 2} = S'_1 \cup \{\min(S_{1, 2})\} \wedge S_{1,2} \neq \emptyset \ \wedge\\
	 S_{1, 2} \subseteq S_{1, 1} \wedge S_{1, 1} \setminus S_{1, 2} \neq \emptyset \wedge \max(S_{1, 1} \setminus S_{1, 2}) < \min(S_{1, 2})\ \wedge\\
	\min(S_1) +1 = \min(S_{1, 1}) \wedge \min(S_{1, 2}) +1 = \min(S'_1)\ \wedge \\
	\forall y, z.\ \suc((S_{1, 1} \setminus S_{1, 2}) \cup \{\min(S_{1, 2})\}, y, z) \rightarrow y+ 1 = z
	\end{array}
	\right)
\end{array}
\right)\ \wedge \\
\left(
\begin{array}{l}
(\min(S_2) + 1 = \min(S'_2)) \vee  \varphi^{[2]}_R(S_2,S'_2) \vee (\varphi^{[2]}_R)^{(2)}(S_2,S'_2)\ \vee \\ 
	\exists S_{2, 1},S_{2, 2}. \left(
	\begin{array}{l}
	S_2 = S_{2, 1} \cup \{\min(S_2)\} \wedge S_{2, 2} = S'_2 \cup \{\min(S_{2, 2})\} \wedge S_{2, 2} \neq \emptyset \ \wedge\\
	 S_{2, 2} \subseteq S_{2, 1} \wedge S_{2, 1} \setminus S_{2, 2} \neq \emptyset \wedge \max(S_{2, 1} \setminus S_{2, 2}) < \min(S_{2, 2})\ \wedge\\
	\min(S_2) +1 = \min(S_{2, 1}) \wedge \min(S_{2, 2}) +1 = \min(S'_2)\ \wedge \\
	\forall y, z.\ \suc((S_{2, 1} \setminus S_{2, 2}) \cup \{\min(S_{2, 2})\}, y, z) \rightarrow y+ 1 = z
	\end{array}
	\right)
\end{array}
\right)\  \wedge\\ 
 \min(S_1 ) < \min(S'_1) \wedge \min(S'_1) - \min(S_1) = \min(S'_2) - \min(S_2)\ \wedge \\
   \max(S_1) = \max(S'_1) \wedge   \max(S_2) = \max(S'_2).
\end{array}
\right).
\end{array}
\]
}

\end{example}

\section{Details of Section~\ref{sec:sat-qgdbs}}\label{app-sat-qgdbs}

\subsection{Details of Step I}

To transform the formulae from $\qgdbs_\intnum$ (resp. $\qgdbs^-_\intnum$) into $\qgdbs_\natnum$ (resp.\ $\qgdbs^-_\natnum$), while preserving the set of models, we use an encoding $\mapzn: \intnum \rightarrow \natnum \times \natnum \uplus \intset \rightarrow \natset\times \natset$, viz., $\mapzn(n) = (n^+, n^-)$, where $(n^+, n^-)=(n, 0)$ if $n \ge 0$ and $(n^+, n^-)=(0, -n)$ if $n < 0$, $\mapzn(A) = (A^+,A^-)$ 
where $A^+ = A \cap \natnum$ and $A^- = \{-n \mid n \in A \setminus A^+\}$. 
%
Naturally, we extend $\mapzn$ to tuples of elements from $\intnum \uplus \intset$ as follows: For $(n_1, \cdots, n_k, A_1, \cdots, A_l) \in \intnum^k \times \intset^l$, $\mapzn(n_1, \cdots, n_k, A_1, \cdots, A_l) = (n^+_1, n^-_1, \cdots, n^+_k, n^-_k, A^+_1,A^-_1, \cdots, A^+_l, A^-_l)$. 
 
\begin{lemma}\label{lem-edbs-z2n}
Let $\Phi(\vec{x}, \vec{S})$ be an $\qgdbs_\intnum$ formula, where $\vec{x} = (x_1, \cdots, x_k)$ and $\vec{S} = (S_1,\cdots, S_l)$. Then an $\qgdbs_\natnum$ formula $\Phi'(\vec{x^\pm}, \vec{S^\pm})$ can be constructed effectively such that $\mapzn(\Ll(\Phi)) = \Ll(\Phi')$, where $\vec{x^\pm} = (x^+_1, x^-_1, \cdots, x^+_k, x^-_k)$ and $\vec{S^\pm}  = (S^+_1, S^-_1,\cdots, S^+_l, S^-_l)$. Moreover, if $\Phi(\vec{x}, \vec{S})$ is an $\qgdbs^-_\intnum$ formula, then $\Phi'(\vec{x^\pm}, \vec{S^\pm})$ is an $\qgdbs^-_\natnum$ formula.
\end{lemma}

\begin{proof}

Let $\Phi(x_1, \cdots, x_k, S_1,\cdots, S_l)$ be an $\qgdbs_\intnum$ formula. W.l.o.g. we assume that for each variable from $ \{x_1,\cdots, x_k, S_1,\cdots, S_l\}$, there are no quantified occurrences of the variable in $\Phi$; moreover, each variable is quantified at most once. By adding fresh free set variables for terms, we can transform $\Phi$ into a formula satisfying that for each integer term of the form $\max(T_s)$ or $\min(T_s)$ in $\Phi$, $T_s$ is a set variable. For instance, if $\Phi \equiv \min(S_1 \setminus S_2) \le \max(S_3 \cup S_4) + 1$, then we can introduce fresh free set variables $S'_1, S'_2$ and turn $\Phi'$ into the formula $S'_1 = S_1 \setminus S_2 \wedge S'_2 = S_3 \setminus S_4 \wedge \min(S'_1) \le \max(S'_2)+1$. Therefore, from now on, we assume that $\Phi$ satisfies that whenever $\max(T_s)$ or $\min(T_s)$ occurs, $T_s$ is a set variable.

We use $\vars_{fo}(\Phi)$ (resp. $\vars_{so}(\Phi)$) to denote the set of (not necessarily free) first-order (resp. second-order) variables occurring in $\Phi$.	
A \emph{$\Phi$-context} $\ectx$ is a function from $\vars(\Phi)$ to $\{+,-, \pm, \bot\}$ such that for each $x \in \vars_{fo}(\Phi)$, $\ectx(x) \in \{+, -\}$. Intuitively, $\ectx(x)=+$ (resp. $\ectx(x)=-$) denotes that $x$ is a non-negative number (resp. $x$ is a negative number), and $\ectx(S)=+$ ($\ectx(S)=-, \pm, \bot$) denotes that $S$ contains only non-negative numbers (resp. $S$ contains only negative numbers, $S$ contains both non-negative and negative numbers, $S$ is an empty set).
We will first show how to transform $\Phi$ into a  $\qgdbs_\natnum$ formula $tr_{\ectx}(\Phi)$, for a given $\Phi$-context $\ectx$. Then we define the desired $\qgdbs_\natnum$ formula $\Phi'$ as 
\[
\bigvee \limits_{\ectx} 
\left(
\begin{array}{l}
\bigwedge \limits_{j \in [l], \ectx(S_j)=+}(S^+_j \neq \emptyset \wedge S^-_j =\emptyset) \wedge \bigwedge \limits_{j \in [l], \ectx(S_j)=-} (S^+_j =\emptyset \wedge S^-_j \neq \emptyset) \ \wedge \\
\bigwedge \limits_{j \in [l], \ectx(S_j)=\pm} (S^+_j \neq \emptyset \wedge S^-_j \neq \emptyset) \wedge \bigwedge \limits_{j \in [l], \ectx(S_j)=\bot} (S^+_j =\emptyset \wedge S^-_j = \emptyset)\ \wedge  \\
 \bigwedge \limits_{i \in [k], \ectx(x_i)=+} (x^-_i = 0) \wedge  \bigwedge \limits_{i \in [k], \ectx(x_i)=-} (x^+_i =0 \wedge x^-_i > 0) \wedge tr_{\ectx}(\Phi)
\end{array}
\right).
\] 
Moreover, we can construct $tr_{\ectx}(\Phi)$ in a way that if $\Phi$ is a $\qgdbs^-_\intnum$ formula, then $tr_{\ectx}(\Phi)$ is a $\qgdbs^-_\natnum$ formula, thus $\Phi'$ is a $\qgdbs^-_\natnum$ formula formula as well.

Suppose that $\Phi$ is a $\qgdbs_\intnum$ formula, $\ectx$ is a $\Phi$-context, and $\Psi$ is a subformula of $\Phi$. We construct $tr_{\ectx}(\Psi)$ inductively as follows.

We start with the atomic formulae of the form $T_m\ \op\ 0$ in $\Phi$. From the definition of $\qgdbs_\intnum$, all the variables occurring in $T_m$ are free variables in $\Phi$. 
%
We construct $tr_\ectx(T_m\ \op\ 0)$ by the following three-step procedure.
\begin{description}
\item[Step 1.] For each first-order variable $x$ occurring in $T_m\ \op\ 0$, replace $x$ with $x^+$ if $\ectx(x) = +$, and replace $x$ with $- x^-$ otherwise.
\item[Step 2.] For each set variable  $S$ occurring in $T_m\ \op\ 0$, 
\begin{itemize}
\item if $\ectx(S) = \bot$, then replace each occurrence of $\max(S)$ or $\min(S)$ in $T_m\ \op\ 0$ with $\bot$,
\item if $\ectx(S) = +$, then replace each occurrence of $\max(S)$ (resp. $\min(S)$) in  $T_m\ \op\ 0$ with $\max(S^+)$ (resp. $\min(S^+)$),
\item if $\ectx(S) = -$, then replace each occurrence of $\max(S)$ (resp. $\min(S)$) in $T_m\ \op\ 0$ with $-\min(S^-)$ (resp. $-\max(S^-)$),
\item if $\ectx(S) = \pm$, then replace each occurrence of $\max(S)$ (resp. $\min(S)$) in $T_m\ \op\ 0$ with $\max(S^+)$ (resp. $-\max(S^-)$).
\end{itemize}  
\item[Step 3.] Let $T'_m\ \op\ 0$ denote the formula obtained after the two steps above. If $T'_m\ \op\ 0$ contains at least one occurrence of $\bot$, then $tr_\ectx(T_m\ \op\ 0) = \lfalse$, otherwise, $tr_\ectx(T_m\ \op\ 0) = T'_m\ \op\ 0$.
\end{description}

 \begin{table}[htbp]
 \vspace{-5mm}
\begin{center}
	\begin{tabular}{|c|c|}
		\hline
		$tr^+_\ectx(\emptyset) = \emptyset$ & $tr^-_\ectx(\emptyset) = \emptyset$ \\
		\hline
		$tr^+_\ectx(S) = \left\{ 
		\begin{array}{l l} 
		S^+, & \mbox{ if } \ectx(S) = \pm \mbox{ or } + \\
		\emptyset, & \mbox{ otherwise }
		\end{array}
		\right.$
		&
		$tr^-_\ectx(S) = \left\{ 
		\begin{array}{l l} 
		S^-, & \mbox{ if } \ectx(S) = \pm \mbox{ or } - \\
		\emptyset, & \mbox{ otherwise }
		\end{array}
		\right.$ 
		 \\
		\hline
		$tr^+_\ectx(\{x\}) =\left\{
		\begin{array}{l l} 
		\{x^+\}, & \mbox{ if } \ectx(x) = + \\
		\emptyset, & \mbox{ otherwise}
		\end{array}
		\right.
		$
		&
		$tr^-_\ectx(\{x\}) =\left\{
		\begin{array}{l c} 
		\{x^-\}, & \mbox{ if } \ectx(x) = - \\
		\emptyset, & \mbox{ otherwise}
		\end{array}
		\right.
		$
		\\
		\hline
		$
		\begin{array}{l l}
		tr^+_\ectx(\{\min(S)\}) =\\
		\hspace{5mm}
		\left\{ 
		\begin{array}{l l}
		\emptyset, & \mbox{ if } \ectx(S) = \pm \mbox{ or } -  \\
		\{\min(S^+)\}, & \mbox{ if } \ectx(S) = +\\
		\bot, & \mbox{ otherwise}
		\end{array}
		\right.
		\end{array}
		$
		&
		$
		\begin{array}{l}
		tr^-_\ectx(\{\min(S)\}) =\\
		\hspace{5mm}\left\{ 
		\begin{array}{l l}
		\{\max(S^-)\}, & \mbox{ if } \ectx(S) = \pm \mbox{ or } - \\
		\emptyset, & \mbox{ if } \ectx(S) = + \\
		\bot, & \mbox{ otherwise}
		\end{array}
		\right.
		\end{array}$ 		
		\\
		\hline
		$
		\begin{array}{l}
		tr^+_\ectx(\{\max(S)\}) =
		\\
		\hspace{5mm}\left\{ 
		\begin{array}{l l}
		\{\max(S^+)\}, & \mbox{ if } \ectx(S) = \pm \mbox{ or } +  \\
		\emptyset, & \mbox{ if } \ectx(S) = -\\
		\bot, & \mbox{ otherwise}
		\end{array}
		\right.
		\end{array}
		$
		&
		$
		\begin{array}{l}
		tr^-_\ectx(\{\max(S)\}) =\\
		\hspace{5mm}\left\{ 
		\begin{array}{l l}
		\emptyset, & \mbox{ if } \ectx(S) = \pm \mbox{ or } + \\
		\{\min(S^-)\}, & \mbox{ if } \ectx(S) = - \\
		\bot, & \mbox{ otherwise}
		\end{array}
		\right.
		\end{array}
		$
		\\
		\hline
		$tr^+_\ectx(T_s \cup T'_s) = tr^+_\ectx(T_s) \cup tr^+_\ectx(T'_s)$  & $tr^-_\ectx(T_s \cup T'_s) = tr^-_\ectx(T_s) \cup tr^-_\ectx(T'_s)$ \\
		\hline
		$tr^+_\ectx(T_s \cap T'_s) = tr^+_\ectx(T_s) \cap tr^+_\ectx(T'_s)$ & $tr^-_\ectx(T_s \cap T'_s) = tr^-_\ectx(T_s) \cap tr^-_\ectx(T'_s)$\\
		\hline
		$tr^+_\ectx(T_s \setminus T'_s) = tr^+_\ectx(T_s) \setminus tr^+_\ectx(T'_s)$ & $tr^-_\ectx(T_s \setminus T'_s) = tr^-_\ectx(T_s) \setminus tr^-_\ectx(T'_s)$\\
		\hline
	\end{tabular}
\end{center}
\caption{Definition of $tr^+_\ectx(T_s)$ and $tr^-_\ectx(T_s)$}
\label{tab-tr-T-s}
\vspace{-5mm}
\end{table}

We then consider the atomic formulae of the form $T_{s,1}\ \opset\ T_{s,2}$. For this purpose, we define two functions $tr^+_\ectx(T_s)$ and $tr^-_\ectx(T_s)$ as shown in Table~\ref{tab-tr-T-s}. Intuitively, $tr^+_\ectx(T_s)$ represents the set of non-negative numbers in $T_s$ under the context $\ectx$, and $tr^-_\ectx(T_s)$ represents the set of $-n$ such that $n$ is a negative number in $T_s$ under the context $\ectx$. Then $tr_\ectx(T_{s,1}\ \opset\ T_{s,2})$ is defined as follows, 
\begin{itemize}
\item if any of $tr^+_\ectx(T_{s,1})$, $tr^-_\ectx(T_{s,1})$, $tr^+_\ectx(T_{s,2})$, or $tr^-_\ectx(T_{s,2})$ contains an occurrence of $\bot$, then $tr_\ectx(T_{s,1}\ \opset\ T_{s,2}) = \lfalse$, 
\item otherwise, $tr_\ectx(T_{s,1}\ \opset\ T_{s,2})  = tr^+_\ectx(T_{s,1}) \opset tr^+_\ectx(T_{s,2}) \wedge tr^-_\ectx(T_{s,1}) \opset tr^-_\ectx(T_{s,2})$.
\end{itemize}

The transformation of the atomic formulae of the form $T_{i,1}\ \op\ T_{i,2}+c$ is more involved, since we want to construct a $\qgdbs^-_\natnum$ formula to encode $T_{i,1}\ \op\ T_{i,2}+c$, so that if the original formula $\Phi$ is a $\qgdbs^-_\intnum$ formula, then $tr_\ectx(\Phi)$ is a $\qgdbs^-_\natnum$ formula. We distinguish between whether $T_{i,2}$ is a constant or not.
We construct  $tr_\ectx(T_{i,1}\ \op\ T_{i,2}+c)$ by the following two-step procedure.
\begin{description}
\item[Step 1.] At first, apply the following replacements: 
\begin{itemize}
\item For each integer variable $x$ occurring in $T_{i,1}\ \op\ T_{i,2}+c$, if $\ectx(x) = +$, then replace each occurrence of $x$ with $x^+$, otherwise, replace each occurrence of $x$ with $-x^-$.
\item For each set variable $S$ occurring in $T_{i,1}\ \op\ T_{i,2}+c$, 
\begin{itemize}
\item if $\ectx(S) = +$, then replace each occurrence of $\min(S)$ (resp. $\max(S)$) with $\min(S^+)$ (resp. $\max(S^+)$),
\item if $\ectx(S) = -$, then replace each occurrence of $\min(S)$ (resp. $\max(S)$) with $- \max(S^-)$ (resp. $- \min(S^-)$),
\item if $\ectx(S) = \pm$, then replace each occurrence of $\min(S)$ (resp. $\max(S)$) with $- \max(S^-)$ (resp. $\max(S^+)$),
\item if $\ectx(S) = \bot$, then replace each occurrence of $\min(S)$ (resp. $\max(S)$) with $\bot$.
\end{itemize}
\end{itemize}
\item[Step 2.] Let $T'_{i,1}\ \op\ T'_{i,2} + c$ be the resulting formula after the replacements above.  
\begin{itemize}
\item If $T'_{i,1}\ \op\ T'_{i,2} + c$ contains an occurrence of $\bot$, the $tr_\ectx(T'_{i,1}\ \op\ T'_{i,2} + c) = \lfalse$. 
%
%
\item Otherwise, we construct $tr_\ectx(T_{i,1}\ \op\ T_{i,2} + c)$ by rewriting $T'_{i,1}\ \op\ T'_{i,2} + c$ into a $\qgdbs^-_\natnum$ formula as follows: Note that $T'_{i,1}\ \op\ T'_{i,2} + c$ is of the form $\alpha\ \op\ \beta + c$, $\alpha\ \op\ -\beta + c$, $-\alpha\ \op\ \beta + c$, or $-\alpha\ \op\ -\beta + c$, where $\alpha, \beta$ are of the form $x^+$, $x^-$, $\max(S^+)$, $\min(S^+)$, $\max(S^-)$, $\min(S^-)$.
\begin{itemize}
\item $T'_{i,1}\ \op\ T'_{i,2} + c$ is of the form $\alpha\ \op\ \beta + c$: then $tr_\ectx(T_{i,1}\ \op\ T_{i,2} + c)   = \alpha\ \op\ \beta + c$.
\item  $T'_{i,1}\ \op\ T'_{i,2} + c$ is of the form $\alpha\ \op\ -\beta + c$: If $ c \ge 0$, then $tr_\ectx(T_{i,1}\ \op\ T_{i,2} + c) = \bigvee \limits_{c_1 + c_2 = c, c_1,c_2 \ge 0} (\alpha\ \op\ c_1 \wedge \beta\ \op\ c_2)$. Otherwise, if $\op \in \{=, \le\}$, then $tr_\ectx(T_{i,1}\ \op\ T_{i,2} + c)  = \lfalse$, otherwise, $tr_\ectx(T_{i,1}\ \op\ T_{i,2} + c)  = \ltrue$.
\item $T'_{i,1}\ \op\ T'_{i,2} + c$ is of the form $-\alpha\ \op\ \beta + c$: If $c \le 0$, then  $tr_\ectx(T_{i,1}\ \op\ T_{i,2} + c) = \bigvee \limits_{c_1 + c_2 = -c, c_1,c_2 \ge 0} (c_1\ \op\ \alpha \wedge c_2\ \op\ \beta)$. Otherwise, if $\op \in \{=, \ge\}$, then $tr_\ectx(T_{i,1}\ \op\ T_{i,2} + c) = \lfalse$, otherwise, $tr_\ectx(T_{i,1}\ \op\ T_{i,2} + c) = \ltrue$.
\item $T'_{i,1}\ \op\ T'_{i,2} + c$ is of the form $-\alpha\ \op\ -\beta + c$: Then $tr_\ectx(T_{i,1}\ \op\ T_{i,2} + c) = \beta\ \op\ \alpha + c$.
\end{itemize}
\end{itemize}
\end{description}

We then consider non-atomic subformulae of $\Phi$.
\begin{itemize}
		\item $tr_\ectx(\Psi_1 \wedge \Psi_2) = tr_\ectx(\Psi_1) \wedge tr_\ectx(\Psi_2)$,
		\item $tr_\ectx(\neg \Psi_1) = \neg tr_\ectx(\Psi_1)$,
		%
		\item $tr_\ectx(\forall x.\ \Psi_1) = \forall x^+.\forall x^-.\ (x^- = 0 \rightarrow tr_{\ectx[x \rightarrow +]}(\Psi_1)) \wedge ((x^+=0 \wedge x^->0) \rightarrow tr_{\ectx[x \rightarrow -]}(\Psi_1))$.
		\item	$tr_\ectx(\forall S.\ \Psi_1) = \forall S^+. \forall S^-.\ ((S^+ \neq \emptyset \wedge S^- \neq \emptyset) \rightarrow tr_{\ectx[S \rightarrow \pm]}(\Psi_1)) \wedge ((S^+ \neq \emptyset \wedge S^- = \emptyset) \rightarrow  tr_{\ectx[S \rightarrow +]}(\Psi_1)) \wedge ((S^+ = \emptyset \wedge S^- \neq \emptyset) \rightarrow tr_{\ectx[S \rightarrow -]}(\Psi_1)) \wedge ((S^+ = \emptyset \wedge S^- = \emptyset) \rightarrow tr_{\ectx[S \rightarrow \bot]}(\Psi_1))$.
\end{itemize}
\qed
\end{proof}

\subsection{Details of Step II} 
Let $\Delta(\vec{x}, \vec{S})$ be a formula in $\qgdbs_\natnum$, where $\vec{x} = (x_1,\cdots, x_k)$ and $\vec{S} = (S_1,\cdots,S_l)$. W.l.o.g., we assume that for each variable from $\vec{x} \cup \vec{S}$, there are no quantified occurrences of the variable in $\Phi$.
Intuitively, as none of the variables from $\vec{x} \cup \vec{S}$ are quantified,  we can separate out all atomic formulae of the form $T_m\ \op\ 0$ which contain \emph{only} variables from $\vec{x} \cup \vec{S}$. 
In detail, let $\mathfrak{F}_{\sf free}(\Phi)$ denote the set of all atomic formulae $T_m \ \op\ 0$ occurring in $\Phi$ such that it contain only variables from $\vec{x} \cup \vec{S}$. Then it is not difficult to see that $\Phi$ can be rewritten into
$$\bigvee \limits_{\mathfrak{F}' \subseteq \mathfrak{F}_{\sf free}(\Phi)} \left(\Red_{\mathfrak{F}'}(\Phi) \wedge \bigwedge \limits_{\Phi' \in \mathfrak{F}'} \Phi'  \wedge \bigwedge \limits_{\Phi' \in \mathfrak{F}_{\sf free}(\Phi) \setminus \mathfrak{F}'} \neg \Phi'  \right),$$ 
where $\Red_{\mathfrak{F}'}(\Phi)$ is obtained from $\Phi$ by replacing each atomic formula in $\mathfrak{F}'$ (resp.\ $\mathfrak{F}_{\sf free}(\Phi) \setminus \mathfrak{F}'$) with $\ltrue$ (resp.\ $\lfalse$). Evidently, $\Red_{\mathfrak{F}'}(\Phi)$ is a formula in $\qgdbs^-_\natnum$. Moreover, $ \neg \Phi' $ can be easily rewritten into a formula of the form $T_m\ \op\ 0$. For instance, $\neg (T_m \ge 0) \equiv T_m < 0$.
  

\hide{
\begin{itemize}
\item $T_{i,2}$ is a constant $c'$: 
\begin{itemize}
\item If $c'+c \ge 0$, then $tr_\ectx(T_{i,1}\ \op\ T_{i,2}+c)$ is presented in Table~\ref{tab-unary-ge-zero}.
\item If $c'+c < 0$, then $tr_\ectx(T_{i,1}\ \op\ T_{i,2}+c)$ is presented in Table~\ref{tab-unary-less-zero}.
\end{itemize}
\item $T_{i,2}$ is not a constant: 
\begin{itemize}
\item If $c \ge 0$, then $tr_\ectx(T_{i,1}\ \op\ T_{i,2}+c)$ is presented in Table~\ref{tab-binary-ge-zero}. Note that in Table~\ref{tab-binary-ge-zero}, there are some formulae which are syntactically not in $(\qgdbs)^-_\natnum$, but can be easily transformed into a formula in $(\qgdbs)^-_\natnum$. For instance, $x^+_1 + x^-_2 \ge c$ can be rewritten into $\bigvee \limits_{c_1 + c_2 = c, c_1 \ge 0, c_2 \ge 0} x^+_1 \ge c_1 \wedge x^-_2 \ge c_2$.
\item If $c<0$, then $T_{i,1}\ \op\ T_{i,2}+c$ can be rewritten into $T_{i,2}\ \op'\ T_{i,1}-c$, where $\op'$ is $=, \ge, \le$ respectively, if $\op$ is $=, \le, \ge$, we define $tr_\ectx(T_{i,1}\ \op\ T_{i,2}+c)$ as $tr_\ectx(T_{i,2}\ \op'\ T_{i,1}-c)$.
\end{itemize}
\end{itemize}

\begin{table}[htbp]
\vspace{-8mm}
	\begin{center}
		\begin{tabular}{|c|c|c|}
			\hline
			$T_{i,1}$ & $\op$ & $tr_\ectx(T_{i,1}\ \op\ T_{i,2}+c)$\\
			\hline
			$x$ & = or $\ge$ & $
			\left\{\begin{array}{l l}
			x^+ \op c'+c, & \mbox{ if } \ectx(x) = + \\
			\lfalse, & \mbox{ otherwise }
			\end{array}
			\right.
			$\\
			\hline
			$x$ & $\le $ & 
			$
			\left\{\begin{array}{l r}
			x^+ \op c'+c, & \mbox{ if } \ectx(x) = + \\
			\ltrue, & \mbox{ otherwise }
			\end{array}
			\right.$\\
			\hline
			$\min(S)$ & $=$ or $\ge$ & 
			$
			\left\{
			\begin{array}{l l}
			\min(tr^+_{\ectx}(S)) \op c'+c, & \mbox{ if } \ectx(S) = + \\
			\lfalse, & \mbox{ otherwise}
			\end{array}
			\right.
			$
			\\
			\hline
			$\min(S)$ & $\le$ & 
			$
			\left\{
			\begin{array}{l l}
			\min(tr^+_{\ectx}(S)) \le c'+c, & \mbox{ if } \ectx(S) = + \\
			\ltrue, & \mbox{ if } \ectx(S) = - \mbox{ or } \pm \\
			\lfalse, & \mbox{ otherwise}
			\end{array}
			\right.
			$\\
			\hline
			$\max(S)$ & $=$ or $\ge$ & 
			$
			\left\{
			\begin{array}{l l}
			\max(tr^+_{\ectx}(S)) \op c'+c, & \mbox{ if } \ectx(S) = + \mbox{ or } \pm \\
			\lfalse, & \mbox{ otherwise}
			\end{array}
			\right.
			$\\
			\hline
			$\max(S)$ & $\le$ & 
			$
			\left\{
			\begin{array}{l l}
			\max(tr^+_{\ectx}(S)) \le c'+c, & \mbox{ if } \ectx(S) = + \mbox{ or } \pm \\
			\ltrue, & \mbox{ if } \ectx(S) = - \\
			\lfalse, & \mbox{ otherwise}
			\end{array}
			\right.
			$\\
			\hline
		\end{tabular}
	\end{center}
	\caption{$tr_\ectx(T_{i,1}\ \op\ T_{i,2}+c)$ for the situation $T_{i,2}=c'$ and $c'+c \ge 0$}
	\label{tab-unary-ge-zero}
	\vspace{-8mm}
\end{table}%

\begin{table}[htbp]
\vspace{-5mm}
	\begin{center}
		\begin{tabular}{|c|c|c|}
			\hline
			$T_{i,1}$ & $\op$ & $tr_\ectx(T_{i,1}\ \op\ T_{i,2}+c)$\\
			\hline
			$x$ & $\ge$ & 
			$
			\left\{
			\begin{array}{l l}
			x^- \le -c'-c, & \mbox{ if } \ectx(x)=- \\
			\ltrue, & \mbox{ otherwise}
			\end{array}
			\right.
			$\\
			\hline
			$x$ & $=$ & 
			$
			\left\{
			\begin{array}{l r}
			x^-\ =\  -c'-c, & \mbox{ if } \ectx(x) = -\\
			\lfalse, & \mbox{ otherwise}
			\end{array}
			\right.$\\
			\hline
			$x$ & $\le$ & 
			$
			\left\{
			\begin{array}{l r}
			x^-\ \ge\  -c'-c, & \mbox{ if } \ectx(x)=- \\
			\lfalse, & \mbox{ otherwise}
			\end{array}
			\right.$\\
			\hline
			$\min(S)$ & $\ge$ & 
			$
			\left\{
			\begin{array}{l l}
			\max(S^-) \le - c' -c, & \mbox{ if } \ectx(S) = - \mbox{ or } \pm \\
			\ltrue, & \mbox{ if } \ectx(S) = +\\
			\lfalse, & \mbox{ otherwise}
			\end{array}
			\right.$\\
			\hline
			$\min(S)$ & $=$ & 
			$
			\left\{
			\begin{array}{l l}
			\max(S^-) = - c' -c, & \mbox{ if } \ectx(S) = - \mbox{ or } \pm \\
			\lfalse, & \mbox{ otherwise }
			\end{array}
			\right.$\\
			\hline			
			$\min(S)$ & $\le$ & 
			$
			\left\{
			\begin{array}{l l}
			\max(S^-) \ge - c' -c, & \mbox{ if } \ectx(S) = - \mbox{ or } \pm \\
			\lfalse, & \mbox{ otherwise }
			\end{array}
			\right.$\\
			\hline
			$\max(S)$ & $\ge$ & 
			$
			\left\{
			\begin{array}{l l}
			\min(S^-) \le - c' -c, & \mbox{ if } \ectx(S) = - \\
			\ltrue, & \mbox{ if } \ectx(S) = + \mbox{ or } \pm\\
			\lfalse, & \mbox{ otherwise}
			\end{array}
			\right.$\\
			\hline
			$\max(S)$ & $=$ & 
			$
			\left\{
			\begin{array}{l l}
			\min(S^-) = - c' -c, & \mbox{ if } \ectx(S) = - \\
			\lfalse, & \mbox{ otherwise}
			\end{array}
			\right.$\\
			\hline
			$\max(S)$ & $\le$ & 
			$
			\left\{
			\begin{array}{l l}
			\min(S^-) \ge - c' -c, & \mbox{ if } \ectx(S) = - \\
			\lfalse, & \mbox{ otherwise}
			\end{array}
			\right.$\\
			\hline
		\end{tabular}
	\end{center}
	\caption{$tr_\ectx(T_{i,1}\ \op\ T_{i,2}+c)$ for the situation $T_{i,2}=c'$ and $c'+c < 0$}
	\label{tab-unary-less-zero}
	\vspace{-5mm}
\end{table}%

\begin{table}[htbp]
\vspace{-5mm}
	\begin{center}
		{\small
			\begin{tabular}{|c|c|c|c|}
				\hline
				$T_{i,1}$ & $T_{i,2}$ & $\op$ & $tr_\ectx(T_{i,1}\ \op\ T_{i,2} + c)$\\
				\hline
				$x_1$ & $x_2$ & &
				$
				\left\{\begin{array}{l r}
				x^+_1 \ \op\  x^+_2 + c,  & \mbox{ if } \ectx(x_1) = +, \ectx(x_2) = +\\
				x^+_1 + x^-_2 \ \op\    c,  & \mbox{ if } \ectx(x_1) = +, \ectx(x_2) = -\\
				-c\ \op\   x^-_1 + x^+_2,  & \mbox{ if } \ectx(x_1) = -, \ectx(x_2) = +\\
				x^-_2 \ \op\   x^-_1 + c,  & \mbox{ if } \ectx(x_1) = -, \ectx(x_2) = -
				\end{array}
				\right.$\\
				\hline
				$x$ & $\min(S)$ & $=$ or $\ge$ & 
				$\begin{array}{l r}
				\begin{array}{l}
				(tr^-_\ectx(T_s) = \emptyset \wedge x^+ \op \min(tr^+_\ectx(T_s)) + c)\ \vee\\
				(x^+ + \max(tr^-_\ectx(T_s)) \op c ), 
				\end{array} 
				& \mbox{ if } \ectx(x) = +\\
				\max(tr^-_\ectx(T_s)) \op x^- + c, & \mbox{ if } \ectx(x) = -
				\end{array}
				$
				\\ 
				\hline
				$x$ & $\min(T_s)$ & $\le$ & 
				$\begin{array}{l r}
				\begin{array}{l}
				(tr^-_\ectx(T_s) = \emptyset \wedge x^+ \le \min(tr^+_\ectx(T_s)) + c)\ \vee\\
				(x^+ + \max(tr^-_\ectx(T_s)) \le c ), 
				\end{array} 
				& \mbox{ if } \ectx(x) = +\\
				(tr^-_\ectx(T_s) = \emptyset) \vee (\max(tr^-_\ectx(T_s)) \le x^- + c), & \mbox{ if } \ectx(x) = -
				\end{array}
				$
				\\
				\hline
				$x$ & $\max(T_s)$ & $=$ or $\le$ & 
				$\begin{array}{l r}
				\begin{array}{l}
				(x^+ \op \max(tr^+_\ectx(T_s)) + c)\ \vee \\
				(tr^+_\ectx(T_s)= \emptyset \wedge x^+ + \min(tr^-_\ectx(T_s)) \op c ), 
				\end{array} 
				& \mbox{ if } \ectx(x) = +\\
				(tr^-_\ectx(T_s) = \emptyset) \vee (\min(tr^-_\ectx(T_s)) \op x^- + c), & \mbox{ if } \ectx(x) = -
				\end{array}
				$
				\\
				\hline
				$x$ & $\max(T_s)$ & $\ge$ & 
				$
				\begin{array}{l r}
				\begin{array}{l}
				(x^+ \ge \max(tr^+_\ectx(T_s)) + c)\ \vee \\
				(tr^+_\ectx(T_s)= \emptyset \wedge x^+ + \min(tr^-_\ectx(T_s)) \ge c ), 
				\end{array} 
				& \mbox{ if } \ectx(x) = +\\
				\min(tr^-_\ectx(T_s)) \ge x^- + c, & \mbox{ if } \ectx(x) = -
				\end{array}
				$
				\\
				\hline
				$\min(T_{s,1})$ & $\min(T_{s,2})$ & $=$ or $\ge$ &
				$
				\begin{array}{l}
				(tr^-_\ectx(T_{s,1}) =  tr^-_\ectx(T_{s,2}) = \emptyset   \wedge \min(tr^+_\ectx(T_{s,1}))\ \op\ \min(tr^+_\ectx(T_{s,2}))+ c) \vee\\
				(tr^-_\ectx(T_{s,1}) =  \emptyset \wedge tr^-_\ectx(T_{s,2}) \neq  \emptyset \wedge \min(tr^+_\ectx(T_{s,1})) + \max(tr^-_\ectx(T_{s,2})) \ \op\ c) \vee\\
				(tr^-_\ectx(T_{s,1}) \neq  \emptyset \wedge tr^-_\ectx(T_{s,2}) \neq  \emptyset \wedge \max(tr^-_\ectx(T_{s,2}))\ \op\  \max(tr^-_\ectx(T_{s,1})) + c)
				\end{array}
				$\\
				\hline
				$\min(T_{s,1})$ & $\min(T_{s,2})$ & $\le$ &
				$
				\begin{array}{l}
				(tr^-_\ectx(T_{s,1}) =  tr^-_\ectx(T_{s,2}) = \emptyset   \wedge \min(tr^+_\ectx(T_{s,1})) \le \min(tr^+_\ectx(T_{s,2}))+ c) \vee\\
				(tr^-_\ectx(T_{s,1}) =  \emptyset \wedge tr^-_\ectx(T_{s,2}) \neq  \emptyset \wedge \min(tr^+_\ectx(T_{s,1})) + \max(tr^-_\ectx(T_{s,2})) \le c) \vee\\
				(tr^-_\ectx(T_{s,1})  \neq \emptyset \wedge tr^-_\ectx(T_{s,2}) =  \emptyset \wedge tr^+_\ectx(T_{s,2}) \neq  \emptyset) \vee\\
				(tr^-_\ectx(T_{s,1}) \neq  \emptyset \wedge tr^-_\ectx(T_{s,2}) \neq  \emptyset \wedge\max(tr^-_\ectx(T_{s,2})) \le  \max(tr^-_\ectx(T_{s,1})) + c)
				\end{array}
				$\\
				\hline
				$\min(T_{s,1})$ & $ \max(T_{s,2})$ & $=$ or $\ge$ &
				$
				\begin{array}{l}
				(tr^-_\ectx(T_{s,1}) =  \emptyset \wedge tr^+_\ectx(T_{s,1}) \neq  \emptyset  \wedge \min(tr^+_\ectx(T_{s,1})) \op \max(tr^+_\ectx(T_{s,2}))+ c) \vee\\
				(tr^-_\ectx(T_{s,1}) =  tr^+_\ectx(T_{s,2}) = \emptyset \wedge \min(tr^+_\ectx(T_{s,1})) + \min(tr^-_\ectx(T_{s,2})) \op c) \vee\\
				(tr^-_\ectx(T_{s,1}) \neq  \emptyset \wedge tr^+_\ectx(T_{s,2}) =  \emptyset \wedge \min(tr^-_\ectx(T_{s,2})) \op \max(tr^-_\ectx(T_{s,1})) + c)
				\end{array}
				$
				\\
				\hline
				$\min(T_{s,1})$ & $ \max(T_{s,2})$ & $\le$ &
				$
				\begin{array}{l}
				(tr^-_\ectx(T_{s,1}) =  \emptyset \wedge tr^+_\ectx(T_{s,2}) \neq  \emptyset  \wedge \min(tr^+_\ectx(T_{s,1})) \le \max(tr^+_\ectx(T_{s,2}))+ c) \vee\\
				(tr^-_\ectx(T_{s,1}) =  tr^+_\ectx(T_{s,2}) = \emptyset \wedge \min(tr^+_\ectx(T_{s,1})) + \min(tr^-_\ectx(T_{s,2})) \le c) \vee\\
				(tr^-_\ectx(T_{s,1}) \neq \emptyset \wedge tr^+_\ectx(T_{s,2}) \neq \emptyset) \vee\\
				(tr^-_\ectx(T_{s,1}) \neq  \emptyset \wedge tr^+_\ectx(T_{s,2}) =  \emptyset \wedge \min(tr^-_\ectx(T_{s,2})) \le \max(tr^-_\ectx(T_{s,1})) + c)
				\end{array}
				$
				\\
				\hline
				$\max(T_{s,1})$ & $ \min(T_{s,2})$ & $=$ or $\ge$ &
				$
				\begin{array}{l}
				(tr^+_\ectx(T_{s,1}) \neq  \emptyset \wedge tr^-_\ectx(T_{s,2}) =  \emptyset  \wedge \max(tr^+_\ectx(T_{s,1})) \op \min(tr^+_\ectx(T_{s,2}))+ c) \vee\\
				(tr^+_\ectx(T_{s,1}) \neq  \emptyset \wedge tr^-_\ectx(T_{s,2}) \neq  \emptyset  \wedge \max(tr^+_\ectx(T_{s,1})) + \max(tr^-_\ectx(T_{s,2})) \op c) \vee\\
				(tr^+_\ectx(T_{s,1}) =  \emptyset \wedge tr^-_\ectx(T_{s,2}) \neq \emptyset \wedge \max(tr^-_\ectx(T_{s,2})) \op  \min(tr^-_\ectx(T_{s,1})) + c)
				\end{array}
				$
				\\
				\hline
				$\max(T_{s,1})$ & $ \min(T_{s,2})$ & $\le$ &
				$
				\begin{array}{l}
				(tr^+_\ectx(T_{s,1}) \neq  \emptyset \wedge tr^-_\ectx(T_{s,2}) =  \emptyset  \wedge \max(tr^+_\ectx(T_{s,1})) \le \min(tr^+_\ectx(T_{s,2}))+ c) \vee\\
				(tr^+_\ectx(T_{s,1}) \neq  \emptyset \wedge tr^-_\ectx(T_{s,2}) \neq  \emptyset  \wedge \max(tr^+_\ectx(T_{s,1})) + \max(tr^-_\ectx(T_{s,2})) \le c) \vee\\
				(tr^+_\ectx(T_{s,1}) =  \emptyset \wedge tr^-_\ectx(T_{s,2}) =  \emptyset \wedge tr^-_\ectx(T_{s,1}) \neq  \emptyset \wedge tr^+_\ectx(T_{s,2}) \neq  \emptyset) \vee\\
				(tr^+_\ectx(T_{s,1}) =  \emptyset \wedge tr^-_\ectx(T_{s,2}) \neq \emptyset \wedge \max(tr^-_\ectx(T_{s,2})) \le  \min(tr^-_\ectx(T_{s,1})) + c)
				\end{array}
				$
				\\
				\hline
				$\max(T_{s,1})$ & $ \max(T_{s,2})$ & $=$ or $\ge$ &
				$
				\begin{array}{l}
				(tr^+_\ectx(T_{s,1}) \neq  \emptyset \wedge tr^+_\ectx(T_{s,2}) \neq  \emptyset  \wedge \max(tr^+_\ectx(T_{s,1})) \op \max(tr^+_\ectx(T_{s,2}))+ c) \vee\\
				(tr^+_\ectx(T_{s,1}) \neq  \emptyset \wedge tr^+_\ectx(T_{s,2}) =  \emptyset  \wedge \max(tr^+_\ectx(T_{s,1})) + \min(tr^-_\ectx(T_{s,2})) \op c) \vee\\
				(tr^+_\ectx(T_{s,1}) =  \emptyset \wedge tr^+_\ectx(T_{s,2}) = \emptyset \wedge \min(tr^-_\ectx(T_{s,2})) \op  \min(tr^-_\ectx(T_{s,1})) + c)
				\end{array}
				$
				\\
				\hline
				$\max(T_{s,1})$ & $ \max(T_{s,2})$ & $\le$ &
				$
				\begin{array}{l}
				(tr^+_\ectx(T_{s,1}) \neq  \emptyset \wedge tr^+_\ectx(T_{s,2}) \neq  \emptyset  \wedge \max(tr^+_\ectx(T_{s,1})) \le \max(tr^+_\ectx(T_{s,2}))+ c) \vee\\
				(tr^+_\ectx(T_{s,1}) \neq  \emptyset \wedge tr^+_\ectx(T_{s,2}) =  \emptyset  \wedge \max(tr^+_\ectx(T_{s,1})) + \min(tr^-_\ectx(T_{s,2})) \le c) \vee\\
				(tr^+_\ectx(T_{s,1}) =  \emptyset \wedge tr^-_\ectx(T_{s,1}) \neq  \emptyset \wedge tr^+_\ectx(T_{s,2}) \neq \emptyset) \vee\\
				(tr^+_\ectx(T_{s,1}) =  \emptyset \wedge tr^+_\ectx(T_{s,2}) = \emptyset \wedge \min(tr^-_\ectx(T_{s,2})) \le  \min(tr^-_\ectx(T_{s,1})) + c)
				\end{array}
				$
				\\
				\hline
			\end{tabular}
		}
	\end{center}
	\caption{$tr_\ectx(T_{i,1}\ \op\ T_{i,2} + c)$ for the situation $T_{i,2}$ is not a constant and $c \ge 0$}
	\label{tab-binary-ge-zero}
	\vspace{-5mm}
\end{table}
}

\hide{
Next, we show how the construction for $\qgdbs^-$ formulae can be generalised to $\qgdbs$ formulae. 
For this purpose, we need an extended version of $\Delta$-contexts to deal with the free set variables. An \emph{extended} $\Delta$-context is defined as $\ectx: \vars_{\sf FO}(\Phi) \rightarrow \{+, -\} \uplus \{S_1,\cdots, S_l\} \rightarrow \{+, -, \pm\}$.

Let $\ectx$ be an extended $\Delta$-context. We construct a $(\qgdbs)_\natnum$ formula inductively as follows: 
\begin{itemize}
\item For each subformula $\Psi$ of $\Delta$ such that $\Psi$ is a $\qgdbs^-$ formula, we set $tr_\ectx(\Psi)$  as $tr_{\ectx |_{\vars_{\sf FO}(\Delta)}}(\Psi)$, where $\ectx |_{\vars_{\sf FO}(\Delta)}$ is the  $\Delta$-context obtained by restricting $\ectx$ to $\vars_{\sf FO}(\Delta)$. 

\item For each atomic formula of the form $T_{i,1}\ \op\ T_{i,2}$ in $\Delta$ such that $T_{i,1}$ and $T_{i,2}$ are the integer terms defined in $\ps$ and they contain only variables from $\{x_1,\cdots, x_k, S_1, \cdots, S_l\}$,  we construct $tr_\ectx(T_{i,1}\ \op\ T_{i,2})$ as follows:  
\begin{itemize}
\item For each $j \in [k]$, replace $x_j$ in $T_{i,1}\op\ T_{i,2}$ with $x^+_j$ if $\ectx(x_j) = +$, and replace $x_j$ with $- x^-_j$ otherwise.
\item For each variable $j \in [l]$, 
\begin{itemize}
\item if $\ectx(S_j) = \bot$ and $\max(S_j)$ or $\min(S_j)$ occurs in $T_{i,1}\op\ T_{i,2}$, then replace $T_{i,1}\op\ T_{i,2}$ with $\lfalse$,
\item if $\ectx(S_j) = +$, then replace $\max(S_j)$ and $\min(S_j)$ in $T_{i,1}\op\ T_{i,2}$ with $\max(S^+_j)$ and $\min(S^+_j)$ respectively,
\item if $\ectx(S_j) = -$, then replace $\max(S_j)$ and $\min(S_j)$ in $T_{i,1}\op\ T_{i,2}$ with $-\min(S^-_j)$ and $-\max(S^-_j)$ respectively,
\item if $\ectx(S_j) = \pm$, then replace $\max(S_j)$ and $\min(S_j)$ in $T_{i,1}\op\ T_{i,2}$ with $\max(S^+_j)$ and $-\max(S^-_j)$ respectively.
\end{itemize}  
\end{itemize}

\item For the other syntactic rules  in $\ps$, we construct $tr_\ectx(\Delta)$ by applying the syntactic induction, similarly to the construction of $tr_\ectx(\Psi)$ above. For instance, $tr_\ectx(\Delta_1 \wedge \Delta_2)= tr_\ectx(\Delta_1) \wedge tr_\ectx(\Delta_2)$ and $tr_\ectx(\forall x.\ \Delta_1)= \forall x^+.\ tr_{\ectx[x \rightarrow +]}(\Delta_1) \wedge \forall x^-.\ tr_{\ectx[x \rightarrow -]}(\Delta_1)$.
\end{itemize}
Finally, let $\Delta' = \bigvee \limits_{\ectx} (tr_{\ectx}(\Delta) \wedge \bigwedge \limits_{\ectx(S)=\bot} (S^+=\emptyset \wedge S^- = \emptyset)  \wedge \bigwedge \limits_{\ectx(S)=+}(S^+\neq \emptyset \wedge S^-=\emptyset) \wedge \bigwedge \limits_{\ectx(S)=-} (S^+=\emptyset \wedge S^- \neq \emptyset) \wedge \bigwedge \limits_{\ectx(S)=\pm} (S^+\neq \emptyset \wedge S^- \neq \emptyset) \wedge \bigwedge \limits_{\ectx(x)=+} (x^-=0) \wedge \bigwedge \limits_{\ectx(x)=-} (x^+=0 \wedge x^->0))$. 
}

\subsection{Details of Step III}

We start with some additional notations. 
First observe that there is a one-to-one correspondence between models of  $\Phi(\vec{x}, \vec{S})$ and finite words over  $2^{AP}$ with $AP =\{x_1, \cdots, x_k, S_1,\cdots, S_l\}$ satisfying that $x_j$ occurs in \emph{exactly one} position for each $j \in [k]$. A finite word $w = w_0 \cdots w_{n-1}$ over $2^{AP}$ is a finite sequence such that $w_i \in 2^{AP}$ for each $i\in \{0\} \cup [n-1]$. On the one hand, any model $(n_1, \cdots, n_k, A_1,\cdots, A_l) \in \natnum^k \times \natset^l$ of $\Phi(\vec{x}, \vec{S})$ can be interpreted as a finite word $w$ as follows: If $k = 0$ and $A_i =  \emptyset$ for all $i \in [l]$, then $w = \varepsilon$; otherwise let $|w|=1+\max(\{n_1,\cdots,n_k\} \cup \bigcup \limits_{i \in [l]} A_i)$, and, for each position $i \in \{0\} \cup [|w|-1]$, $w_i = P \subseteq AP$ iff $P = \{x_j \mid j \in [k], i = n_j\} \cup \{S_j \mid j \in [l], i \in A_j\}$. On the other hand, for a word $w \in (2^{AP})^*$ where $x_j$ occurs in exactly one position for each $j \in [k]$, a tuple $(n_1,\cdots, n_k, A_1,\cdots, A_l) \in \natnum^k \times \natset^l$ can be constructed such that for each $j \in [k]$, $n_j  = i$ iff $x_j \in w_i$, and for each $j \in [l]$, $A_j = \{i \in \{0\} \cup [|w|-1] \mid S_j \in w_i \}$. 
By slightly abusing the notation, we also use $\Ll(\Phi(\vec{x}, \vec{S}))$ to denote the set of words $w \in (2^{AP})^*$ such that $w \models \Phi$.

\vspace{-1mm}
\begin{definition}[Presburger automata]
	A Presburger automaton (PA) $\Aa$ is a tuple $(Q, \Sigma, \delta, q_0, F, \Psi)$, where $(Q, \Sigma, \delta, q_0, F)$ is an NFA with $Q=\{q_0, q_1, \ldots, q_m\}$, and $\Psi(x_{q_0}, \cdots, x_{q_m})$ is a \emph{quantifier-free Presburger arithmetic} formula over the set of variables $\{x_{q_i} \mid i \in \{0\} \cup [m]\}$. 
\end{definition} 
\vspace{-1mm}

A word $w = w_0 \cdots w_{n-1} \in (2^{AP})^*$ is accepted by $\Aa$ if there is a run $R = q_0 \xrightarrow{w_0} q_1 \cdots q_{n-1} \xrightarrow{w_{n-1}} q_n$ such that $q_n \in F$ and 
$\Psi(|R|_{q_0}/x_{q_0}, \cdots, |R|_{q_m}/x_{q_m})$ holds, 
where 
the vector $(|R|_q)_{q \in Q}$ is the Parikh image of the sequence $q_0, \cdots, q_n$, that is,  $|R|_q$ is the number of occurrences of $q$ in $R$. 
We use $\Ll(\Aa)$ to denote the set of words accepted by $\Aa$.

\vspace{-2mm}
\begin{theorem}[\cite{SSM08}]
	Nonemptiness of Presburger automata is decidable. 
\end{theorem}
\vspace{-1mm}

Given $\Phi_{\mathrm{core}}\wedge \Phi_{\mathrm{count}}$, where $\Phi_{\mathrm{core}}$ is an $\qgdbs^-_\natnum$ formula and $\Phi_{\mathrm{count}}$ is a conjunction of the formulae of the form $T_m\ \op\ 0$ which contains only variables from $\vec{x} \cup \vec{S}$, our aim is to construct a PA to accept models---as words---of $\Phi_{\mathrm{core}}\wedge \Phi_{\mathrm{count}}$. To this end, we first show how to construct an NFA from $\Phi_{\mathrm{core}}$, an $\qgdbs^-_\natnum$ formula. 

It is a simple observation that an $\qgdbs^-_\natnum$ formula can be rewritten in exponential time into a formula in 
MSOW defined by the following rules,

\smallskip
\hspace{1cm}$
\Phi ::= x + 1 = y \mid x < y \mid S(x) \mid \Phi \wedge \Phi \mid \neg \Phi \mid \forall x.\ \Phi \mid \forall S.\ \Phi,
$
\smallskip

\noindent where $x,y$ are variables ranging over $\natnum$, and $S$ is a (second-order) set variable ranging over the set of \emph{finite} subsets of $\natnum$. 
Note that the exponential blow-up is because, only the successor operator is available in MSOW 
while constants $c$ are encoded in binary. 
For instance, $x_1 \le x_2 + 2$ has to be rewritten into $\exists z, z'.\ z = x_2 + 1 \wedge z' = z +1 \wedge x_1 \le z'$.


We  can  then invoke the celebrated B\"{u}chi-Elgot theorem:

\vspace{-1mm}
\begin{theorem}[\cite{Bu60,Elg61}]\label{thm-msow-nfa}
	Let $\Phi(S_1,\cdots, S_k)$ be an MSOW formula. Then an NFA $\Aa_\Phi$ over $2^{\{S_1,\cdots,S_k\}}$ can be constructed so that $\Ll(\Aa_\Phi) = \Ll(\Phi(S_1,\cdots, S_k))$. 
\end{theorem}
\vspace{-1mm}

It follows from Theorem~\ref{thm-msow-nfa} that an NFA $\Aa_\Phi=(Q, AP, \delta, q_0, F)$ can be constructed from a $\qgdbs^-_\natnum$ formula $\Phi(\vec{x}, \vec{S})$ such that $\Ll(\Phi)=\Ll(\Aa_\Phi)$.
As the next step we construct a \emph{quantifier-free Presburger arithmetic} formula $\Psi$ for the sought PA out of $\Phi_{\mathrm{count}}$. We first construct, for each $x_i$, an NFA $\Aa_i$ illustrated in Fig.\ref{fig-auto-b}(a), and for each $S_j$, an NFA $\Bb_j$ illustrated in Fig.\ref{fig-auto-b}(b). We then consider an NFA  $\Aa^\times_\Phi$ which is the product of $\Aa_{\Phi}$ and all $\Aa_i$ for $i \in [k]$ and $\Bb_j$ for $j \in [l]$. Note that each state of  $\Aa^\times_\Phi$ is a vector of states $\vec{q}=(q, q_1, \cdots, q_k, q_{k+1},\cdots, q_{k+l})$ such that $q \in Q$, $q_i  \in \{p_{0,i}, p_{1,i}\}$ for each $i \in [k]$, and $q_{k+j} \in \{q_{0, j}, q_{1, j}, q_{2, j}\}$ for each $j \in [l]$. We write $\vec{q}_r$ for the $r$-th entry of $\vec{q}$, i.e., $\vec{q}_0=q$ and $\vec{q}_i=q_i$ for each $i \in [k+l]$.


We observe that, for each $i \in [k]$, $x_i$ is expressed by $\sum_{\{\vec{q} \mid \vec{q}_{i}=p_{0,i}\}} x_{\vec{q}}-1$, and for each $j \in [l]$, $\min(S_j)$ is expressed by $\sum_{\{\vec{q} \mid \vec{q}_{k+j}=q_{0,j}\}} x_{\vec{q}}-1$ and $\max(S_j)$ is expressed by $\sum_{\{\vec{q} \mid \vec{q}_{k+j}=q_{0,j}\}} x_{\vec{q}}  + \sum_{\{\vec{q}\mid\vec{q}_{k+j}=q_{1,j}\}} x_{\vec{q}}-1$. We then substitute them into $\Phi_{\mathrm{count}}$ and obtain $\Psi$ which is over the variables $x_{\vec{q}}$. 
%

\begin{figure}[htbp]
	\vspace{-4mm}
	\begin{center}
		\includegraphics[scale=0.67]{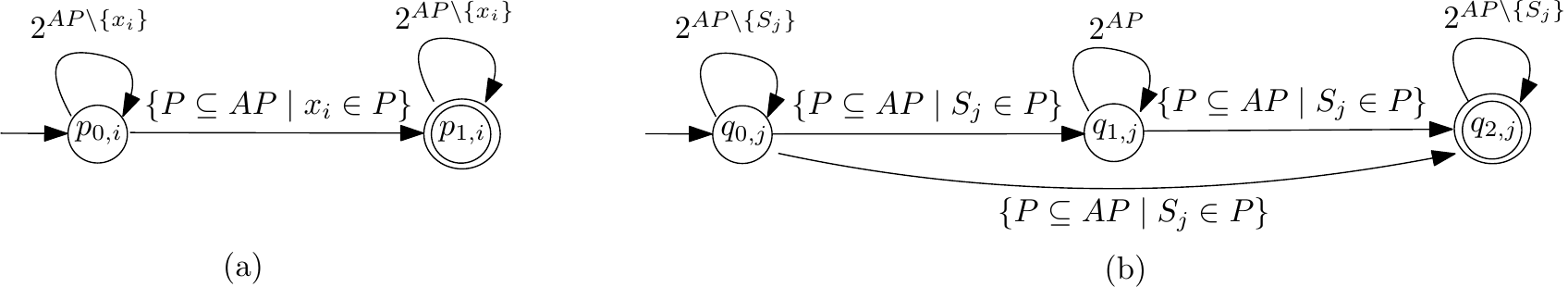}
		\vspace{-2mm}
		\caption{NFA $\Aa_i$ for $x_i$ and $\Bb_j$ for $S_j$}
		\label{fig-auto-b}
	\end{center}
	\vspace{-1cm}
\end{figure}
%

\hide{
	\begin{example} 
		Let us illustrate the idea by considering an example that $\Phi_{\sf core}$ contains two set variables $S_1,S_2$ and $\Phi_m = (\max(S_1) - \min(S_1)) - (\max(S_2) - \min(S_2)) = 0$. 
		
		
		
		%
		%
		%
		%

		%
		Then $\min(S_i)$ is expressed by $x_{q_{0,i}} + 1$ and $\max(S_i)$ is expressed by $x_{q_{0,i}} + x_{q_{1,i}}+1$.
		
		Therefore, the Presburger automaton $\Aa_{\Phi}$ is constructed as $(\Aa', \Psi')$, where 
		\begin{itemize}
			\item $\Aa' = (Q', AP, \delta', q'_0, F')$ is the product of $\Aa_{\Phi_{\sf core}}$, $\Bb_1$ and $\Bb_2$, 
			\item Then $\Phi'$ is from $\Phi_m$ by replacing 
			\begin{itemize}
				\item $\min(S_1)$ with $\sum \limits_{(q, q_{0,1}, q_{j,2}) \in Q'} x_{(q,q_{0,1}, q_{j,2})}+ 1$, 
				\item $\max(S_1)$ with  $\sum \limits_{(q, q_{0,1}, q_{j,2}) \in Q'} x_{(q,q_{0,1}, q_{j,2})}+ \sum \limits_{(q, q_{1,1}, q_{j,2}) \in Q'} x_{(q,q_{1,1}, q_{j,2})}+ 1$, 
				\item $\min(S_2)$ with $\sum \limits_{(q, q_{j,1}, q_{0,2}) \in Q'} x_{(q, q_{j,1}, q_{0,2})}+ 1$, 
				\item $\max(S_2)$ with $\sum \limits_{(q, q_{j,1}, q_{0,2}) \in Q'} x_{(q, q_{j,1}, q_{0,2})} + \sum \limits_{(q, q_{j,1}, q_{1,2}) \in Q'} x_{(q, q_{j,1}, q_{1,2})}+1$.
			\end{itemize}
			Therefore, we have
			\[
			\Psi' = \sum \limits_{(q, q_{1,1}, q_{j,2}) \in Q'} x_{(q,q_{1,1}, q_{j,2})} - \sum \limits_{(q, q_{j,1}, q_{1,2}) \in Q'} x_{(q, q_{j,1}, q_{1,2})} = 0.
			\]
		\end{itemize}
		
		It is easy to figure out the construction for the general case from the illustration above.
	\end{example}
}

\begin{proposition}
	For an $\qgdbs_\natnum$ formula $\Phi = \Phi_{\mathrm{core}} \wedge \Phi_{\mathrm{count}}$, $\Phi_{\mathrm{core}}$ is an $\qgdbs^-_\natnum$ formula, and $\Phi_{\mathrm{count}}$ is a conjunction of formulae of the form $T_m\ \op\ 0$ which contain only variables from $\vec{x} \cup \vec{S}$, a PA $\Aa_\Phi=(\Aa^\times_{\Phi}, \Psi)$ can be constructed effectively such that $\Ll(\Aa_\Phi) = \Ll(\Phi)$.
\end{proposition}

}{}

\end{document}